\documentclass[10pt]{article}

\usepackage{authblk}

\usepackage[margin=1in,letterpaper]{geometry}

\usepackage{amsmath,amsthm,amsfonts,dsfont}
\usepackage{amssymb,latexsym,graphicx}
\usepackage{mathtools}
\usepackage{hyperref}
\usepackage{graphicx}
\usepackage{caption}
\usepackage{subcaption}

\usepackage{algorithm}
\usepackage{algpseudocode}
\usepackage{mathrsfs}
\usepackage{balance}
\newtheorem{theorem}{Theorem}
\newtheorem{proposition}[theorem]{Proposition}
\newtheorem{claim}[theorem]{Claim}
\newtheorem{lemma}[theorem]{Lemma}
\newtheorem{corollary}[theorem]{Corollary}
\newtheorem{definition}[theorem]{Definition}

\newcommand{\B}{\ensuremath{\mathbb{B}}}

\newcommand{\N}{\ensuremath{\mathbb{N}}}

\newcommand{\R}{\ensuremath{\mathbb{R}}}

 \newcommand{\eps}{\varepsilon}
\renewcommand{\epsilon}{\varepsilon}

\newcommand{\eqdef}{\mathbin{\stackrel{\rm def}{=}}}

\renewcommand{\vec}[1]{\ensuremath{\mathbf{#1}}}

\newcommand{\poly}{\mathrm{poly}}

\DeclareMathOperator*{\expect}{\mathbb{E}}

\newcommand{\sign}{\text{sign}}

\newcommand{\set}[1]{\left\{#1\right\}}

\newcommand{\E}{\expect}

\newcommand{\T}{\mathsf{T}}
\newcommand{\pr}[2]{\langle{#1, #2}\rangle}
\makeatletter
\def\imod#1{\allowbreak\mkern4mu({\operator@font mod}\,\,#1)}
\makeatother
\newcommand{\transpose}{\intercal}

\newcommand{\cut}[1]{}

\title{Towards a Constructive Version of Banaszczyk's Vector Balancing Theorem}

\author[1]{Daniel Dadush\thanks{Supported by the NWO Veni grant 639.071.510.}}
\author[2]{Shashwat Garg\thanks{Supported by the Netherlands Organisation for
Scientific Research (NWO) under project no. 022.005.025.}}
\author[3]{Shachar Lovett\thanks{Supported by an NSF CAREER award 1350481 and a Sloan fellowship}}
\author[4]{Aleksandar Nikolov\thanks{Supported by an NSERC Discovery Grant}}
\affil[1]{{\small Centrum Wiskunde \& Informatica, Amsterdam, \texttt{dadush@cwi.nl}}}
\affil[2]{\small Department of Mathematics and Computer Science, Eindhoven University of Technology, \texttt{s.garg@tue.nl}}
\affil[3]{\small Department of Computer Science and Engineering, University of California, San Diego, \texttt{slovett@cs.ucsd.edu}}
\affil[4]{\small Department of Computer Science, University of Toronto, \texttt{anikolov@cs.toronto.edu}}

\date{}
\begin{document}

\maketitle

\begin{abstract}
An important theorem of Banaszczyk (Random Structures \& Algorithms `98) states that for any sequence of vectors of $\ell_2$ norm at most $1/5$ and any convex body $K$ of Gaussian measure $1/2$ in $\R^n$, there exists a signed combination of these vectors which lands inside $K$. A major open problem is to devise a constructive version of Banaszczyk's vector balancing theorem, i.e.~to find an efficient algorithm which constructs the signed combination.

We make progress towards this goal along several fronts. As our first contribution, we show an equivalence between Banaszczyk's theorem and the existence of $O(1)$-subgaussian distributions over signed combinations. For the case of symmetric convex bodies, our equivalence implies the existence of a \emph{universal} signing algorithm (i.e.~independent of the body), which simply samples from the subgaussian sign distribution and checks to see if the associated combination lands inside the body. For asymmetric convex bodies, we provide a novel \emph{recentering procedure}, which allows us to reduce to the case where the body is symmetric. 

As our second main contribution, we show that the above framework can be efficiently implemented when the vectors have length $O(1/\sqrt{\log n})$, recovering Banaszczyk's results under this stronger assumption. More precisely, we use random walk techniques to produce the required $O(1)$-subgaussian signing distributions when the vectors have length $O(1/\sqrt{\log n})$, and use a stochastic gradient ascent method to implement the recentering procedure for asymmetric bodies.
\end{abstract}

\section{Introduction}

Given a family of sets $S_1,\dots,S_m$ over a universe $U = [n]$, the goal of
combinatorial discrepancy minimization is to find a bi-coloring $\chi: U
\rightarrow \{-1,1\}$ such that the discrepancy, i.e.~the maximum imbalance
$\max_{j \in [m]} |\sum_{i \in S_j} \chi(i)|$, is made as small as possible. 
Discrepancy theory, where discrepancy minimization plays a major role, has a
rich history of applications in computer science as well as mathematics, and we
refer the reader to~\cite{Matousek99, Chazelle00,Panorama} for a general
exposition.

A beautiful question regards the discrepancy of sparse set systems,
i.e.~set systems in which
each element appears in at most $t$ sets. A classical theorem of Beck and
Fiala~\cite{BeckFiala81} gives an upper bound of $2t-1$ in this setting. They
also conjectured an $O(\sqrt{t})$ bound, which if true would be tight. An
improved Beck-Fiala bound of $2t-\log^*t$ was given by
Bukh~\cite{Bukh13}, where $\log^*t$ is the iterated logarithm function
in base $2$.
Recently, it was shown by Ezra and Lovett~\cite{EzraLovett15} that a bound of
$O(\sqrt{t \log t})$ holds with high probability when $m \geq n$ and each
element is assigned to $t$ sets uniformly at random. The best general bounds
having sublinear dependence in $t$ currently depend on $n$ or $m$.
Srinivasan~\cite{Srinivasan97} used Beck's \emph{partial coloring
method}~\cite{Beck81} to give a bound of $O(\sqrt{t} \log \min \set{n,m})$.
Using techniques from convex geometry,  Banaszczyk~\cite{Bana98} proved a
general result on vector balancing (stated below) which implies an
$O(\sqrt{t\log \min \set{n,m}})$ bound.

The proofs of both Srinivasan's and Banaszczyk's bounds were
non-constructive, that is, they provided no efficient algorithm to construct the guaranteed colorings, short of
exhaustive enumeration. In the last 6 years, tremendous progress has been made
on the question of matching classical discrepancy bounds algorithmically.
Currently, essentially all discrepancy bounds proved using the partial coloring
method, including Srinivasan's, have been made
constructive~\cite{Bansal10,LovettMeka12,HSS14,Rothvoss14,EldanSingh14}. Constructive
versions of Banaszczyk's result have, however, proven elusive until very recently.
In recent work~\cite{BDG16}, the first and second named authors jointly with
Bansal gave a constructive algorithm for recovering Banaszczyk's bound in the
Beck-Fiala setting as well as the more general Koml{\'o}s setting. An alternate algorithm via multiplicative weight updates was also given recently in \cite{LRT16}. However,
finding a constructive version of Banaszczyk's more general vector balancing
theorem, which has further applications in approximating hereditary discrepancy, remains an open problem. This theorem is stated as follows:   

\begin{theorem}[Banaszczyk~\cite{Bana98}] 
\label{thm:bana-disc}
Let $v_1,\dots,v_n \in \R^m$ satisfy $\|v_i\|_2 \leq 1/5$.  Then for any
convex body $K \subseteq \R^m$ of Gaussian measure at least $1/2$, there
exists $\chi \in \set{-1,1}^n$ such that $\sum_{i=1}^n \chi_i v_i \in K$.
\end{theorem}
\cut{To apply this theorem to derive the $O(\sqrt{t \log n})$ bound for Beck-Fiala,
first take the incidence matrix $A$ of the system,~i.e. $A_{ji} = 1$ if element
$i \in [n]$ is in set $j \in [m]$ and $0$ otherwise. Noting that the $O(\sqrt{t
\log n})$ bound is trivial for $t \geq n^2$ (discrepancy can't exceed $n$), we
can assume that the number of rows of $m$ of $A$ satisfies $m \leq t n \leq
n^3$. From here, we simply take $v_1,\dots,v_n$ to be the columns of
$A/(5\sqrt{t})$ (note their $\ell_2$ norms are at most $1/5$) and take $K =
\set{x \in \R^m: \|x\|_\infty \leq O(\sqrt{\log n})}$ to be the $m$-dimensional
cube of Gaussian measure $1/2$. The desired discrepancy bound for Banaszczyk's
coloring follows by multiplying through by $5\sqrt{t}$.\dnote{Should shorten the
above. Was mostly to convince myself actually.}}

The lower bound $1/2$ on the Gaussian measure of $K$ is easily seen to be tight.
In particular, if all the vectors are equal to $0$, we must have that $0
\in K$. If we allow Gaussian measure $<1/2$, then $K = \set{x \in \R^n: x_1 \geq
\eps}$, for $\eps > 0$ small enough, is a clear counterexample. On the other
hand, it is not hard to see that if $K$ has Gaussian measure $1/2$ then $0 \in
K$. Otherwise, there exists a halfspace $H$ containing $K$ but not $0$, where $H$
clearly has Gaussian measure less than $1/2$.

Banaszczyk's theorem gives the best known bound for the notorious K{o}ml\'os
conjecture~\cite{Spencer87ten}, a generalization of the Beck-Fiala conjecture,
which states that for any sequence of vectors $v_1,\dots,v_n \in \R^m$ of
$\ell_2$ norm at most $1$, there exists $\chi \in \set{-1,1}^n$ such that
$\|\sum_{i=1}^n \chi_i v_i\|_\infty$ is a constant independent of $m$ and $n$.
In this context, Banaszczyk's theorem gives a bound of $O(\sqrt{\log m})$,
because an $O(\sqrt{\log m})$ scaling of the unit ball of $\ell_\infty^m$ has
Gaussian measure $1/2$. Banaszczyk's theorem together with estimates on the
Gaussian measure of slices of the $\ell^m_\infty$ ball due to Barthe, Guedon,
Mendelson, and Naor~\cite{BartheGMN05} give a bound of $O(\sqrt{\log d})$, where
$d\le \min\{m,n\}$ is the dimension of the span of $v_1, \ldots, v_n$.  A
well-known reduction (see e.g.~Lecture 9 in~\cite{Spencer87ten}), shows that
this bound for the Koml\'os problem implies an $O(\sqrt{t \log \min\{m,n\}})$ bound in the
Beck-Fiala setting. 

While the above results only deal with the case of $K$ being a cube,
Banaszczyk's theorem has also been applied to other cases. It was used
in~\cite{Bana12} to give the best known bound on the Steinitz
conjecture. In this problem, the input is a set of vectors
$v_1,\dots,v_n$ in $\mathbb{R}^m$ of norm at most one and summing to
0. The aim is to find a permutation $\pi:[n]\rightarrow [n]$ to
minimise the maximum sum prefix of the vectors rearranged according to
$\pi$ i.e. to minimize $\max_{k\in[n]}\|\sum_{i=1}^k
v_{\pi(i)}\|$. The Steinitz conjecture is that this bound should
always be $O(\sqrt{m})$, irrespective of the number of vectors, and
using the vector balancing theorem Banaszczyk proved a bound of $O(\sqrt{m}+\sqrt{\log n})$ for the $\ell_2$ norm.

More recently, Banaszczyk's theorem was applied to more general symmetric polytopes
in Nikolov and Talwar's approximation algorithm~\cite{NT15} for a hereditary notion of
discrepancy. Hereditary
discrepancy is defined as the maximum discrepancy of any restriction
of the set system to a subset of the universe. In~\cite{NT15} it was shown that
an effan efficiently computable quantity, denoted $\gamma_2$,
bounds hereditary discrepancy from above and from below for any given set
system, up to polylogarithmic factors. For the upper bound they used
Banaszczyk's theorem for a natural polytope associated with the set
system. However, since there is no known algorithmic version of
Banaszczyk's theorem for a general body, it is not known how to
efficiently compute colorings that achieve the discrepancy upper bounds in terms
of $\gamma_2$. The recent work on algorithmic bounds in the Koml\'os
setting does not address this more general problem.



Banaszczyk's proof of Theorem~\ref{thm:bana-disc} follows an ingenious
induction argument, which folds the effect of choosing the sign of
$v_n$ into the body $K$. The first observation is that finding a point
of the set $\sum_{i = 1}^n{\{-v_i, v_i\}}$ inside $K$ is equivalent to
finding a point of $\sum_{i = 1}^{n-1}{\{-v_i, v_i\}}$ in $K-v_n \cup
K + v_n$. Inducting on this set is not immediately possible because it
may no longer be convex. Instead, Banaszczyk shows that a convex
subset $K * v_n$ of $(K-v_n)\cup (K+v_n)$ has Gaussian measure at
least that of $K$, as long as $K$ has measure at least $1/2$, which allows him to induct on $K * v_n$. In the
base case, he needs to show that a convex body of Gaussian measure at
least $1/2$ must contain the origin, but this fact follows easily from
the hyperplane separation theorem, as indicated above. While extremely elegant,
Banaszczyk's proof can be seen as relatively mysterious, as it does not
seem to provide any tangible insights as to what the colorings look
like. 
\cut{Moreover, even if we start with a very simple set $K$, the sets
$K * v_n$, $K * v_n * v_{n-1}$, etc., quickly get more complicated and
it is not clear how to maniputale them algorithmically.}

\subsection{Our Results} 
\label{sec:results}

As our main contribution, we help demystify Banaszczyk's theorem, by showing
that it is equivalent, up to a constant factor in the length of the vectors, to
the existence of certain subgaussian coloring distributions. Using this
equivalence, as our second main contribution, we give an efficient algorithm
that recovers Banaszczyk's theorem up to a $O(\sqrt{\log \min \set{m,n}})$
factor for all convex bodies. This improves upon the best previous algorithms of
Rothvoss~\cite{Rothvoss14}, Eldan and Singh~\cite{EldanSingh14}, which only
recover the theorem for symmetric convex bodies up to a $O(\log \min
\set{m,n})$ factor.

As a major consequence of our equivalence, we show that for any sequence
$v_1,\dots,v_n \in \R^m$ of short enough vectors there exists a probability
distribution $\chi \in \set{-1,1}^n$ over colorings such that, for
\emph{any symmetric convex body} $K \subseteq \R^m$ of
Gaussian measure at least $1/2$, the random variable
$\sum_{i=1}^n \chi_i v_i$ lands inside $K$ with probability at least $1/2$.
Importantly, if such a distribution can be efficiently sampled, we immediately
get a \emph{universal sampler} for constructing Banaszczyk colorings for all
symmetric convex bodies (we remark that the recent work of~\cite{BDG16}
constructs a more restricted form of such distributions). Using random walk
techniques, we show how to implement an approximate version of this sampler
efficiently, which guarantees the same conclusion when the vectors are
of length $O(1/\sqrt{\log \min \set{m,n}})$. We provide more details on these
results in Sections~\ref{sec:intro-subg}~and~\ref{sec:intro-rand-walk}.

To extend our results to asymmetric convex bodies, we develop a novel
\emph{recentering procedure} and a corresponding efficient implementation which
allows us to reduce the asymmetric setting to the symmetric one.  After this
reduction, a slight extension of the aforementioned sampler again yields the
desired colorings.  We note that our recentering procedure in fact depends on
the target convex body, and hence our algorithms are no longer universal in this
setting. We provide more details on these results in
Sections~\ref{sec:intro-asymmetric}~and~\ref{sec:intro-recentering}. 

Interestingly, we additionally show that this procedure can be
extended to yield a completely different coloring algorithm, i.e.~not using the
sampler, achieving the same $O(\sqrt{\log \min \set{m,n}})$ approximation
factor.  Surprisingly, the coloring outputted by this procedure is essentially
deterministic and has a natural analytic description,
which may be of independent interest.


Before we continue with a more detailed description on our results, we begin
with some terminology and a well-known reduction. Given a set of vectors
$v_1,\dots,v_n \in \R^m$, we shall call a property \emph{hereditary} if it holds
for all subsets of the vectors. We note that Banaszczyk's vector balancing bounds
restricted to a set of vectors are hereditary, since a bound on the
maximum $\ell_2$ norm
of the vectors is hereditary. We shall say that a property of colorings holds in
the \emph{linear setting}, if when given any shift $t \in \sum_{i=1}^n
[-v_i,v_i] \eqdef \set{\sum_{i=1}^n \lambda_i v_i: \lambda \in [-1,1]^n}$, one
can find a coloring (or distribution on colorings) $\chi \in \set{-1,1}^n$ such
that $\sum_{i=1}^n \chi_i v_i - t$ satisfies the property. It is well-known that
Banaszczyk's theorem also extends by standard arguments to the linear setting
after reducing the $\ell_2$ norm bound from $1/5$ to $1/10$ (a factor
$2$ drop). This follows, for example, from the general inequality
between hereditary and linear discrepancy proved by Lovasz, Spencer,
and Vesztergombi~\cite{LSV86}.

All the results in this work will in fact hold in the linear
setting. When treating the linear setting, it is well known that one
can always reduce to the case where the vectors $v_1,\dots,v_n$ are
linearly independent, and in our setting, when $m=n$. In particular,
assume we are given some shift $t \in \sum_{i = 1}^n[-v_i,v_i]$ and
that $v_1, \dots, v_n$ are \emph{not} linearly independent. Then,
using a standard linear algebraic technique, we can find a
``fractional coloring'' $x \in [-1, 1]^n$ such that $\sum_{i =
  1}^n{x_i v_i} = t$, and the vectors $(v_i: i \in A_x)$ are linearly
independent, where $A_x \eqdef \{i: x_i \in (-1,1)\}$ is the set of
fractional coordinates (see Lecture 5 in~\cite{Spencer87ten}, or
Chapter 4 in~\cite{Matousek99}\cut{, or Section~\ref{sec:prelims} in this
paper for details}). We can think of this as a reduction to coloring
the linearly independent vectors indexed by $A_x$. Specifically, given
$x$ as above, define the lifting function $L_x: [-1,1]^{A_x}
\rightarrow [-1,1]^n$ by
\begin{equation}\label{eq:lifting}
L_x(z)_i = \begin{cases} z_i &: i \in A_x \\ 
                          x_i &: i \in [n] \setminus A_x \end{cases},~\forall i \in [n] \text{ .}
\end{equation}
This map takes any coloring $\chi \in \set{-1,1}^{A_x}$ and ``lifts''
it to a full coloring $L_x(\chi) \in \{-1,1\}^n$. It also satisfies
the property that $L_x(\chi) - t = \sum_{i \in A_x}{\chi_i v_i} -
\sum_{i \in A_x}{x_i v_i}$. So, if we can find a coloring $\chi \in
\{-1, 1\}^{A_x}$ such that $\sum_{i \in A_x}{\chi_i v_i} - \sum_{i \in
  A_x}{x_i v_i} \in K$, then we would have $L_x(\chi) - t \in K$ as
well. Moreover, if we define $W$ as the span of $(v_i: i \in A_x)$,
then $\sum_{i \in A_x}{\chi_i v_i} - \sum_{i \in A_x}{x_i v_i} \in K$
if and only if $\sum_{i \in A_x}{\chi_i v_i} - \sum_{i \in A_x}{x_i
  v_i} \in K \cap W$, so we can replace $K$ with $K \cap W$, and work
entirely inside $W$. For convex bodies $K$ with Gaussian measure at
least $1/2$, the central section $K \cap W$ has Gaussian measure that
is at least as large, so we have reduced the problem to the case of
$|A_x|$ linearly independent vectors in an $|A_x|$-dimensional
space. (See Section~\ref{sec:prelims} for the full
details.)  We shall thus, for simplicity, state all our results in the
setting where the vectors $v_1,\dots,v_n$ are in $\R^n$ and are
linearly independent.


\subsubsection{Symmetric Convex Bodies and Subgaussian Distributions}
\label{sec:intro-subg}

In this section, we detail the equivalence of Banaszczyk's theorem restricted
to symmetric convex bodies with the existence of certain subgaussian
distributions. We begin with the main theorem of this section, which we note
holds in a more general setting than Banaszczyk's result.
 
\begin{theorem}[Main Equivalence]
\label{thm:main-equiv}
Let $T \subseteq \R^n$ be a finite set. Then, the following parameters are
equivalent up to a universal constant factor independent of $T$ and $n$: 
\begin{enumerate}
\item\label{thm:main-equiv-bal} The minimum $s_b > 0$ such that for any
symmetric convex body $K \subseteq \R^n$ of Gaussian measure at least $1/2$,
we have that $T \cap s_b K \neq \emptyset$. 
\item\label{thm:main-equiv-subg} The minimum $s_g > 0$ such that there exists
an $s_g$-subgaussian random variable $Y$ supported on $T$.
\end{enumerate}
\end{theorem}

We recall that a random vector $Y \in \R^n$ is $s$-subgaussian, or subgaussian
with parameter $s$, if for any unit vector $\theta \in S^{n-1}$ and $t \geq 0$,
$\Pr[|\pr{Y}{\theta}| \geq t] \leq 2e^{-(t/s)^2/2}$. In words, $Y$ is
subgaussian if all its $1$-dimensional marginals satisfy the same tail 
bound as
the $1$-dimensional Gaussian of mean $0$ and standard deviation $s$. 

To apply the above to discrepancy, we set $T = \sum_{i=1}^n \set{-v_i,v_i}$,
i.e.~all signed combinations of the vectors $v_1,\dots,v_n \in \R^n$. In this
context, Banaszczyk's theorem directly implies that $s_b \leq 5 \max_{i \in [n]}
\|v_i\|_2$, and hence by our equivalence that $s_g = O(1) \max_{i \in [n]}
\|v_i\|_2$. Furthermore, the above extends to the linear setting
letting $T = \sum_{i=1}^n \set{-v_i,v_i} - t$, for $t \in \sum_{i=1}^n
[-v_i,v_i]$, because, as mentioned above, Banaszczyk's theorem extends
to this setting as well.

The existence of the \emph{universal sampler} claimed in the previous section is
in fact the proof that $s_b = O(s_g)$ in the above Theorem. In particular, it
follows directly from the following lemma.  

\begin{lemma} 
\label{lem:subg-in-K}
Let $Y \in \R^n$ be an $s$-subgaussian random variable. There exists an absolute
constant $c > 0$, such for any symmetric convex body $K \subseteq \R^n$ of
Gaussian measure at least $1/2$, $\Pr[Y \in s \cdot cK] \geq 1/2$.
\end{lemma} 

Here, if $Y$ is the $s_g$-subgaussian distribution supported on $\sum_{i=1}^n
\set{-v_i,v_i}-t$ as above, we simply let $\chi$ denote the random variable such
that $Y = \sum_{i=1}^n \chi_i v_i - t$. That $\chi$ now yields the desired
universal distribution on colorings is exactly the statement of the lemma. 

As a consequence of the above, we see that to recover Banaszczyk's theorem for
symmetric convex bodies, it suffices to be able to efficiently sample from an
$O(1)$-subgaussian distribution over  sets of the type $\sum_{i=1}^n \set{-v_i,v_i}-t$,
for $t \in \sum_{i=1}^n [-v_i,v_i]$, when $v_1,\dots,v_n \in \R^n$ are linearly
independent and have $\ell_2$ norm at most $1$. Here we rely on homogeneity,
that is, if $Y$ is an $s$-subgaussian random variable supported on $\sum_{i=1}^n
\set{-v_i,v_i}-t$  then $\alpha Y$ is $\alpha s$-subgaussian on $\sum_{i=1}^n
\set{-\alpha v_i,\alpha v_i}-\alpha t$, for $\alpha > 0$.    

The proof of Lemma~\ref{lem:subg-in-K} (see section~\ref{sec:subgaussian} for
more details) follows relatively directly from well-known convex geometric
estimates combined with Talagrand's majorizing measures theorem, which gives a
powerful characterization of the supremum of any Gaussian process.

Unfortunately, Lemma~\ref{lem:subg-in-K} does not hold for asymmetric convex
bodies. In particular, if $Y = -e_1$, the negated first standard basis vector, and $K
= \set{x \in \R^n: x_1 \geq 0}$, the conclusion is clearly false no matter how
much we scale $K$, even though $Y$ is $O(1)$-subgaussian and $K$ has Gaussian
measure $1/2$. One may perhaps hope that the conclusion still holds if we ask
for either $Y$ or $-Y$ to be in $s \cdot c K$ in the asymmetric setting, though
we do not know how to prove this. We note however that this only makes sense
when the support of $Y$ is symmetric, which does not necessarily hold in the
linear discrepancy setting.

We now describe the high level idea of the proof for the reverse
direction, namely, that $s_g = O(s_b)$. For this purpose, we show that
the existence of a $O(s_b)$-subgaussian distribution on $T$ can be
expressed as a two player zero-sum game, i.e.~the first player chooses
a distribution on $T$ and the second player tries to find a
non-subgaussian direction. Here the value of the game will be small if
and only if the $O(s_b)$-subgaussian distribution exists.  To bound
the value of the game, we show that an appropriate ``convexification''
of the space of subgaussianity tests for the second player can be
associated with symmetric convex bodies of Gaussian measure at least
$1/2$. From here, we use von Neumann's minimax principle to switch the
first and second player, and deduce that the value of the game is
bounded using the definition of $s_b$. 

\subsubsection{The Random Walk Sampler}
\label{sec:intro-rand-walk}

From the algorithmic perspective, it turns out that subgaussianity is a very
natural property in the context of random walk approaches to discrepancy
minimization.  Our results can thus be seen as a good justification for
the random walk approaches to making Banaszczyk's theorem constructive.

At a high level, in such approaches one runs a random walk over the coordinates
of a ``fractional coloring'' $\chi \in [-1,1]^n$ until  all the coordinates
hit either $1$ or $-1$. The steps of such a walk usually come from Gaussian
increments (though not necessarily spherical), which try to balance the
competing goals of keeping discrepancy low and moving the fractional coloring
$\chi$ closer to $\set{-1,1}^n$. Since a sum of small centered Gaussian increments is
subgaussian with the appropriate parameter, it is natural to hope that the output of a correctly
implemented random walk is subgaussian. Our main result in this setting
is that this is indeed possible to a limited extent, with the main caveat
being that the walk's output will not be ``subgaussian enough'' to fully recover
Banaszczyk's theorem.

\begin{theorem}
\label{thm:rand-walk}
Let $v_1,\dots,v_n \in \R^n$ be vectors of $\ell_2$ norm at most $1$
 and let $t \in \sum_{i=1}^n [-v_i,v_i]$. Then, there is an expected
polynomial time algorithm which outputs a random coloring $\chi \in
\set{-1,1}^n$ such that the random variable $\sum_{i=1}^n \chi_i v_i - t$ is
$O(\sqrt{\log n})$-subgaussian.
\end{theorem}

To achieve the above sampler, we guide our random walk using solutions to the
so-called vector K{\'o}mlos program, whose feasibility was first
given by Nikolov~\cite{Nikolov13}, and show subgaussianity using well-known
martingale concentration bounds. Interestingly, the random walk's analysis does
not rely on phases, and is instead based on a simple relation between the walk's
convergence time and the subgaussian parameter. As an added bonus, we also give
a new and simple constructive proof of the feasibility of the vector K{\'o}mlos
program (see section~\ref{sec:vector-komlos} for details) which avoids the use
of an SDP solver. 

Given the results of the previous section, the above random walk is a universal
sampler for constructing the following colorings.

\begin{corollary}
\label{cor:sym-body-sample}
Let $v_1,\dots,v_n \in \R^n$ be vectors of $\ell_2$ norm at most
$1$,  let $t \in
\sum_{i=1}^n [-v_i,v_i]$, and let $K \subseteq \R^n$ be a symmetric convex
body of Gaussian measure $1/2$ (given by a membership oracle).  Then,
there is an expected polynomial time algorithm which outputs a coloring $\chi
\in \set{-1,1}^n$ such that $\sum_{i=1}^n \chi_i v_i - t \in
O(\sqrt{\log n}) K$.
\end{corollary}

As mentioned previously, the best previous algorithms in this setting are due to
Rothvoss~\cite{Rothvoss14}, Eldan and Singh~\cite{EldanSingh14}, which find a
signed combination inside $O(\log n)K$. Furthermore, these algorithms are not
universal, i.e.~they heavily depend on the body $K$. We note that these
algorithms are in fact tailored to find \emph{partial colorings} inside a
symmetric convex body $K$ of Gaussian measure at least $2^{-cn}$, for $c > 0$
small enough, a setting in which our sampler does not provide any guarantees.

We now recall prior work on random walk based discrepancy minimization. The
random walk approach was pioneered by Bansal~\cite{Bansal10}, who used a
semidefinite program to guide the walk and gave the first efficient algorithm
matching the classic $O(\sqrt{n})$ bound of Spencer~\cite{Spencer85} for the
combinatorial discrepancy of set systems satisfying $m=O(n)$. Later, Lovett and
Meka~\cite{LovettMeka12} provided a greatly simplified walk, removing the need
for the semidefinite program, which recovered the full power of Beck's entropy
method for constructing partial colorings. Harvey, Schwartz, and
Singh~\cite{HSS14} defined another random walk based algorithm, which,
unlike previous work and similarly to our algorithm, doesn't
explicitly use phases or produce partial colorings. The random walks
of~\cite{LovettMeka12}~and~\cite{HSS14} both depend on the convex
body $K$; the walk in~\cite{LovettMeka12} is only well-defined in a
polytope, while the one in~\cite{HSS14} remains well-defined in any
convex body, although the analysis still applies only to the polyhedral
setting.
Most directly related to this paper
is the recent work~\cite{BDG16}, which gives a walk that can be viewed as a
randomized variant of the original $2t-1$ Beck-Fiala proof. This walk induces a
distribution $\chi \in \set{-1,1}^n$ on colorings for which \emph{each
coordinate} of the output $\sum_{i=1}^n \chi_i v_i$ is $O(1)$-subgaussian. From
the discrepancy perspective, this gives a sampler which finds colorings inside
any axis parallel box of Gaussian measure at least $1/2$ (and their
rotations, though not in a universal manner), matching Banaszczyk's result for
this class of convex bodies. 

\subsubsection{Asymmetric Convex Bodies}
\label{sec:intro-asymmetric}

In this section, we explain how our techniques extend to the asymmetric setting.
The main difficulty in the asymmetric setting is that one cannot hope to
increase the Gaussian mass of an asymmetric convex body by simply scaling it. In
particular, if we take $K \subseteq \R^n$ to be a halfspace through the origin,
e.g.~$\set{x \in \R^n: x_1 \geq 0}$, then $K$ has Gaussian measure exactly $1/2$
but $sK = K$ for all $s > 0$. At a technical level, the lack of any measure
increase under scaling breaks the proof of Lemma~\ref{lem:subg-in-K}, which is
crucial for showing that subgaussian coloring distributions produce combinations
that land inside $K$. 

The main idea to circumvent this problem will be to reduce to a setting where
the mass of $K$ is ``symmetrically distributed'' about the origin, in
particular, when the barycenter of $K$ under the induced Gaussian measure is at
the origin. For such a body $K$, we show that a constant factor scaling of $K
\cap -K$ also has Gaussian mass at least $1/2$, yielding a direct reduction to
the symmetric setting. 
  
To achieve this reduction, we will use a novel \emph{recentering procedure},
which will both carefully fix certain coordinates of the coloring as well as
shift the body $K$ to make its mass more ``symmetrically distributed''. The
guarantees of this procedure are stated below:

\begin{theorem}[Recentering Procedure]
\label{thm:recenter}
Let $v_1,\dots,v_n \in \R^n$ be linearly independent, $t \in \sum_{i=1}^n
[-v_i,v_i]$, and $K \subseteq \R^n$ be a convex body of Gaussian measure at
least $1/2$. Then, there exists a fractional coloring $x \in [-1,1]^n$,
such that for $p = \sum_{i=1}^n x_i v_i - t$, $A_x = \set{i \in [n]:
x_i \in (-1,1)}$ and $W = {\rm span}(v_i: i \in A_x)$, the following
holds:
\begin{enumerate}
\item $p \in K$. \label{pr1}
\item The Gaussian measure of $(K-p) \cap W$ on $W$ is at least the Gaussian
measure of $K$.  \label{pr2}
\item The barycenter of $(K-p) \cap W$ is at the origin, 
     i.e.~$\int_{(K-p) \cap W} y e^{-\|y\|^2/2} dy = 0$.  \label{pr3}
\end{enumerate}
\end{theorem}

By convention, if the procedure returns a full coloring $x \in
\set{-1,1}^n$ (in which case, since $p \in K$, we are done), we shall treat
conditions $2$ and $3$ as satisfied, even though $W = \{0\}$. At a high level, the recentering procedure
allows us to reduce the initial vector balancing problem to one in a possibly
lower dimension with respect to ``well-centered'' convex body of no smaller
Gaussian measure, and in particular, of Gaussian measure at least $1/2$.
Interestingly, as mentioned earlier in the introduction, the recentering
procedure can also be extended to yield a full coloring algorithm. We explain
the high level details of its implementation together with this extension in the
next subsection.  

To explain how to use the fractional coloring $x$ from
Theorem~\ref{thm:recenter} to get a useful reduction, recall the
lifting function $L_x: [-1,1]^{A_x} \rightarrow [-1,1]^n$ defined in
\eqref{eq:lifting}. We reduce the initial vector balancing problem to
the problem of finding a coloring $\chi \in \{-1,1\}^{A_x}$ such that
$\sum_{i \in A_x}{\chi_i v_i} - \sum_{i \in A_x}{x_i v_i} \in
(K-p)\cap W$ (note that $\sum_{i \in A_x}{\chi_i v_i} - \sum_{i \in
  A_x}{x_i v_i} \in W$ by construction). Then we can lift this
coloring to $L_x(\chi)$, which satisfies 
\[
\sum_{i \in A_x} \chi_i v_i - \sum_{i \in A_x} x_i v_i \in (K-p) \cap W \Leftrightarrow
\sum_{i=1}^n L_x(\chi)_i v_i - t \in K.
\]

From here, the guarantee that $K' \eqdef (K-p) \cap W$ has Gaussian measure at
least $1/2$ and barycenter at the origin allows a direct reduction to the
symmetric setting. Namely, we can replace $K'$ by the symmetric convex body $K'
\cap -K'$ without losing ``too much'' of the Gaussian measure of $K'$. This is
formalized by the following extension of Lemma~\ref{lem:subg-in-K}, which
directly implies a reduction to subgaussian sampling as in
section~\ref{sec:intro-subg}.   

\begin{lemma} 
\label{lem:subg-in-asym-K}
Let $Y \in \R^n$ be an $s$-subgaussian random variable. There exists an absolute
constant $c > 0$, such for any convex body $K \subseteq \R^n$ of Gaussian
measure at least $1/2$ and barycenter at the origin, $\Pr[Y \in s \cdot c(K \cap
-K)] \geq 1/2$.
\end{lemma} 

In particular, if there exists a distribution over colorings $\chi \in
\set{-1,1}^{A_x}$ such that $\sum_{i \in A_x} \chi_i v_i - \sum_{i \in A_x} x_i v_i$ as above is
$1/c$-subgaussian, Lemma~\ref{lem:subg-in-asym-K} implies that the random signed
combination lands inside $K'$ with probability at least $1/2$. Thus, the
asymmetric setting can be effectively reduced to the symmetric one, as claimed. 

Crucially, the recentering procedure in Theorem~\ref{thm:recenter} can be
implemented in probabilistic polynomial time if one relaxes the
barycenter condition from being exactly $0$ to having ``small'' norm (see
section~\ref{sec:recenter} for details). Furthermore, the estimate in
Lemma~\ref{lem:subg-in-asym-K} will be robust to such perturbations. Thus, to
constructively recover the colorings in the asymmetric setting,
it will still suffice to be able to generate good subgaussian coloring
distributions.  

Combining the sampler from Theorem~\ref{thm:rand-walk} together with the
recentering procedure, we constructively recover Banaszczyk's theorem for
general convex bodies up to a $O(\sqrt{\log n})$ factor. 

\begin{theorem}[Weak Constructive Banaszczyk]
\label{thm:weak-bana}
There exists a probabilistic polynomial time algorithm which, on input a
linearly independent set of vectors $v_1, \ldots, v_n \in \R^n$ of $\ell_2$ norm
at most $c/\sqrt{\log n}$, $c > 0$ small enough, $t \in \sum_{i=1}^n
[-v_i,v_i]$, and a (not necessarily symmetric) convex body $K \subseteq
\R^n$ of Gaussian measure at least $1/2$ (given by a membership oracle),
computes a coloring $\chi \in \set{-1,1}^n$ such that with high probability
$\sum_{i = 1}^n \chi_i v_i-t \in K$. 
\end{theorem}


As far as we are aware, the above theorem gives the first algorithm to recover
Banaszczyk's result for asymmetric convex bodies under any non-trivial
restriction. In this context, we note that the algorithm of Eldan and
Singh~\cite{EldanSingh14} finds ``relaxed'' partial colorings, i.e.~ where the
fractional coordinates of the coloring are allowed to fall outside $[-1,1]$,
lying inside an $n$-dimensional convex body of Gaussian measure at least
$2^{-cn}$. However, it is unclear how one could use such partial colorings to
recover the above result, even with a larger approximation factor.

\subsubsection{The Recentering Procedure}
\label{sec:intro-recentering}


In this section, we describe the details of the recentering procedure. We leave
a thorough description of its algorithmic implementation however to
section~\ref{sec:recenter}, and only provide its abstract instantiation here.   

Before we begin, we give a more geometric view of the vector balancing problem
and the recentering procedure, which help clarify the exposition. Let
$v_1,\dots,v_n \in \R^n$ be linearly independent vectors and $t \in \sum_{i=1}^n
[-v_i,v_i]$. Given the target body $K \subseteq \R^n$ of Gaussian measure at
least $1/2$, we can restate the vector balancing problem geometrically as that
of finding a vertex of the parallelepiped $P = \sum_{i=1}^n [-v_i,v_i] - t$
lying inside $K$. Here, the choice of $t$ ensures that $0 \in P$. Note that this
condition is necessary, since otherwise there exists a halfspace separating $P$
from $0$ having Gaussian measure at least $1/2$. \cut{Moreover, note that any convex
body of Gaussian measure at least $1/2$ must contain the origin by the same
reasoning. }

Recall now that in the linear setting, and using this geometric language, Banaszczyk's theorem implies
that if $P$ contains the origin, and $\max_{i \in [n]} \|v_i\|_2 \leq 1/10$
(which we do not need to assume here), then any convex body of Gaussian measure at
least $1/2$ contains a vertex of $P$. Thus, for our given target body $K$, we
should make our situation better replacing $P$ and $K$ by $P-q$ and $K-q$, if
$q \in P$ is a shift such that $K-q$ has higher Gaussian measure than $K$. In
particular, given the symmetry of Gaussian measure, one would intuitively expect
that if the Gaussian mass of $K$ is not symmetrically distributed around $0$, there should be a shift of $K$ which increases its
Gaussian measure. 

In the current language, fixing a color $\chi_i \in \set{-1,1}$ for vector
$v_i$, corresponds to restricting ourselves to finding a vertex in the facet $F
= \chi_i v_i + \sum_{j \neq i} [-v_j,v_j] - t$ of $P$ lying inside $K$.  Again
intuitively, restricting to a facet of $P$ should improve our situation if the
Gaussian measure of the corresponding slice of $K$ in the lower dimension is
larger than that of $K$. To make this formal, note that when inducting on a
facet $F$ of $P$ (which is an $n-1$ dimensional parallelepiped), we must choose
a center $q \in F$ to serve as the new origin in the lower dimensional space.
Precisely, this can be expressed as inducting on the parallelepiped $F-q$ and
shifted slice $(K-q) \cap {\rm span}(F-q)$ of $K$, using the $n-1$ dimensional
Gaussian measure on ${\rm span}(F-q)$. 

With the above viewpoint, one can restate the goal of the recentering procedure
as that of finding a point $q \in P \cap K$, such that smallest facet $F$ of $P$
containing $q$, satisfies that $(K-q) \cap {\rm span}(F-q)$ has its barycenter
at the origin and Gaussian measure no smaller than that of $K$. Recall that as
long as $(K-q) \cap {\rm span}(F-q)$ has Gaussian measure at least $1/2$, we are
guaranteed that $0 \in K-q \Rightarrow q \in K$. With this geometry in mind, we
implement the recentering procedure as follows:
 
Compute $q \in P$ so that the Gaussian mass of $K-q$ is maximized. If $q$ is on
the boundary of $P$, letting $F$ denote a facet of $P$ containing $q$, induct on
$F-q$ and the slice $(K-q) \cap {\rm span}(F-q)$ as above. If $q$ is in the
interior of $P$, replace $P$ and $K$ by $P-q$ and $K-q$, and terminate. 

We now explain why the above achieves the desired result. Firstly, if the
maximizer $q$ is in a facet $F$ of $P$, then a standard convex geometric
argument reveals that the Gaussian measure of $(K-q) \cap {\rm span}(F-q)$ is
no smaller than that of $K-q$, and in particular, no smaller than that of $K$.
Thus, in this case, the recentering procedure fixes a color for ``free''. In the
second case, if $q$ is in the interior of $P$, then a variational argument gives
that the barycenter of $K-q$ under the induced Gaussian measure must be at the
origin, namely, $\int_{K-q} xe^{-x^2/2} dx = 0$. 

To conclude this section, we explain how to extend the recentering procedure to
directly produce a deterministic coloring satisfying
Theorem~\ref{thm:weak-bana}. For this purpose, we shall assume that
$v_1,\dots,v_n$ have length at most $c/\sqrt{\log n}$, for a small enough
constant $c > 0$. To begin, we run the recentering procedure as above, which
returns $P$ and $K$, with $K$ having its barycenter at the origin. We now
replace $P,K$ by a joint scaling $\alpha P, \alpha K$, for $\alpha > 0$ a large
enough constant, so that $\alpha K$ has Gaussian mass at least $3/4$. At this
point, we run the original recentering procedure again with the following
modification: every time we get to the situation where $K$ has its barycenter at
the origin, induct on the closest facet of $P$ closest to the origin. More precisely,
in this situation, compute a point $p$ on the boundary of $P$ closest to the
origin, and, letting $F$ denote the facet containing $p$, induct on $F-p$ and
$(K-p) \cap {\rm span}(F-p)$. At the end, return the final found vertex. 

Notice that, as claimed, the coloring (i.e.~vertex) returned by the algorithm
is indeed deterministic. The reason the above algorithm works is the following.
While we cannot guarantee, as in the original recentering procedure, that the
Gaussian mass of $(K-p) \cap {\rm span}(F-p)$ does not decrease, we can instead show
that it decreases only \emph{very slowly}. In particular, we use the
bound of $O(1/\sqrt{\log n})$ on the length of the vectors $v_1,\dots,v_n$ to
show that every time we induct, the Gaussian mass drops by at most a $1-c/n$
factor. More generally, if the vectors had length at most $d > 0$,
for $d$ small enough, the drop would be of the order $1-ce^{-1/(cd)^2}$, for some
constant $c > 0$. Since we ``massage'' $K$ to have Gaussian mass at least $3/4$
before applying the modified recentering algorithm, this indeed allows to induct
$n$ times while keeping the Gaussian mass above $1/2$, which guarantees that the
final vertex is in $K$. To derive the bound on the rate of decrease of Gaussian
mass, we prove a new inequality on the Gaussian mass of sections of a convex
body near the barycenter (see Theorem~\ref{thm:estneg}), which may be of
independent interest. 

As a final remark, we note that unlike the subgaussian sampler, the recentering
procedure is not scale invariant. Namely, if we jointly scale $P$ and $K$ by
some factor $\alpha$, the output of the recentering procedure will not be an
$\alpha$-scaling of the output on the original $K$ and $P$, as Gaussian measure
is not homogeneous under scalings. Thus, one must take care to appropriately
normalize $P$ and $K$ before applying the recentering procedure to achieve the
desired results.

We now give the high level overview of our recentering step implementation. The
first crucial observation in this context, is that the task of finding $t \in P$
maximizing the Gaussian measure of $K-t$ is in fact a \emph{convex program}.
More precisely, the objective function (Gaussian measure of $K-t$) is a
logconcave function of $t$ and the feasible region $P$ is convex. Hence, one can
hope to apply standard convex optimization techniques to find the desired
maximizer. 

It turns out however, that one can significantly simplify the required task by
noting that the recentering strategy does not in fact necessarily need an exact
maximizer, or even a maximizer in $P$. To see this, note that if $p$ is a shift
such that $K-p$ has larger Gaussian measure than $K$, then by logconcavity the
shifts $K-\alpha p$, $0 < \alpha \leq 1$, also have larger Gaussian measure.
Thus, if a we find a shift $p \notin P$ with larger Gaussian measure, letting
$\alpha p$ be the intersection point with the boundary $\partial P$, we can
induct on the facet of $P-\alpha p$ containing $0$ and the corresponding slice
of $K-\alpha p$ just as before. Given this, we can essentially ``ignore'' the
constraint $p \in P$ and we treat the optimization problem as unconstrained.

This last observation will allow us to use the following simple gradient ascent
strategy. Precisely, we simply take steps in the direction of the gradient until
either we pass through a facet of $P$ or the gradient becomes ``too small''. As
alluded to previously, the gradient will exactly equal a fixed scaling of the
barycenter of $K-p$, $p$ the current shift, under the induced Gaussian measure.
Thus, once the gradient is small, the barycenter will be very close to the
origin, which will be good enough for our purposes. The last nontrivial
technical detail is how to efficiently estimate the barycenter, where we note
that the barycenter is the expectation of a random point inside $K-p$. For this
purpose, we simply take an average of random samples from $K-p$, where
we generate the samples using rejection sampling, using the fact that
the Gaussian measure of $K$ is large.


\paragraph*{Conclusion and Open Problems}
In conclusion, we have shown a tight connection between the existence of
subgaussian coloring distributions and Banaszczyk's vector balancing theorem.
Furthermore, we make use of this connection to constructively recover a weaker
version of this theorem. The main open problem we leave is thus to fully recover
Banaszczyk's result. As explained above, this reduces to finding a distribution
on colorings such that the output random signed combination is
$O(1)$-subgaussian, when the input vectors have $\ell_2$ norm at most $1$. We
believe this approach is both attractive and feasible, especially given the
recent work~\cite{BDG16}, which builds a distribution on colorings for which
each coordinate of the output random signed combination is $O(1)$-subgaussian. 

\paragraph*{Organization}
In section~\ref{sec:prelims}, we provide necessary preliminary background
material. In section~\ref{sec:subgaussian}, we give the proof of the equivalence
between Banaszczyk's vector balancing theorem and the existence of subgaussian
coloring distributions. In section~\ref{sec:rand-walk}, we give our random walk
based coloring algorithm. In section~\ref{sec:recenter}, we describe the
implementation of the recentering procedure. In
section~\ref{sec:alg-asym-to-sym}, we give the algorithmic reduction from
asymmetric bodies to symmetric bodies, giving the proof of
Theorem~\ref{thm:weak-bana}. In section~\ref{sec:asymmetric}, we show how extend the
recentering procedure to a full coloring algorithm. In
section~\ref{sec:slice-estimate-bar}, we prove the main technical estimate on
the Gaussian measure of slices of a convex body near the barycenter, which is
needed for the algorithm in~\ref{sec:asymmetric}. Lastly, in
section~\ref{sec:vector-komlos}, we give our constructive proof of the
feasibility of the vector K{\'o}mlos program. 

\paragraph*{Acknowledgments}
We would like to thank the American Institute for Mathematics for hosting a
recent workshop on discrepancy theory, where some of this work was done.

\section{Preliminaries}
\label{sec:prelims}

\paragraph*{Basic Concepts}
We write $\log x$ and $\log_2 x$, $x > 0$, for the logarithm base $e$ and base
$2$ respectively.  

For a vector $x \in \R^n$, we define $\|x\|_2 = \sqrt{\sum_{i=1}^n x_i^2}$ to be
its Eucliean norm. Let $B_2^n = \set{x \in \R^n: \|x\|_2 \leq 1}$ denote the
unit Euclidean ball and $S^{n-1} = \set{x \in \R^n: \|x\|_2 = 1} =
\partial \B_2^n$ denote the unit sphere in $\R^n$. For $x,y \in \R^n$, we denote
their inner product $\pr{x}{y} = \sum_{i=1}^n x_i y_i$. 

For subsets $A,B \subseteq \R^n$, we denote their Minkowski sum $A+B = \set{a+b:
a \in A, b \in B}$. Define ${\rm span}(A)$ to be the smallest linear subspace
containing $A$. We denote the boundary of $A$ by $\partial A$. We use the phrase
$\partial A$ relative to ${\rm span}(A)$ to specify that we are computing the
boundary with respect to the subspace topology on ${\rm span}(A)$.

A set $K \subseteq \R^n$ is convex if for all $x,y \in K$,$\lambda \in [0,1]$,
$\lambda x + (1-\lambda)y \in K$. $K$ is symmetric if $K=-K$. We shall say that
$K$ is a convex body if additionally it is closed and has non-empty interior. We
note that the usual terminology, a convex body is also compact (i.e.~bounded),
but we will state this explicitly when it is necessary. If convex body contains
the origin in its interior, we say that $K$ is $0$-centered.

We will need the concept of a gauge function for $0$-centered convex bodies. For
bounded symmetric convex bodies, this functional will define a standard norm. 

\begin{proposition} 
\label{prop:gauge}
Let $K \subseteq \R^n$ be a $0$-centered convex body. Defining the \emph{gauge
function} of the body $K$ by $\|x\|_K =  \inf \set{s \geq 0: x \in sK}$,
the following holds:
\begin{enumerate}
\item {\bf Finiteness:} $\|x\|_K < \infty$, for $x \in \R^n$.
\item {\bf Positive homogeneity:} $\|\lambda x\|_K = \lambda \|x\|_K$, for $x \in \R^n, \lambda \geq 0$.
\item {\bf Triangle inequality:} $\|x+y\|_K \leq \|x\|_K + \|y\|_K$, for $x,y \in \R^n$.
\end{enumerate}
Furthermore, if $K$ is additionally bounded and symmetric, then $\|\cdot\|_K$ is a
norm which we call the norm induced by $K$. In particular, $\|\cdot\|_K$
additionally satisfies that $\|x\|_K = 0$ iff $x = 0$ and $\|x\|_K = \|-x\|_K$~
$\forall x \in \R^n$. 
\end{proposition}

\paragraph*{Gaussian and subgaussian random variables}
We define $n$-dimensional standard Gaussian $X \in \R^n$ to be the random
variable with density $\frac{1}{\sqrt{2\pi}^n} e^{-\|x\|_2^2/2}$ for $x \in \R^n$. 

\begin{definition}[Subgaussian Random Variable]
\label{def:subg}
A random variable $X \in \R$ is $\sigma$-subgaussian, for $\sigma > 0$, if
$\forall t \geq 0$,
\[
\Pr[|X| \geq t] \leq 2 e^{-\frac{1}{2}(t/\sigma)^2} \text{ .}
\]
We note that the canonical example of a $1$-subgaussian distribution is the
$1$-dimensional standard Gaussian itself.

For a vector valued random variable $X \in \R^n$, we say that $X$ is
$\sigma$-subgaussian if all its one dimensional marginals are. Precisely, $X$ is
$\sigma$-subgaussian if $\forall \theta \in S^{n-1}$, the random variable
$\pr{X}{\theta}$ is $\sigma$-subgaussian.
\end{definition}

We remark that from definition~\ref{def:subg}, it follows directly that if $X$
is $\sigma$-subgaussian then $\alpha X$ is $|\alpha|\sigma$-subgaussian for any
$\alpha \in \R$.

The following standard lemma allows us to deduce subgaussianity from upper
bounds on the Laplace transform of a random variable. We include a proof in the
appendix for completeness.

\begin{lemma}
\label{lem:mom-to-subg}
Let $\cosh(x) = \frac{1}{2}(e^{x} + e^{-x})$ for
$x \in \R^n$.  Let $X \in \R^n$ be a random vector. Assume that
\[
\E[\cosh(\pr{w}{X})] \leq \beta e^{\|\sigma w\|_2^2/2}, \quad \forall w
\in \R^n \text{ ,}
\]
for some $\sigma > 0$ and $\beta \geq 1$. Then $X$ is $\sigma \sqrt{\log_2 \beta
+ 1}$-subgaussian. Furthermore, for $X \in \R^n$ standard Gaussian,
$\E[\cosh(\pr{w}{X})] = e^{\|w\|_2^2/2}$ for $w \in \R^n$.
\end{lemma}

\paragraph*{Gaussian measure}
We define $\gamma_n$ to be the $n$-dimensional Gaussian measure on $\R^n$.
Precisely, for any measurable set $A \subseteq \R^n$, 
\begin{equation}
\gamma_n(A) = \frac{1}{\sqrt{2\pi}^n} \int_A e^{-\|x\|^2_2/2} dx,
\end{equation}
noting that $\gamma_n(\R^n) = 1$. We will also need lower dimensional Gaussian
measures restricted to linear subspaces of $\R^n$. Thus, if $A \subseteq W$, $W
\subseteq \R^n$ a linear subspace of dimension $k$, then $\gamma_k(A)$ should be
understood as the Gaussian measure of $A$ within $W$, where $W$ is treated as
the whole space. When convenient, we will also use the notation $\gamma_W(A)$ to
denote $\gamma_{\dim(W)}(A \cap W)$. When treating one dimensional Gaussian
measure, we will often denote $\gamma_1((a,b))$, where $(a,b)$ is an interval,
simply by $\gamma_1(a,b)$ for notational convenience. By convention, we define
$\gamma_0(A) = 1$ if $0 \in A$ and $0$ otherwise. 

An important concept used throughout the paper is that of the barycenter under
the induced Gaussian measure.
\begin{definition}[Barycenter]
\label{def:barycenter}
For a convex body $K \subseteq \R^n$, we define its barycenter under the induced
Gaussian measure, by
\[
b(K) = \frac{1}{\sqrt{2\pi}^n} \int_K x e^{-\|x\|_2^2/2} \frac{dx}{\gamma_n(K)} \text{ .}
\]
Note that $b(K) = \E[X]$, if $X$ is the random variable supported on $K$ with
probability density $\frac{1}{\sqrt{2\pi}^n} e^{-\|x\|_2^2/2}/\gamma_n(K)$.
Extending the definition to slices of $K$, for any linear subspace $W \subseteq
\R^n$, we refer to the barycenter of $K \cap W$ to denote the one relative to
the $\dim(W)$-dimensional Gaussian measure on $W$ (i.e.~treating $W$ as the
whole space).
\end{definition}

Throughout the paper, we will need many inequalities regarding the Gaussian
measure. The first important inequality is the Pr{\'e}kopa-Leindler inequality,
which states that for $\lambda \in [0,1]$ and $A,B,\lambda A + (1-\lambda)B
\subseteq \R^n$ measurable subsets, that
\begin{equation}
\label{thm:prekopa-leinder}
\gamma_n(\lambda A + (1-\lambda) B) \geq \gamma_n(A)^\lambda
\gamma_n(B)^{1-\lambda} \text{ .}
\end{equation}
We note that the Pr{\'e}kopa-Leindler inequality applies more generally to any
logconcave measure on $\R^n$, i.e.~a measure defined by a density whose
logarithm is concave. Importantly, this inequality directly implies that if $A
\subseteq \R^n$ is convex, then $\log \gamma_n(A + t)$, for $t \in \R^n$, is a
concave function of $t$.  

We will need the following powerful inequality of Ehrhard, which provides a
crucial strengthening of Pr{\'e}kopa-Leindler for Gaussian measure.   

\begin{theorem}[Ehrhard's inequality~\cite{Ehrhard83,Borell03}] 
For Borel sets $A,B\subseteq \mathbb{R}^n $ and $0\le \lambda \le 1$, 
\[\Phi^{-1}(\gamma_n(\lambda A+(1-\lambda)B)) \ge \lambda \Phi^{-1}(\gamma_n(A))+(1-\lambda)\Phi^{-1}(\gamma_n(B))\]
where $\Phi(a)=\gamma_1((-\infty, a])$ for all $a\in\mathbb{R}$.\\
\end{theorem}

The power of the Ehrhard inequality is that it allows us to reduce many
non-trivial inequalities about Gaussian measure to two dimensional ones.

One can use it to show the following standard inequality on the Gaussian
measures of slices of a convex body. We include a proof for completeness.

\begin{lemma}
\label{lem:origincut}
Given a convex body $K\subseteq\mathbb{R}^n$ with $\gamma_n(K) \geq 1/2$, and a
linear subspace $H \subseteq \mathbb{R}^n$ of dimension $k$. Then,
$\gamma_k(K\cap H)\geq \gamma_n(K)$. 
\end{lemma}
\begin{proof}
Clearly it suffices to prove the lemma for $k=n-1$. Since Gaussian distribution
is rotation invariant, without loss of generality, $H=\{x\in \mathbb{R}^n :
x_1=0\}$. Let $K_t=\{x\in K: x_1=t\}$ denote a slice of $K$ at $x_1=t$. Then, 
\[
\gamma_n(K) = \int_{-\infty}^\infty
\frac{e^{-t^2/2}}{\sqrt{2\pi}} \gamma_{n-1}(K_t-te_1)dt
\]
where $\gamma_{n-1}(K_t-te_1)=0$ outside support of $K$.

\noindent
Define $W\subseteq \mathbb{R}^2$ as $W=\{(x,y) : y\le f(x)\}$ where $f:\mathbb{R}\rightarrow \mathbb{R}$ is defined as 
\[f(t)=\Phi^{-1}(\gamma_{n-1}(K_t-t e_1)) \]
and $f(t)=-\infty$ outside the support of $K$. It follows that $\gamma_2(W)=\gamma_n(K)\ge1/2$. By Ehrhard's inequality, $f$ is concave on its support. Hence, $W$ is a closed convex body. 

\noindent
Let $g=\Phi^{-1}(\gamma_n(K))\ge0$.
$\gamma_{n-1}(K\cap H)\ge \gamma_n(K)$ is then equivalent to showing $(0,g)\in
W$. If $(0,g)\not\in W$, then there exists a halfspace $L$ such that $W
\subseteq L$ and $(0,g)\not\in L$. Let $d$ be the distance of origin $(0,0)$
from $\partial L$, the boundary of $L$. Since $(0,g)\not\in L$ and $\gamma_2(L)\ge 1/2$, $d<g$. But this implies 
\[\gamma_2(L)=\gamma_1(-\infty, d)<\gamma_1(-\infty, g)=\gamma_n(K)=\gamma_2(W) \]
contradicting $W\subseteq L$.
\end{proof}

\paragraph*{Vector Balancing: Reduction to the Linearly Independent Case}

In this section, we detail the standard vector balancing reduction to the case
where the vectors are linearly independent. We will also cover some useful
related concepts and definitions, which will be used throughout the paper.

\begin{definition}[Lifting Function]
\label{def:lifting}
For a fractional coloring $x \in [-1,1]^n$, denote the set of fractional
coordinates by $A_x = \set{i \in [n]: x_i \in (-1,1)}$. From here, for $z
\in [-1,1]^{A_x}$, we define the lifting function $L_x: [-1,1]^{A_x}
\rightarrow [-1,1]^n$ by
\[
L_x(z)_i = \begin{cases} z_i &: i \in A_x \\ 
                          x_i &: i \in [n] \setminus A_x \end{cases},~\forall i \in [n] \text{ .}
\]
Importantly, for $\chi \in \set{-1,1}^{A_x}$ we have that $L_x(\chi) \in
\set{-1,1}^n$.  Thus, $L_x$ sends full colorings in $\set{-1,1}^{A_x}$ to full
colorings in $\set{-1,1}^n$.
\end{definition}

The lifting function above is useful in that it allows us, given a fractional
coloring $x \in [-1,1]$ with some of its coordinates set to $\set{-1,1}$, to
reduce any linear vector balancing problem to one on a smaller number of
coordinates. We detail this in the following lemma.

\begin{lemma} 
\label{lem:partial-red}
Let $v_1,\dots,v_m \in \R^m$, $t \in \sum_{i=1}^n [-v_i,v_i]$,
and $K \subseteq \R^n$. Then given a fractional coloring $x \in [-1,1]^n$ 
and $p = \sum_{i=1}^n x_i v_i - t$, the following holds:
\begin{enumerate}
\item\label{lem:partial-red-id} For $z \in [-1,1]^{A_x}$, we have that 
\[
\sum_{i \in A_x} z_i v_i - \sum_{i \in A_x} x_i v_i + p 
= \sum_{i=1}^n L_x(z)_iv_i - t \text{ .}
\]
\item\label{lem:partial-red-stat} For $z \in [-1,1]^{A_x}$, we have that
\[
\sum_{i \in A_x} z_i v_i - \sum_{i \in A_x} x_i v_i \in K-p \Leftrightarrow
\sum_{i=1}^n L_x(z)_i v_i - t \in K \text{ .}
\]  
\end{enumerate}
\end{lemma}
\begin{proof}
The first part follows from the computation
\begin{align*}
\sum_{i \in [n]} L_x(z)_i v_i - t &= 
\sum_{i \in A_x} z_i v_i + \sum_{i \in [n] \setminus A_x} x_i v_i - t \\
                                      &=
\sum_{i \in A_x} z_i v_i - \sum_{i \in A_x} x_i v_i + 
(\sum_{i \in [n]} x_i v_i - t) \\
                                      &=
\sum_{i \in A_x} z_i v_i - \sum_{i \in A_x} x_i v_i + p \text{ .}
\end{align*}
The second part follows since
\[
\sum_{i \in A_x} z_i v_i - \sum_{i \in A_x} x_i v_i \in K-p
\Leftrightarrow
\sum_{i \in A_x} z_i v_i - \sum_{i \in A_x} x_i v_i + p \in K 
\Leftrightarrow
\sum_{i=1}^n L_x(z)_i v_i - t \in K \text{ ,}
\]
where the last equivalence is by part (1).
\end{proof}

In terms of a reduction, the above lemma says in words that the linear vector
balancing problem with respect to the vectors $(v_i: i \in [n])$, shift $t$ and
set $K$, reduces to the linear discrepancy problem on $(v_i: i \in A_x)$, shift
$\sum_{i \in A_x} x_i v_i$ and set $K-p$. 

We now give the reduction to the linearly independent setting.
\begin{lemma}
\label{lem:red-to-ind}
Let $v_1,\dots,v_n \in \R^m$, $t \in \sum_{i=1}^n [-v_i,v_i]$. Then, there
is a polynomial time algorithm computing a fractional coloring $x \in [-1,1]^n$
such that:
\begin{enumerate}
\item\label{lem:red-to-ind-1} $\sum_{i=1}^n x_i v_i = t$.
\item\label{lem:red-to-ind-2} The vectors $(v_i: i \in A_x)$ are linearly independent.
\item\label{lem:red-to-ind-3} For $z \in [-1,1]^{A_x}$, 
$\sum_{i \in A_x} z_i v_i - \sum_{i \in A_x} x_i v_i  
= \sum_{i=1}^n L_x(z)_i v_i - t$. 
\end{enumerate}
\end{lemma}
\begin{proof}
Let $x$ denote a basic feasible solution to the linear system
\[
\set{y \in \R^n: \sum_{i=1}^n y_i v_i = t, y_i \in [-1,1] ~ \forall 
i \in [n]} \text{ ,}
\]
which clearly can be computed in polynomial time. Note the system is feasible by
construction of $t$. We now show that $x$ satisfies the required
conditions.

Let $r \leq n$ denote the rank of the matrix $(v_1,\dots,v_n)$. Since
$x$ is basic, it must satisfy at least $n$ least of the constraints at
it equality. In particular, at least $n-r$ of the bound constraints must be
tight. Thus, since $A_x$ is the set of fractional coordinates, we must have
$|A_x| \leq r$. Furthermore, the vectors $(v_i: i \in A_x)$ must be linearly
independent, since otherwise $x$ is not basic. Finally, for $z \in
[-1,1]^{A_x}$, we have that
\[
\sum_{i \in A_x} z_i v_i - \sum_{i \in A_x} x_i v_i  = 
\sum_{i \in A_x} z_i v_i + \sum_{i \in [n] \setminus A_x} x_i v_i
- \sum_{i \in [n]} x_i v_i = \sum_{i \in [n]} L_x(z)_i
  v_i - t \text{ ,}
\]
as needed.
\end{proof}

Let us now apply the above lemma to both the vector balancing problem and the
subgaussian sampling problem. First assume that we have a vector balancing
problem with respect to $v_1,\dots,v_n \in \R^m$, shift $t \in \sum_{i=1}^n
[-v_i,v_i]$, and $K \subseteq \R^m$ a convex body of Gaussian measure at least
$1/2$. Then applying the above lemma, we get $x \in [-1,1]^n$, such that our
vector balancing reduces to the one with respect to $(v_i: i \in A_x)$, shift
$\sum_{i \in A_x} x_i v_i \in \sum_{i \in A_x} [-v_i,v_i]$, and $K$. This
follows directly from Lemma~\ref{lem:red-to-ind} part~\ref{lem:red-to-ind-3}
using the lifting function $L_x$. Now let $W = {\rm span}(v_i: i \in A_x)$,
where $\dim(W) = |A_x|$ by linear independence. Clearly, the reduced vector
balancing problem looks for signed combinations in $W$, and hence we may replace
$K$ by $K \cap W$. Here, note that by Lemma~\ref{lem:origincut},
$\gamma_{|A_x|}(K \cap W) \geq \gamma_m(K) \geq 1/2$. Hence, this reduction
reduces to a problem of the same type, where in addition, the vectors form a
basis of the ambient space $W$. For the subgaussian sampling problem, by the
identity~\ref{lem:red-to-ind-3} in Lemma~\ref{lem:red-to-ind}, sampling a random
coloring $\chi \in \set{-1,1}^n$ such that $\sum_{i=1}^n \chi_i v_i - t$ is
subgaussian clearly reduces to sampling a random coloring $\chi \in
\set{-1,1}^{A_x}$ such that $\sum_{i \in A_x} \chi_i v_i - \sum_{i \in A_x} x_i
v_i$ is subgaussian since this equals $\sum_{i=1}^n L_x(\chi)_i v_i-t$.
Furthermore, since the support of such a support distribution lives in $W$, to
test subgaussianity we need only check the marginals $\pr{\theta}{\sum_{i \in
A_x} \chi_i v_i - \sum_{i \in A_x} x_i v_i}$ for $\theta \in W \cap S^{m-1}$. Thus,
we may assume that $W$ is the full space.  This completes the needed reductions.

\paragraph*{ Computational Model}
To formalize how our algorithms interact with convex bodies, we will use the
following computational model. 

To interact algorithmically with a convex body $K \subseteq \R^n$, we will
assume that $K$ is presented by a membership oracle.  Here a membership oracle
$O_K$ on input $x \in \R^n$, outputs $1$ if $x \in K$ and $0$ otherwise.
Interestingly, since we will always assume that our convex bodies have Gaussian
measure at least $1/2$, we will not need any additional centering (known point
inside $K$) or well-boundedness (inner contained and outer containing ball)
guarantees.

The runtimes of our algorithms will be measured by the number of oracle calls
and arithmetic operations they perform. We note that we use a simple model of
real computation here, where we assume that our algorithms can perform standard
operations on real numbers (multiplication, division, addition, etc.) in
constant time.  

%

\section{Banaszczyk's Theorem and Subgaussian Distributions}
\label{sec:subgaussian}

In this section, we give the main equivalences between Banaszczyk's vector
balancing theorem and the existence of subgaussian coloring distributions.

The fundamental theorem which underlies these equivalences is known as
Talagrand's majorizing measure theorem, which provides a nearly tight
characterization of the supremum of any Gaussian process using chaining
techniques. We now state an essential consequence of this theorem, which will be
sufficient for our purposes. For a reference, see~\cite{Talagrand05}.

\begin{theorem}[Talagrand]
\label{thm:talagrand}
Let $K \subseteq \R^n$ be a $0$-centered convex body and $Y \in \R^n$ be an
$s$-subgaussian random vector. Then for $X \in \R^n$ the $n$-dimensional
standard Gaussian, we have that
\[
\E[\|Y\|] \leq s \cdot C_T \cdot \E[\|X\|] \text{ ,}
\]
where $C_T > 0$ is an absolute constant.
\end{theorem}

As a consequence of the above theorem together with geometric estimates proved
in subsection~\ref{sec:geom-est}, we derive the following lemma, which will be
crucial to our equivalences and reductions.

\begin{lemma}[Reduction to Subgaussianity]
\label{lem:subg-red}
Let $Y \in \R^n$ be $s$-subgaussian. Then,
\begin{enumerate}
\item\label{lem:subg-red-sym} If $K \subseteq \R^n$ is a symmetric convex body
with $\gamma_n(K) \geq 1/2$, then 
\[
\E[\|Y\|_K] \leq 1.5 \cdot C_T \cdot s \text{.}
\]
In particular, $\Pr[Y \in 3 \cdot C_T \cdot s K] \geq 1/2$. 
\item\label{lem:subg-red-asym} If $K \subseteq \R^n$ is a convex body with
$\gamma_n(K) \geq 1/2$ and $\|b(K)\|_2 \leq \frac{1}{32\sqrt{2\pi}}$, then 
\[
\E[\|Y\|_{K \cap -K}] \leq 2(1+\pi\sqrt{8 \ln 2}) \cdot C_T \cdot s \text{.}
\]
In particular, $\Pr[Y \in 4 (1+\pi\sqrt{8 \ln 2}) \cdot C_T \cdot s (K \cap -K)] \geq 1/2$.
\end{enumerate}
\end{lemma}

\begin{proof}[Proof of Lemma~\ref{lem:subg-red}]
The proof follows immediately by combining
Lemmas~\ref{lem:med-to-exp-sym},~\ref{lem:med-to-exp} and
Theorem~\ref{thm:talagrand}. We note that the lower bounds on the probabilities
follow directly by Markov's inequality. 
\end{proof}

To state our equivalence, we will need the definitions of the following
geometric parameters.

\begin{definition}[Geometric Parameters]
\label{def:balance}
Let $T \subseteq \R^n$ be a finite set. 
\begin{itemize}
\item Define $s_g(T) > 0$ to be least number $s > 0$ such that there exists
an $s$-subgaussian random vector $Y$ supported on $T$.
\item Define $s_b(T) > 0$ to be the least number $s > 0$ such that for any
symmetric convex body $K \subseteq \R^n$, $\gamma_n(K) \geq 1/2$,  $T \cap s K
\neq \emptyset$. 
\end{itemize}
\end{definition}

We now state our main equivalence, which gives a quantitative version of
Theorem~\ref{thm:main-equiv} in the introduction. 

\begin{theorem}
\label{thm:bal-equiv}
For $T \subseteq \R^n$ be a finite set, the following holds:
\begin{enumerate}
\item $s_b(T) \leq 1.5 C_T \cdot s_g(T)$.
\item $s_g(T) \leq \sqrt{2} \cdot s_b(T)$.
\end{enumerate}
\end{theorem}

Using the above language, we can restate Banaszczyk's vector balancing theorem
restricted to \emph{symmetric convex bodies} as follows:

\begin{theorem}[\cite{Bana98}]
\label{thm:banaszczyk}
Let $v_1,\dots,v_m \in \R^n$. Then $s_b(\sum_{i=1}^n \set{-v_i,v_i})
\leq 5 \max_{i \in [m]} \|v_i\|_2$.
\end{theorem}

As an immediate corollary of Theorems~\ref{thm:bal-equiv}
and~\ref{thm:banaszczyk} (extended to the linear setting) we deduce: 

\begin{corollary}
\label{cor:subg-exist}
Let $v_1,\dots,v_m \in \R^n$. Then $s_g(\sum_{i=1}^n \set{-v_i,v_i}) \leq
\sqrt{2} \cdot 5 \max_{i \in [m]} \|v_i\|_2$. Furthermore, for $t \in
\sum_{i=1}^m [-v_i,v_i]$, $s_g(\sum_{i=1}^n \set{-v_i,v_i}-t) \leq \sqrt{2}
\cdot 10 \max_{i \in [m]} \|v_i\|_2$.
\end{corollary}

As explained in the introduction, the above equivalence shows the existence of a
\emph{universal sampler} for recovering Banaszczyk's vector balancing theorem for
symmetric convex bodies up to a constant factor in the length of the vectors.
Precisely, this follows directly from Lemma~\ref{lem:subg-red}
part~\ref{lem:subg-red-sym} and Corollary~\ref{cor:subg-exist} (for more details
see the proof of Theorem~\ref{thm:bal-equiv} below).


The following theorem, which we will need, is the classical minimax principle of
Von-Neumann.

\begin{theorem}[Minimax Theorem~\cite{Neumann28}]
\label{thm:minimax}
Let $X \subseteq \R^n$, $Y \subseteq \R^m$ be compact convex sets. Let $f: X
\times Y \rightarrow \R$ be a continuous function such that
\begin{enumerate}
\item $f(\cdot, y): X \rightarrow \R$ is convex for fixed $y \in Y$.
\item $f(x,\cdot): Y \rightarrow \R$ is concave for fixed $x \in X$.
\end{enumerate}
Then,
\[
\min_{x \in X} \max_{y \in Y} f(x,y) = \max_{y \in
Y} \min_{x \in X} f(x,y) \text{ .}
\]
\end{theorem}

We now proceed to the proof of Theorem~\ref{thm:bal-equiv}. 

\begin{proof}[Proof of Theorem~\ref{thm:bal-equiv}]~

\paragraph*{{\bf Proof of 1:}} Let $Y \in T$ be the $s_g(T)$-subgaussian random
variable. Let $K \subseteq \R^n$ be a symmetric convex body such that
$\gamma_n(K) \geq 1/2$. By Lemma~\ref{lem:subg-red} part~\ref{lem:subg-red-sym},
we have that
\[
\E[\|Y\|_K] \leq 1.5 C_T \cdot s_g(T) \text{. }
\]
Thus, there exists $x \in T$ such that $x \in 1.5 C_T \cdot s_g(T) K$.  Since
this holds for all such $K$, we have that $s_b(T) \leq 1.5 C_T \cdot s_g(T)$ as
needed.

\paragraph*{{\bf Proof of 2:}}

Recall the definition of $\cosh(x) = \frac{1}{2} (e^{x} + e^{-x})$ for $x \in
\R$. Note that $\cosh$ is convex, symmetric ($\cosh(x) = \cosh(-x)$), and
non-negative. For $w \in \R^n$, define $g_{w}: \R^n \rightarrow
\R_{\geq 0}$ by $g_{w}(x) =
\cosh(\pr{x}{w})/e^{\|w\|_2^2/2}$. By
Lemma~\ref{lem:mom-to-subg}, note that $\E[g_{w}(X)] = 1$ for $X$ an
$n$-dimensional standard Gaussian.

Let $\mathcal{D}$ denote the set of probability distributions on $T$.  Our goal
is to show that there exists $D \in \mathcal{D}$ such that $Y \sim D$ is
$\sqrt{2} \cdot s_b(T)$-subgaussian. By homogeneity, we may replace $T$ by
$T/s_b(T)$, and thus assume that $s_b(T)=1$. To show the existence of the
subgaussian distribution, we will show that
\begin{equation}
\label{eq:game-value}
\inf_{D \in \mathcal{D}} \sup_{w \in \R^n} \E_{Y \sim D}[g_{w}(Y)] \leq 2 \text{ .}
\end{equation}

Before proving the bound (\ref{eq:game-value}), we show that this suffices to
show the existence of the desired $\sqrt{2}$-subgaussian distribution. Let $D^*
\in \mathcal{D}$ denote the minimizing distribution for (\ref{eq:game-value}).
Then by definition of $g_{w}$, we have that
\begin{equation}
\label{eq:bal-equiv-1}
\E_{Y \sim D^*}[\cosh(\pr{w}{Y})] \leq 2 e^{\|w\|_2^2/2}
\quad \forall~w \in \R^n \text{ .}
\end{equation}
With the bounds on the Laplace transform in (\ref{eq:bal-equiv-1}), by
Lemma~\ref{lem:mom-to-subg} with $\beta = 2$ and $\sigma = 1$, we have that
$Y$ is $\sqrt{\log_2 2 + 1} = \sqrt{2}$-subgaussian as needed.

We now prove the estimate in (\ref{eq:game-value}). Let $C$ denote the closed
convex hull of the functions $g_w$. More precisely, $C$ is the closure of the
set of functions
\[
\set{f:T \rightarrow \R_{\geq 0} | f = \sum_{i=1}^k \lambda_i g_{w_i}, k \in \N,
\sum_{i=1}^k \lambda_i = 1, \lambda_i \geq 0 ~\forall~i \in [k], w_i
\in \R^n ~\forall~i \in [k]} \text{.}
\]
By continuity, we clearly have that
\begin{equation}
\label{eq:game-value-2}
\inf_{D \in \mathcal{D}} \sup_{w \in \R^n} \E_{Y \sim
D}[g_{w}(Y)] =
\inf_{D \in \mathcal{D}} \sup_{f \in C} \E_{Y \sim D}[f(Y)]
\text{ .}
\end{equation}
The strategy will now be to apply the minimax theorem~\ref{thm:minimax} to
(\ref{eq:game-value-2}). For this to hold, we first need that both $\mathcal{D}$
and $C$ are both convex and compact. This is clear for $\mathcal{D}$, since
$\mathcal{D}$ can be associated with the standard simplex in $\R^{|T|}$. By
construction $C$ is also convex, hence we need only prove compactness. Since $T$
is finite and $C$ is a closed subset of non-negative functions on $T$, $C$ can be
associated in the natural way with a closed subset of $\R^{|T|}_{\geq 0}$. To
show compactness, it suffices to show that this set is bounded. In particular,
it suffices to show that for $f \in C$, $\max_{x \in T} f(x) \leq M$
for some universal constant $M < \infty$. Since every $f \in C$ is a limit of convex
combinations of the functions $g_{w}$, $w \in \R^n$, it suffices to show
that
\[
\sup_{w \in \R^n} \max_{x \in T} g_{w}(x) \leq \infty
\text{ .}
\]
We prove this with the following computation:
\begin{align*}
\sup_{w \in \R^n} \max_{x \in T} g_{w}(x)
&= \sup_{w \in \R^n} \max_{x \in T}
\frac{\cosh(\pr{w}{x})}{e^{\|w\|_2^2/2}} \\
&\leq \sup_{w \in \R^n} \frac{\cosh(\max_{x \in T} \|x\| \|w\|_2)}
{e^{\|w\|_2^2/2}} \\
&\leq \sup_{w \in \R^n}
e^{\left(\max_{x \in T} \|x\|_2\right) \|w\|_2 -\|w\|_2^2/2} \\
&= e^{\left(\max_{x \in T} \|x\|_2\right)^2/2} < \infty \text{. }
\end{align*}
Thus $C$ is compact as needed. Lastly, note that the function $\E_{Y \sim
D}[f(Y)]$ from $\mathcal{D} \times C$ is bilinear, and hence is both continuous
and satisfies (trivially) the convexity-concavity conditions in
Theorem~\ref{thm:minimax}.

By compactness of $\mathcal{D}$ and $C$ and continuity, we have that
\[
\inf_{D \in \mathcal{D}} \sup_{f \in C} \E_{Y \sim D}[f(Y)] =
\min_{D \in \mathcal{D}} \max_{f \in C} \E_{Y \sim D}[f(Y)] \text{ .}
\]
Next, by the minimax theorem~\ref{thm:minimax}, we have that
\begin{align*}
\min_{D \in \mathcal{D}} \max_{f \in C} \E_{Y \sim D}[f(Y)]
&= \max_{f \in C} \min_{D \in \mathcal{D}} \E_{Y \sim D}[f(Y)]
= \max_{f \in C} \min_{x \in T} f(x) \\
&= \sup_{\substack{f = \sum_{i=1}^k \lambda_i g_{w_i} \\ \sum_{i=1}^k \lambda_i
= 1, ~\lambda_i \geq 0 ~\forall i \in [k] \\ w_i \in \R^n~\forall~i \in [k],~ k \in \N}} \min_{x \in
T} f(x) \text{ .}
\end{align*}
Take $f = \sum_{i=1}^k \lambda_i g_{w_i}$ as above (now as a function on
$\R^n$) and let $K_f = \set{x \in \R^n: f(x) \leq 2}$. Our task now
reduces to showing that $\exists x \in T$ such that $x \in K_f$.
Since $s_b(T) \leq 1$, it suffices to show that $\gamma_n(K_f) =
\Pr[X \in K_f] \geq 1/2$, for $X$ the $n$-dimensional standard Gaussian, and that
$K_f$ is symmetric and convex. Since $f$ is a convex combination of symmetric and
convex functions, it follows that $K_f$ is symmetric and convex. Since $f$ is
non-negative, by Markov's inequality
\[
\Pr[X \notin K_f] = \Pr[f(X) \geq 2] \leq \frac{\E[f(X)]}{2} = \frac{1}{2}
\sum_{i=1}^k \lambda_i \E[g_{w_i}(X)] = \frac{1}{2} \sum_{i=1}^k \lambda_i =
\frac{1}{2} \text{.}
\]
Hence $\Pr[X \in K_f] = 1-\Pr[X \notin K_f] \geq 1/2$, as needed.
\end{proof}

\section{Analysis of the Recentering Procedure}

We now give the crucial tool to reduce the asymmetric setting to the symmetric
setting, namely, the recentering procedure corresponding to
Theorem~\ref{thm:recenter} in the introduction. In the next subsection
(subsection~\ref{sec:red-asym-to-sym-formal}), we detail how to use this
procedure to yield the desired reduction. 

\begin{proof}[Proof of Theorem~\ref{thm:recenter} (Recentering Procedure)]
We first recall the desired guarantees. For linearly independent vectors
$v_1,\dots,v_n \in \R^n$, a shift $t \in \sum_{i=1}^n [-v_i,v_i]$, and a convex
body $K \subseteq \R^n$ of Gaussian measure at least $1/2$, we would like to find
a fractional coloring $x \in [-1,1]^n$, such that for $p = \sum_{i=1}^n x_i v_i
- t$ and the subspace $W = {\rm span}(v_i: i \in A_x)$, the following holds:

\begin{enumerate}
\item $p \in K$.
\item $\gamma_{|A_x|}((K-p) \cap W) \geq \gamma_n(K)$.
\item $b((K-p) \cap W) = 0$.
\end{enumerate}

We shall prove this by induction on $n$. Note that the base case $n=0$, reduces
to the statement that $0 \in K$, which is trivial. 

For a fractional coloring $x \in [-1,1]^n$, we remember first that $A_x$ denotes
the set of fractional coordinates and that $L_x: [-1,1]^{A_x} \rightarrow
[-1,1]^n$ is the lifting function (see Definition~\ref{def:lifting} for
details).

Let $P = \sum_{i = 1}^n [-v_i,v_i] - t$. Define the function $f(p) = \log
\gamma_n(K-p)$ for $p \in P$. Compute the maximizer $p$ of $f$ over $P$. Let $x
\in [-1,1]^n$ satisfy that $p = \sum_{i=1}^n x_i v_i - t$ and let $W = {\rm
span}(v_i: i \in A_x)$. Note first that since $\gamma_n(K-p) \geq \gamma_n(K)
\geq 1/2$, by Lemma~\ref{lem:ball-cont} part~\ref{lem:ball-cont-big} we have
that $0 \in K-p \Rightarrow p \in K$.

Assume first that $p$ is in the interior of $P$. Then, since $p$ is a maximizer
and does not touch the boundary of $P$, by the KKT conditions we must have that
$\nabla f(p) = 0$. From here, direct computation reveals that $\nabla f(p) =
b(K-p)$. Again, since $y$ does not touch again constraints of $P$, we see that
$A_x = [n]$ and hence $W = \R^n$. Thus, as claimed, $x$ satisfies the conditions
of the theorem.

Assume now that $p \in \partial P$. From here, we must have that $|A_x| < n$
and hence $\dim(W) = |A_x| < n$. Next, by Lemma~\ref{lem:origincut}, we see that 
\[
\gamma_{|A_x|}((K-p) \cap W) \geq \gamma_n(K-p) \geq \gamma_n(K) \geq 1/2 \text{ .}
\]
By Lemma~\ref{lem:partial-red} part~\ref{lem:partial-red-stat}, 
for $z \in [-1,1]^{A_x}$, we have that
\[
\sum_{i \in A_x} z_i v_i - \sum_{i \in A_x} x_i v_i \in K-p
\Leftrightarrow \sum_{i=1}^n L_x(z)_i v_i - t \in K \text{ .}
\]
Thus, we may apply induction on the vectors $(v_i: i \in A_x)$, the shift
$\sum_{i \in A_x} x_i v_i$ and convex body $(K-p) \cap W$, and recover $z \in
[-1,1]^{A_x}$, such that for $W_z = {\rm span}(v_i: i \in A_z)$, we get 
\begin{equation}
\label{eq:rec-eq-1}
\sum_{i \in A_x} z_i v_i - \sum_{i \in A_x} x_i v_i \eqdef p_z \in K-p
\text{ ,}
\end{equation}
and 
\begin{equation}
\label{eq:rec-eq-2}
\gamma_{|A_z|}((K-p-p_z) \cap W_z) \geq \gamma_{|A_x|}(K-p) \text{ ,}
\end{equation}
and 
\begin{equation}
\label{eq:rec-eq-3}
b((K-p-p_z) \cap W_z) = 0 \text{ .}
\end{equation}
We now claim that $w = L_x(z)$ satisfies the conditions of the theorem.
To see this, note that by Lemma~\ref{lem:partial-red} part~\ref{lem:partial-red-id},
\[
\sum_{i=1}^n w_i v_i - t = \sum_{i=1}^n L_x(z)_i v_i - t = 
\sum_{i \in A_x} z_i v_i - \sum_{i \in A_x} x_i v_i + p = p_z+p \text{ .}
\]
Furthermore, since $p_z \in K-p$ we have that $p_z+p \in K$. Next, clearly
$A_{w} = A_{z}$ and hence ${\rm span}(v_i: i \in A_w) = W_z$. The claim
thus follows by
combining~\eqref{eq:rec-eq-1},~\eqref{eq:rec-eq-2},~\eqref{eq:rec-eq-3}.

\end{proof}

\subsection{Reduction from Asymmetric to Symmetric Convex Bodies}
\label{sec:red-asym-to-sym-formal}

As explained in the introduction, the recentering procedure allows us to
reduce Banaszczyk's vector balancing theorem for all convex bodies to the
symmetric case, and in particular, to the task of subgaussian sampling. 
We give this reduction in detail below.

Let $v_1,\dots,v_n$, $t$, and $K$ be as in Theorem~\ref{thm:recenter}, and let
$x \in [-1,1]^n$ the fractional coloring guaranteed by the recentering
procedure. As in Theorem~\ref{thm:recenter}, let $p = \sum_{i=1}^n x_i v_i
- t$ and $W = {\rm span}(v_i: i \in A_x)$. We shall now assume that $\max_{i \in
  [n]} \|v_i\|_2 \leq c$, a constant to be chosen later.  From here, by
Lemma~\ref{lem:partial-red} part~\ref{lem:partial-red-stat}, for $\chi \in
\set{-1,1}^{A_x}$,  
\[
\sum_{i \in A_x} \chi_i v_i - \sum_{i \in A_x} x_i v_i \in (K-p) \cap W
\Leftrightarrow \sum_{i \in [n]} L_x(\chi)_i v_i - t \in K \text{ .}
\]
Let $C = (K-p) \cap W$ and $d = |A_x|$. By the guarantees on the recentering
procedure, we know that $\gamma_d(C) \geq 1/2$ and $b(C) = 0$. Then by
Lemma~\ref{lem:med-to-exp} in section~\ref{sec:geom-est}, for $X \in W$ the 
$d$-dimensional standard Gaussian on $W$, we have that
\[
\E[\|X\|_{C \cap -C}] \leq 2\E[\|X\|_C] \leq 2(1+\pi\sqrt{8 \ln 2}) \text{ .}
\]
Hence by Markov's inequality, $\Pr[X \in 4(1+\pi\sqrt{8 \ln 2})(C \cap -C))]
\geq 1/2$. At this point, using Banaszczyk's theorem in the linear setting for
symmetric bodies (which loses a factor of $2$), if the $\ell_2$ norm bound $c$
satisfies 
\[
1/c \geq 10 \cdot 4(1+\pi\sqrt{8 \ln 2}) \text{ ,}
\]
then by homogeneity there exists $\chi \in \set{-1,1}^{A_x}$ such that 
\[
\sum_{i \in A_x} \chi_i v_i - \sum_{i \in A_x} x_i v_i \in C \cap -C 
\Rightarrow \sum_{i=1}^n L_x(\chi)_i v_i - t \in K \text{ .}
\]
Hence, the reduction to the symmetric case follows. 

We can also achieve the same with a subgaussian sampler, though the vectors
should be shorter. In particular, applying corollary~\ref{cor:subg-exist}, if
\[
1/c \geq \sqrt{2} \cdot 10 \cdot C_T \cdot 4(1+\pi\sqrt{8 \ln 2}) \text{,}
\]
then there exists a distribution on colorings $\chi \in \set{-1,1}^{A_x}$ such
that $\sum_{i \in A_x} \chi_i v_i - \sum_{i \in A_x} x_i v_i$ is $(C_T \cdot
4(1+\pi\sqrt{8 \ln 2}))^{-1}$-subgaussian. From here, by
Lemma~\ref{lem:subg-red} part~\ref{lem:subg-red-asym} applied to $C$, 
\[
\Pr\left[(\sum_{i \in A_x} \chi_i v_i - \sum_{i \in A_x} x_i v_i) \in
C \cap -C\right] \geq 1/2 \text{ ,}
\]
as needed.

\subsection{Geometric Estimates}
\label{sec:geom-est}

In this section, we present the required estimates for the proof of
Lemma~\ref{lem:subg-red}.

The following theorem of Lata{\l}a and Oleszkiewicz will allow us to translate
bounds on Gaussian measure to bounds on Gaussian norm expectations.

\begin{theorem}[\cite{LO99}]
\label{thm:gaussian-dilation}
Let $X \in \R^n$ be a standard $n$-dimensional Gaussian. Let $K \subseteq \R^n$ be a
symmetric convex body, and let $\alpha \geq 0$ be chosen such that $\Pr[X \in K]
= \Pr[|X_1| \leq \alpha]$. Then the following holds:
\begin{enumerate}
\item For $t \in [0,1]$, $\Pr[X \in tK] \leq \Pr[|X_1| \leq t\alpha]$.
\item For $t \geq 1$, $\Pr[X \in tK] \geq \Pr[[X_1| \leq t\alpha]$.
\end{enumerate}
\end{theorem}

Using the above theorem we derive can derive bound goods bounds on Gaussian norm
expectations. We note that much weaker and more elementary estimates than those
given in~\ref{thm:gaussian-dilation} would suffice (e.g.~Borell's inequality),
however we use the stronger theorem to achieve a better constant.

\begin{lemma}
\label{lem:med-to-exp-sym} Let $X \in \R^n$ be a standard $n$-dimensional Gaussian.
Let $K \subseteq \R^n$ be a symmetric convex body such that $\Pr[X \in K] \geq
1/2$. Then $\E[\|X\|_K] \leq 1.5$.
\end{lemma}
\begin{proof}
Let $m > 0$ satisfy $\Pr[X \in mK] = \Pr[\|X\|_K \leq m] = \frac{1}{2}$.
Note that $m \leq 1$ by our assumption on $K$. Let $\alpha > 0$ denote the
number such that $\Pr[|X_1| \leq \alpha] = \gamma_1([-\alpha,\alpha]) =
\frac{1}{2}$. Here a numerical calculation reveals $\alpha \geq .67$.

By Theorem \ref{thm:gaussian-dilation}, we have that $\Pr[\|X\|_K \geq tm] \leq
\Pr[|X_1| \geq t\alpha]$ for $t \geq 1$. Thus,
\begin{align*}
\E[\|X\|_K] &= \int_0^\infty \Pr[\|X\|_K \geq t] dt \leq m + \int_m^\infty \Pr[\|X\|_K \geq t] dt \leq m + \int_m^\infty
\Pr[|X_1| \geq \frac{\alpha t}{m}]dt \\ &= (1 + \frac{1}{\alpha}
\int_{\alpha}^\infty \Pr[|X_1| \geq t] dt)m \leq 1 + \frac{1}{\alpha}
\int_{\alpha}^\infty \Pr[|X_1| \geq t] dt \\
&= 1 + \frac{1}{\alpha}
\int_{\alpha}^\infty \frac{2}{\sqrt{2\pi}}(t-\alpha)e^{-t^2/2} dt
= 1 + \frac{1}{\alpha}(\sqrt{\frac{2}{\pi}} e^{-\alpha^2/2} - \alpha/2)
= 1/2 + \sqrt{\frac{2}{\pi}}\frac{e^{-\alpha^2/2}}{\alpha} \\
&\leq 1/2 + \sqrt{\frac{2}{\pi}} \frac{e^{-(.67)^2/2}}{.67} \leq 1.5
\end{align*}
\end{proof}

The following lemma shows that we can find a large ball in $K$ centered around
the origin, if either its Gaussian mass is large or its barycenter is close to
the origin.

\begin{lemma}
\label{lem:ball-cont}
Let $K \subseteq \R^n$ be a convex body. Then the following holds:
\begin{enumerate}
\item \label{lem:ball-cont-big} If for $r \geq 0$, $\gamma_1((-\infty,r]) \leq \gamma_n(K)$, then
$r B_2^n \subseteq K$. \\
In particular, if $\gamma_n(K) = 1/2+\eps$, for $\eps
\geq 0$, this holds for $r = \sqrt{2\pi}\eps$.
\item \label{lem:ball-cont-bar} If for $0 \leq r \leq 1/2$, $\gamma_1([-2r,2r]) \leq \gamma_n(K)$ and
$\|b(K)\|_2 \gamma_n(K) \leq \frac{1}{\sqrt{2\pi}} \cdot \frac{r^2}{2}$, then $r
B_2^n \subseteq K$. In particular, this holds for $r=1/4$ if $\gamma_n(K) \geq
1/2$ and $\|b(K)\|_2 \leq \frac{1}{32\sqrt{2\pi}}$. 
\end{enumerate}
\end{lemma}
\begin{proof}
We begin with part (1). Assume for the sake of contradiction that there exist $x
\in \R^n$, $\|x\|_2 \leq r$, such that $x \notin K$. Then, by the separator
theorem there exists a unit vector $\theta \in S^{n-1}$, such that $\max_{z \in
K} \pr{\theta}{z} < \pr{\theta}{x}$. In particular, $K$ is strictly contained in
the halfspace $H = \set{z \in \R^n: \pr{\theta}{z} \leq g}$ where $g =
\pr{\theta}{x}$. Thus $\gamma_n(K) < \gamma_n(H) = \gamma_1((-\infty,g])$. But
note that by Cauchy-Schwarz $g \leq \|\theta\|_2 \|x\|_2 = r$, a clear
contradiction to the assumption on $r$. 

For the furthermore, we first see that 
\[
\frac{1}{\sqrt{2\pi}} \int_0^{\sqrt{2\pi} \eps} e^{-x^2/2} dx \leq \frac{(\sqrt{2\pi}
\eps)}{\sqrt{2\pi}} = \eps \text{ .}
\]
Hence, $\gamma_1((-\infty,\sqrt{2\pi} \eps]) \leq 1/2 + \eps = \gamma_n(K)$, as
needed.

We now prove part (2). Similarly to the above, if $K$ does not contain a ball of
radius $r$, then there exists a halfspace $H = \set{z \in \R^n:
\pr{\theta}{z} \leq r-\eps}$, for some $0 < \eps \leq r$, such that $K \subseteq H$.
By rotational invariance of the Gaussian measure, we may assume that $\theta =
e_1$. Now let $K_t = \set{x \in K: x_1 = t}$ and let $f(t) =
\gamma_{n-1}(K_t-te_1)$, where clearly $f(t) \in [0,1]$. From here, we see that 
\begin{equation}
\label{eq:ball-cont-1}
\begin{split}
\int_{-\infty}^{r-\eps} t \frac{e^{-t^2/2}}{\sqrt{2\pi}}
\frac{f(t)}{\int_{-\infty}^{r-\eps} \frac{e^{-t^2/2}}{\sqrt{2\pi}} f(t) dt} dt &= 
\int_{-\infty}^{r-\eps} t \frac{e^{-t^2/2}}{\sqrt{2\pi}} 
\frac{\gamma_{n-1}(K_t-te_1)}{\gamma_n(K)} dt \\
&= b(K)_1 \geq -\|b(K)\|_2\text{ .}
\end{split}
\end{equation}
Thus to get a contradiction, it suffices to show that
\begin{equation}
\label{eq:ball-cont-3}
\int_{-\infty}^{r-\eps} t \frac{e^{-t^2/2}}{\sqrt{2\pi}} f(t) dt < - 
\gamma_n(K) \|b(K)\|_2 \text{ ,}
\end{equation}
for any function $f:(-\infty,r-\eps] \rightarrow [0,1]$ satisfying
\begin{equation}
\label{eq:ball-cont-2}
\int_{-\infty}^{r-\eps} \frac{e^{-t^2/2}}{\sqrt{2\pi}} f(t) dt =
\gamma_n(K) \text{.}
\end{equation}
From here, it is easy to see that the function $f$ maximizing the left hand
side of~\eqref{eq:ball-cont-3} satisfying~\eqref{eq:ball-cont-2} must be the
indicator function of an interval with right end point $r-\eps$, i.e.~the
function $f$ which pushes mass ``as far to the right'' as possible. Now let $l
\leq r-\eps$ denote the unique number such that $\gamma_1([l,r-\eps]) =
\gamma_n(K)$, noting that the optimizing $f$ is now the indicator function of
$[l,r-\eps]$. From here, a direct computation reveals that
\[
\int_l^{r-\eps} t \frac{e^{-t^2/2}}{\sqrt{2\pi}} =
\frac{1}{\sqrt{2\pi}}(e^{-l^2/2} - e^{-(r-\eps)^2/2}) <
\frac{1}{\sqrt{2\pi}}(e^{-l^2/2} - e^{-r^2/2}) \text{ .}
\]
Now since $\gamma_1([l,r]) > \gamma_n(K) \geq \gamma_1([-2r,2r])$, we must have
that $l \leq -2r$. Using the inequalities $1+x \leq e^x \leq 1+x+x^2$, for $|x|
\leq 1/2$, we have that
\begin{align*}
\frac{1}{\sqrt{2\pi}}(e^{-l^2/2} - e^{-r^2/2}) 
&\leq \frac{1}{\sqrt{2\pi}}(e^{-(2r)^2/2} - e^{-r^2/2}) \\
&\leq \frac{1}{\sqrt{2\pi}}(1-(2r)^2/2+(2r)^4/4-(1-r^2/2)) \\
&= \frac{1}{\sqrt{2\pi}}((4r^2)r^2-3r^2/2) 
\leq - \frac{1}{\sqrt{2\pi}} \cdot r^2/2 \text{ .} \left( \text{ since $r \leq
1/2$ } \right)
\end{align*}
But by assumption $-\frac{1}{\sqrt{2\pi}} \cdot r^2/2 \leq
-\gamma_n(K)\|b(K)\|_2$, yielding the desired contradiction.

For the furthermore, it follows by a direct numerical computation.
\end{proof}

We will now extend the bound to asymmetric convex bodies having their barycenter
near the origin. To do this, we will need the standard fact that the gauge
function of a body is Lipschitz when it contains a large ball. We recall that a
function $f: \R^n \rightarrow \R$ is $L$-Lipschitz if for $x,y \in \R^n$,
$|f(x)-f(y)| \leq L \|x-y\|_2$.

\begin{lemma}
\label{lem:lipsch}
Let $K\subseteq \mathbb{R}^n$ be a convex body satisfying $r B_2^n \subseteq K$
for some $r>0$.  Then, the gauge function $\|\cdot\|_K: \R^n \rightarrow \R_+$
of $K$ is $(1/r)$-Lipschitz.
\end{lemma}
\begin{proof}
We need to show
\[ 
\left| \|x\|_K-\|y\|_K \right| \le \frac{1}{r}\|x-y\|_2 \quad \textrm{ for all
} x,y\in\mathbb{R}^n \text{ .}
\]
To see this, note that
\begin{eqnarray*}
\|x\|_K & = & \|(x-y)+y\|_K \\
& \le & \| x-y\|_K+\|y\|_K \qquad\textrm{(by triangle inequality)}\\
& \le & \| x-y\|_{r B_2^n}+\|y\|_K \qquad\textrm{(since } r B_2^n \subseteq K \textrm{ )} \\ 
& = & \frac{1}{r}\| x-y\|_2+\| y\|_K \text{ ,}
\end{eqnarray*}
yielding $\|x\|_K-\|y\|_K \leq \frac{1}{r}\|x-y\|_2$. The other inequality 
follows by switching $x$ and $y$. 
\end{proof}

We will also need the following concentration inequality of Maurey and Pisier.

\begin{theorem}[Maurey-Pisier] 
\label{thm:gaussian-con}
Let $f:\R^n \rightarrow \R$ be an $L$-Lipschitz function. Then for $X \in \R^n$
standard Gaussian and $t \geq 0$, we have the inequalities
\[
\Pr[f(X)-\E[f(X)] \geq tL] \leq e^{-\frac{2t^2}{\pi^2}} \quad \text{ and } \quad
\Pr[f(X)-\E[f(X)] \leq -tL] \leq e^{-\frac{2t^2}{\pi^2}} \text{ .}
\]
\end{theorem}

We now prove the main estimate for asymmetric convex bodies.

\begin{lemma}
\label{lem:med-to-exp} 
Let $K \subseteq \R^n$ be a $0$-centered convex body and $X \in \R^n$ be the
standard $n$-dimensional Gaussian. Then the following holds:
\begin{enumerate}
\item \label{lem:med-to-exp-1} $\E[\|X\|_{K \cap -K}] \leq 2\E[\|X\|_K]$.
\item \label{lem:med-to-exp-2} If $\gamma_n(K) \geq 1/2$ and $\|b(K)\|_2 \leq
\frac{1}{32\sqrt{2\pi}}$, then $\E[\|X\|_K] \leq (1+\pi\sqrt{8 \ln 2})$.
\end{enumerate}
\end{lemma}
\begin{proof}
We prove part (1). By symmetry of the Gaussian measure
\[
\E[\|X\|_{K \cap -K}] = \E[\max \set{\|X\|_K,\|-X\|_K}]
\leq \E[\|X\|_K + \|-X\|_K] = 2 \E[\|X\|_K] \text{ ,}
\]
as needed.

We prove part (2). Let $c = \pi \sqrt{8 \ln 2}$. First, by
Lemma~\ref{lem:ball-cont} part~\ref{lem:ball-cont-bar} and our assumptions on
$K$, we have that $(1/4) B_2^n \subseteq K$. Thus, by Lemma~\ref{lem:lipsch}
$\|\cdot\|_K$ is $4$-Lipschitz. Assume for the sake of contradiction that 
$\E[\|X\|_K] > 1+c$. Then, since 
\[
1/2 \leq \gamma_n(K) = \Pr[X \in K] = \Pr[\|X\|_K \leq 1] \text{ ,}
\]  
we must have that $\Pr[\|X\|_K-\E[\|X\|_K] \leq -c] > 1/2$. But by
Theorem~\ref{thm:gaussian-con} and the Lipschitz proporties of $\|\cdot\|_K$, 
\[
\Pr[\|X\|_K-\E[\|X\|_K] \leq -c] \leq e^{-\frac{2(c/4)^2}{\pi^2}} = 1/2 \text{ ,}
\]
a clear contradiction.
\end{proof}

\section{An $O(\sqrt{\log n})$-subgaussian Random Walk}
\label{sec:rand-walk}

The $O(\sqrt{\log n})$-subgaussian random walk algorithm is given as
Algorithm~\ref{alg:randwalk}. Step~\ref{alg-step:sdp} can be executed
in polynomial time by either calling an SDP solver, or executing the
algorithm from Section~\ref{sec:vector-komlos}. The feasibility of the
program is guaranteed by Theorem~\ref{thm:SDP-feas}, and also by the
results of~\cite{Nikolov13}. The matrix $U(t)$ in
step~\ref{alg-step:rand-step} can be computed by Cholesky decomposition.

\begin{algorithm}[h]
  \caption{$O(\sqrt{\log n})$-subgaussian Random Walk Algorithm}
  \label{alg:randwalk}
  \begin{algorithmic}
    \State \textbf{Input}: $v_1, \ldots, v_n \in \R^m$ such that
    $\max_{i = 1}^n\|v_i\|_2 \le 1$; a vector $y \in \sum_{i =
      1}^n[-v_i, v_i]$. 
    \State \textbf{Output}: random signs $\chi_1, \ldots, \chi_n \in
    \{-1, 1\}$ such that $\sum_{i = 1}^n{\chi_i v_i} - y$ is $O(\sqrt{\log
      n})$-subgaussian.
    \State Let $V = (v_i)_{i=1}^n$.
    \State Let $\gamma = \frac{2\log_2 2n}{n^{5/2}}$; $\delta =
    \frac{2\sqrt{2\ln 2 \log_2 2n}}{n}$; $T = \lceil 2/\gamma^2\rceil \cdot \lceil
    \log_2 2n \rceil$; 
    \State\label{alg-step:start}Let $\chi(0) \in [-1, 1]^n$ be such that $\sum_i \chi(0)_i
    v_i = y$.
    \State Let $A(1) = \{1, \ldots, n\}$
    \For{$t = 1, \ldots, T$}
    \State \label{alg-step:sdp} Compute $\Sigma(t) \in \R^{n\times n}$, $\Sigma(t) \succeq 0$, such that
    \begin{align*}
      \Sigma(t)_{ii} &= 1 \ \ \ \forall i \in A(t)\\
      \Sigma(t)_{ii} &= 0 \ \ \ \forall i \not \in A(t)\\
      V\Sigma(t)V^\T &\preceq I_m,
    \end{align*}
    \State Pick $r(t) \in \{-1, +1\}^{n}$ uniformly at random.
    \State \label {alg-step:rand-step} $\chi(t) = \chi(t-1) + \gamma U(t)r(t)$, where $U(t)U(t)^\transpose = \Sigma(t)$. 
    \State $A(t+1) = \{i: |\chi(t)_i| < 1-\delta\}$.
    \EndFor
    \State Set $\chi_i = \sign(\chi(T)_i)$ for all $i \in \{1, \ldots,
    n\}$. 
    \If{$A(T) = \emptyset$}
    \State Return $\chi$
    \Else
    \State Restart algorithm from line~\ref{alg-step:start}.
    \EndIf
  \end{algorithmic}
\end{algorithm}

Let us first make some observations that will be useful throughout the
analysis. Notice first that the random process $\chi(0), \ldots,
\chi(T)$ is Markovian. Let $u_i(t)$ be the $i$-th row of $U(t)$. By
the definition of $\Sigma(t)$ and $U(t)$, $\|u_i(t)\|_2 =
\Sigma(t)_{ii}$ equals $1$ if $i \in A(t)$, and $0$ otherwise. We have
$\chi(t)_i - \chi(t-1)_i = \gamma \langle u_i(t), r(t)\rangle$, and,
because $r(t)\in \{-1, 1\}^n$, by Cauchy-Schwarz we get
\begin{equation}
  \label{eq:step-bound}
  |\chi(t)_i - \chi(t-1)_i| \le 
  \begin{cases}
    \gamma \sqrt{n} &\text{ if }i \in A(t)\\
    0 &\text{ otherwise}
  \end{cases}.
\end{equation}

We first analyze the convergence of the
algorithm: we show that, with constant probability, the random walk
fixes all coordinates to have absolute value between $1-\delta$ and
$1$. First we prepare a lemma.

\begin{lemma}\label{lm:stopping-time}
  Let $X_0, X_1, X_2, \ldots$ form a martingale sequence adapted to
  the filtration $\{\mathcal{F}_t\}$ such that
  $X_0 \in [-1, 1]$, and for every $t \ge 1$ we have $\E[(X_t -
  X_{t-1})^2\mid \mathcal{F}_{t-1}] = \sigma^2$, and $|X_t - X_{t-1}|
  \le \delta$. Denote $\tau = \inf\{t: |X_t| \ge 1-\delta\}$. Then
  $\E[\tau] < \frac{1-X_0^2}{\sigma^2}$. 
\end{lemma}
\begin{proof}
  Define $Y_0, Y_1, Y_2, \ldots$ to be a martingale with respect to
  $X_0, X_1, X_2, \ldots$ defined by $Y_t = X_{\min\{t,
    \tau\}}$. Because $|X_t - X_{t-1}|<\delta$, we easily see by
  induction that $|Y_t| < 1$ for all $t \geq 0$. Therefore, for any $t
  \geq 1$, we compute
  \begin{align*}
    1 > \E Y_t^2 &= \sum_{s = 1}^t{\E\E[(Y_t - Y_{t-1})^2\mid
      \mathcal{F}_{t-1}]} + X_0^2\\ 
    &= \sum_{s=1}^t{\sigma^2\Pr[\tau\ge s]} + X_0^2\\
    &= \sigma^2 \E[\min\{t, \tau\}] + X_0^2.
  \end{align*}
  By the monotone convergence theorem, we have that $\sigma^2\E[\tau]
  < 1- X_0^2$, which was to be proved.
\end{proof}

The next lemma gives our convergence analysis of the random walk.

\begin{lemma}\label{lm:convergence}
  With probability $1$, $|\chi(t)_i| \le 1$ for all $1 \leq t \leq
  T$ and all $1 \le i \le n$. With probability at least $1/2$,
  $|\chi(T)_i| \ge 1-\delta$ for all $1 \le i \le n$.
\end{lemma}
\begin{proof}
  We prove the first claim by induction on $t$. It is clearly
  satisfied for $t = 0$; assume then that the claim holds up to $t-1$,
  and we will prove it for $t$. If $i \not \in A(t)$, then $\chi(t)_i
  = \chi(t-1)_i$ by \eqref{eq:step-bound}, and the claim follows by
  the inductive hypothesis. If $i \in A(t)$, then $|\chi_i(t)| <
  1-\delta$, and by \eqref{eq:step-bound} and the triangle inequality,
  we have $|\chi_i(t)| < 1-\delta + \gamma\sqrt{n} < 1$, where the
  final inequality follows because $\gamma\sqrt{n} < \delta$.

  To prove the second claim, we will show that $\Pr[|\chi(T)_i| <
  1-\delta] \le 1/2n$ holds for every $i$. The claim will then follow
  by a union bound. Let us fix an arbitrary $i \in \{1, \ldots,
  n\}$. Define $\Delta = \lceil 2/\gamma^2\rceil$, and let $E_j$, $0
  \le j \le (T-\Delta)/\Delta$ be the event that $|\chi(t)_i| <
  1-\delta$ for all $t \in [j\Delta + 1, (j+1)\Delta]$.  Observe that
  if $|\chi(t)_i| \ge 1-\delta$, then $|\chi(s)_i| = |\chi(t)_i| \ge
  1-\delta$ for all $t \le s \le T$, and, therefore $\chi(T)_i <
  1-\delta$ if and only if all the events $E_0, \ldots,
  E_{(T-\Delta)/\Delta}$ hold simultaneously. By this observation, and
  the Markov property, we have
  \begin{equation}
    \label{eq:phases}
    \Pr[|\chi(T)_i| <  1-\delta] = \prod_{j =
      0}^{T/\Delta - 1}{\Pr[E_j \mid \chi(j\Delta)]}.
  \end{equation}
  Let $\tau_j = (j+1)\Delta$ if $E_j$ holds and $\tau_j = \min\{t \ge
  j\Delta: |\chi(t)_i| \ge 1-\delta\}$, otherwise.
  The sequence $\chi(j\Delta)_i, \ldots, \chi(\tau_j)_i$, conditioned
  on $\chi(j\Delta)$, is a martingale. Moreover, since for any $t \in
  [j\Delta, \tau_j]$ we have $i \in A(t)$, for all such $t$ we get
  \begin{equation}
    \label{eq:variance}
    \E[(\chi(t)_i - \chi(t-1)_i)^2 \mid \chi(t-1)] =
    \gamma^2 \E |\langle u_i(t), r(t) \rangle|^2
    = \gamma^2\|u_i\|_2^2 =  \gamma^2 .
  \end{equation}
  By \eqref{eq:step-bound} and \eqref{eq:variance}, the sequence
  $\chi(j\Delta)_i, \ldots, \chi(\tau_j)_i$, conditioned on
  $\chi(j\Delta)$, satisfies the assumptions of
  Lemma~\ref{lm:stopping-time}. By the lemma, we have
  \[
  \E[\tau_j -  j\Delta \mid  \chi(j\Delta)] \le
  \frac{1-\chi(j\Delta)_i^2}{\gamma^2}\le \frac{1}{\gamma^2}.
  \] 
  Since the event $E_j$ holds only if $\tau_j \ge (j+1)\Delta$, by
  Markov's inequality we have \[\Pr[E_j \mid \chi(j\Delta)] \le
  \frac{1}{\gamma^{2}\Delta} \le \frac12.\] This bound and
  \eqref{eq:phases} imply that $\Pr[|\chi(T)_i| < 1-\delta] \le
  2^{-T/\Delta} \le 1/2n$, which was to be proved.
\end{proof}

To prove that the walk is subgaussian, we will need the following
martingale concentration inequality due to Freedman.

\begin{theorem}[\cite{Freedman71}]
\label{thm:freedman}
Let $Z_1,\ldots,Z_T$ be a martingale adapted to the filtration $\{\mathcal{F}_t\}$
$|Z_t - Z_{t-1}| \leq M$ for all $t$, and let $W_t = \sum_{j=1}^t
\E_{j-1}[(Z_j - Z_{j-1})^2\mid  \mathcal{F}_{j-1}] =
\sum_{j=1}^t\textrm{Var}[Z_j\mid  \mathcal{F}_{j-1}]$.
Then for all $\lambda \geq 0$  and $\sigma^2 \geq 0$, we have 
\[ \Pr[ \exists t \textrm{ s.t.}~|Z_t - Z_0| \geq \lambda \textrm{ and } W_t \leq \sigma^2]\le 2 \exp\left(-\frac{\lambda^2}{2 (\sigma^2 + M \lambda)} \right). \] 
\end{theorem}

Next we state the main lemma, which, together with an estimate
on the error due to rounding, implies subgaussianity.

\begin{lemma}\label{lm:subg}
  The random variable $\sum_{i = 1}^n{\chi(T)_iv_i} - y$ is
  $(\gamma\sqrt{2T})$-subgaussian.
\end{lemma}
\begin{proof}
  Define $Y_t = \sum_{i = 1}^n{\chi(t)_iv_i}$ for all $t = 1, \ldots,
  T$. Notice that $Y_0 = y$. Let us fix a $\theta \in S^{n-1}$ once
  and for all, and let $Z_t = \langle \theta, Y_t \rangle$ for $t = 0,
  \ldots, T$. We need to show that for every $\lambda > 0$, $\Pr[|Z_T|
  \ge \lambda] \le 2e^{-\lambda^2/2\sigma^2}$. We first observe that
  $Z_t$ is bounded, so we only need to consider $\lambda$ in a finite
  range. Indeed, by Lemma~\ref{lm:convergence}, $Y_t \in \sum_{i =
    1}^n{[-v_i, v_i]}$ with probability $1$, so by the triangle
  inequality, $\|Y_t\|_2 \le \sum_{i = 1}^n{\|v_i\|_2} \le n$. Then,
  by Cauchy-Schwarz, $|Z_t| \le n$ as well, and, therefore, $\Pr[|Z_T|
  > n] = 0$. For the rest of the proof we will assume that $0 <
  \lambda \le n$.

  Observe that $Z_0, \ldots, Z_T$ is a martingale. First we prove that
  the increments are bounded: this follows from the boundedness of the
  increments of the coordinates of $\chi(t)$. Indeed, by the triangle
  inequality and \eqref{eq:step-bound},
  \[
  \|Y_t- Y_{t-1}\|_2 =  \left\|\sum_{i = 1}^n{(\chi(t)_i -  \chi(t-1)_i)v_i}\right\|_2 
  \le \sum_{i = 1}^n{|\chi(t)_i -  \chi(t-1)_i| \|v_i\|_2}
  \le \gamma n^{3/2}.
  \]
  Then, it follows from Cauchy-Schwarz that $|Z_t - Z_{t-1}| \le
  \gamma n^{3/2}$. 

  Next we bound the variance of the increments. By the Markov property
  of the random walk, $Z_t - Z_{t-1}$ is entirely determined by
  $\chi_{t-1}$. Denoting $V=(v_i)_{i = 1}^n$ as in the description of
  Algorithm~\ref{alg:randwalk}, we have
  \begin{align*}
    \E[(Z_t - Z_{t-1})^2 \mid \chi_{t-1}] &= \theta^\transpose \E[(Y_t - Y_{t-1})(Y_t - Y_{t-1})^\transpose]
    \theta\\
    &= \theta^\transpose V \E[(\chi(t) - \chi(t-1))(\chi(t) - \chi(t-1))^\transpose] V^\transpose  \theta\\
    &= \gamma^2 \theta^\transpose V U(t)\E[r(t)r(t)^\transpose]
    U(t)^\transpose V^\transpose  \theta\\
    &= \gamma^2\theta^\transpose V \Sigma(t) V^\transpose  \theta\\
    &\le \gamma^2.
  \end{align*}
  The penultimate equality follows because $\E[r(t)r(t)^\transpose] = I_n$ and
  $U(t)U(t)^\transpose = \Sigma(t)$, and the final inequality follows because
  $\Sigma(t)$ was chosen so that $V \Sigma(t) V^\transpose \preceq I_m$. 

  We are now ready to apply Theorem~\ref{thm:freedman}. Using the
  notation from the theorem, we have shown that $M \le \gamma
  n^{3/2}$, and that $W_t \le \gamma^2t$ for all $t$, and both bounds
  hold with probability $1$. Let $\sigma^2 = \gamma^2 T$. First we
  claim that for any $\lambda \le n$, $M\lambda \le \sigma^2$. Indeed,
  \[
  M\lambda \le \gamma n^{5/2} \le 2\log_2(2n) \le \gamma^2 T = \sigma^2.
  \]
  Now, Theorem~\ref{thm:freedman} and the above calculation imply that
  $\Pr[|Z_T - Z_0| \ge \lambda] \le 2e^{-\lambda^2/4\sigma^2}$ for all
  $0 < \lambda \le n$. This proves the lemma.
\end{proof}

Finally we state our main theorem.
\begin{theorem}[Restatement of Theorem~\ref{thm:rand-walk}]\label{thm:randwalk-subg}
  Algorithm~\ref{alg:randwalk} runs in expected polynomial time, and
  outputs a random vector $\chi$ such that the random variable
  $\sum_{i = 1}^n{\chi_iv_i} - y$ is $O(\sqrt{\log n})$-subgaussian.
\end{theorem}
\begin{proof}
  Let $E$ be the event that for all $i$, $|\chi(T)_i| \ge 1-\delta$
  (equivalently, that $A(T) = \emptyset$). The algorithm takes returns
  if $E$ holds, and otherwise it restarts. By
  Lemma~\ref{lm:convergence} this event occurs with probability at
  least $1/2$, so there will be a constant number of restarts in
  expectation. Since the random walk talks $T$ steps, where $T$ is
  polynomial in the input size, and each step can also be executed in
  polynomial time, it follows that the expected running time of the
  algorithm is polynomial.

  Because the algorithm returns an output exactly when $E$ holds, the
  output is distributed as the random vector $\chi$ conditioned on
  $E$.  Let us fix a vector $\theta \in S^{n-1}$ once and for all. Let
  $Y$ be the random variable $\sum_{i = 1}^n{\chi_iv_i} - y$ and let
  $Z = \langle \theta, Y\rangle$. Let $Y_t$ and $Z_t$ be defined as in
  the proof of Lemma~\ref{lm:subg}.  Let $s = \gamma\sqrt{2T}$ be the
  parameter with which we proved $Z_T-Z_0$ is subgaussian in
  Lemma~\ref{lm:subg}. We will show that $Z - Z_0$, conditioned on
  $E$, is $2s$-subgaussian, i.e.~we will prove that $\Pr[|Z-Z_0| \ge
  \lambda \mid E] \le 2e^{-\lambda^2/8s^2}$.  Observe that this inequality is
  trivially satisfied for $\lambda \le \lambda_0 = 2\sqrt{2\ln 2}s$,
  since the right hand side is at least $1$ in this range. For the
  rest of the proof we will assume that $\lambda > \lambda_0$.

  Conditional on $E$, and using the
  triangle inequality, we can bound the distance between $Y$ and
  $Y(T)$ by
  \begin{align*}
  \|Y - Y(T)\|_2 &= \left\|\sum_{i = 1}^n{(\sign(\chi(T)_i) - \chi(T)_i)v_i} \right\|_2\\
  &\le \sum_{i = 1}^n{(1- |\chi(T)_i|) \|v_i\|_2} \le \delta n.
  \end{align*}
  By Cauchy-Schwarz, we get that, conditional on $E$, $|Z - Z_T| \le
  \delta n$, which implies that $|Z - Z_0| \le |Z_T - Z_0| + \delta
  n$, by the triangle inequality. Then, we have that, conditional on
  $E$, $|Z - Z_0| \ge \lambda \Rightarrow |Z_T - Z_0| \ge \lambda -
  \delta n$. For every $\lambda > \lambda_0$, $\delta n \le 2\sqrt{2\ln
    2 \log_22n} \le \lambda_0/2 < \lambda/2$. Therefore, conditional
  on $E$ and for every $\lambda > \lambda_0$, $|Z - Z_0| \ge \lambda
  \Rightarrow |Z_T - Z_0| \ge \lambda/2$. By Lemma~\ref{lm:subg} we
  have
  \[
  \Pr[|Z - Z_0| \ge \lambda \mid E] 
  \le \Pr[|Z_T - Z_0| \ge \lambda/2 \mid E]
  \le 2\Pr[|Z_T - Z_0| \ge \lambda/2]
  \le 4e^{-\lambda^2/4s^2}.
  \]  
  Recall that for every $\lambda > \lambda_0$, $e^{-\lambda^2/8s^2} <
  1/2$, so the right hand side above is at most
  $2e^{-\lambda^2/8s^2}$, as claimed. Therefore, $Z-Z_0$, conditioned
  on $E$, is $2s$-subgaussian, and, since $s = O(\sqrt{\log n})$, this
  suffices to prove the theorem.
\end{proof}

\section{Recentering procedure}
\label{sec:recenter}
In this section we will give an algorithmic variant of the recentering
procedure in Theorem~\ref{thm:recenter}.

Given a convex body $K\subseteq \mathbb{R}^n$, let $b$ be its barycenter under the Gaussian distribution. The following lemma shows that if we have an estimate $b'$ of the barycenter, which is close to $b$ but farther from the origin, then shifting $K$ to $K-b'$, increases the Gaussian volume of $K$.

\begin{lemma}
\label{shiftbary}
Let $b$ be the barycenter of $K\subseteq \mathbb{R}^n$ and $b'$ a point in $\mathbb{R}^n$ satisfying $\|b-b'\|_2\le \delta/3$ and $\| b'\|_2\ge
\delta$. Then, $\gamma_n(K-b')\ge e^{\delta^2/6}\gamma_n(K)$ 
\end{lemma}
\begin{proof}
Let $\eta=b-b'$.
\begin{eqnarray*}
\frac{\gamma_n(K-b')}{\gamma_n(K)} &=& \frac{(\frac{1}{2\pi})^{n/2}\int_{x\in K-b'} e^{-\|x\|_2^2/2}dx}{(\frac{1}{2\pi})^{n/2}\int_{x\in K} e^{-\|x\|_2^2/2}dx} \\
& = &  \frac{(\frac{1}{2\pi})^{n/2}\int_{y\in K}
  e^{-\|y\|_2^2/2-\|b'\|_2^2/2+\langle y,b'
    \rangle}dy}{(\frac{1}{2\pi})^{n/2}\int_{x\in K} e^{-\|x\|_2^2/2}dx}
\qquad \textrm{(change of variables $y=x+b'$)}
\end{eqnarray*}
Let $Y$ be a random variable drawn from the $n$-dimensional Gaussian
distribution restricted to the body $K$. Then the right hand side
above is equal to
\begin{eqnarray*}
e^{-\|b'\|_2^2/2} \E[e^{\langle Y,b'\rangle}] 
& \ge& e^{-\|b'\|_2^2/2} e^{\langle \E[Y],b'\rangle} \qquad \textrm{(by Jensen's inequality)}\\
& =& e^{-\|b'\|_2^2/2+\langle b,b'\rangle}  = e^{-\|b'\|_2^2/2+\langle b',b'\rangle+\langle \eta,b'\rangle}= e^{\|b'\|_2^2/2+\langle \eta,b'\rangle}\\
& =& e^{\|b\|_2^2/2-\|\eta\|_2^2/2}  \\
& \ge& e^{(2\delta/3)^2/2-(\delta/3)^2/2}=e^{\delta^2/6} \qquad (\|b\|_2\ge\|b' \|_2-\|\eta \|_2\ge 2\delta/3)
\end{eqnarray*}
\end{proof}

\subsection{Algorithm}

In our recentering algorithm we use the geometric language of
section~\ref{sec:intro-recentering}. Instead of the vectors $v_1,
\ldots, v_n$ and the shift $t \in \sum_{i = 1}^n{[-v_i,v_i]}$, we work
directly with the parallelepiped $P = \sum_{i = 1}^n{[-v_i,v_i]} -
t$. Notice that a facet of $P$ corresponds to a fractional coloring
with some coordinates fixed. Indeed, a facet $F$ of $P$ is determined by a
subset $S \subseteq [n]$, and a coloring $\chi \in \{-1, 1\}^S$, and
equals $F = \sum_{i \not \in S}[-v_i, v_i] + \sum_{i \in S}{\chi_i
  v_i} - t$. The size of the set $S$ is equal to the co-dimension of
$F$, so a vertex (face of dimension 0) is equivalent to full coloring
$\chi \in \{-1, 1\}^n$. The edges (faces of dimension 1) are linear
segments that have length
exactly twice the length of the corresponding vectors. We say that $P$
has side lengths at most $\ell$ if each edge of $P$ has length at most
$\ell$: this corresponds to requiring that $\max_i \|v_i\|_2 \le
\ell/2$. 
Given a point $p \in P$, we denote by $F_P(p)$ the face of $P$ that
contains $p$ and has minimal dimension. We denote by $W_P(p)$ the
subspace $\mathrm{span}(F_P(p) - p)$

In this language, the (linear) discrepancy problem is translated to
the problem of finding a vertex of $P$ inside $K$. The recentering
problem can also be expressed in this way: we are looking for a point
$p \in P \cap K$ such that the Gaussian measure of $(K - p) \cap
W_P(p)$, restricted to $W_P(p)$, is at least that of $K$, and $b((K -
p) \cap W_P(p))$ is close to $0$. To do this, we start out by
approximating $b = b(K)$, the barycenter of $K$. If $b$ is close to
the origin, then we are already done and can return.  If $b$ is far
from origin, then moving the origin to $b$ (i.e. shifting $K$ and $P$
to $K-b,P-b$ respectively), should only help us by increasing the
Gaussian volume of $K$. But we cannot make this move if $b$ lies
outside $P$. In this case, we start moving towards $b$; when we hit
$\partial P$, the boundary of $P$, we stop and induct on the facet we
land on, choosing the point on boundary of $P$ we stopped on as our
new origin. We show that even this partial move towards $b$ does not
decrease the volume of $K$. Moreover, it ensures that the origin
always stays inside $P$.

One difficulty is that we cannot efficiently compute the barycenter of $K$
exactly. To get around this, we use random sampling from Gaussian distribution
restricted to $K$ to estimate the barycenter with high accuracy. We will then
return a shift of the body $K$ such that its barycenter is $\delta$-close to the
origin, where the running time is polynomial in $n$ and $(1/\delta)$ and it
suffices to choose $\delta$ as inversely polynomial in $n$.  We assume that we
have access to a membership oracle for the convex body $K$.

\begin{algorithm}[h]
\caption{Recentering procedure}
\label{alg:constrbana}
\begin{algorithmic}[1]
\State \textbf{Input}: Convex body $K\subseteq\mathbb{R}^n$ with $\gamma_n(K)\ge
1/2$, an $n$-dimensional parallelepiped $P \ni 0$, $\delta \ge 0$ and error
probability $\eps \in (0,1)$. 
\State \textbf{Output}: See statement of Theorem~\ref{thm:recenter-alg}. 

\State If $0 \not\in P \cap K$, return ${\rm FAIL}$.
\State Set $N = \lceil 24/\delta^2\rceil + n$.
\State Set $q = 0$, $W = W_P(0)$, $\bar{K} = K \cap W$, $\bar{P} = F_P(0)$.
\For{$i = 1,\ldots,N$} 
	\State Compute an estimate $b'$ of the barycenter $b$ of $\bar{K}$ restricted 
        to the subspace $W$, 

satisfying $\|b-b'\|_2\le \delta/6$ with probability at least $1-\epsilon/N$. 

If $b' \not\in \bar{K}$, return ${\rm FAIL}$, otherwise continue. \label{step1} 
	\If{$\|b'\|_2 \le \delta/2$}
			Return $q$. \label{step2} 
	\ElsIf{$\|b'\|_2 > \delta/2$ and $b'\not\in \bar{P}$}
	\State Compute $\lambda\in (0,1)$ be such that $\lambda b' \in
\partial \bar{P}$ relative to $W$. 
	\State Set $s = \lambda b'$.
  \Else
  \State Set $s = b'$.
	\EndIf
	\State Set $q = q+s$.
	\State Set $W = W_{\bar{P}}(s)$, $\bar{P} = F_{\bar{P}}(s)-s$,
$\bar{K} = (\bar{K}-s) \cap W$. \label{step3}
\EndFor
\State Return ${\rm FAIL}$.
\end{algorithmic}
\end{algorithm}

The following theorem is an algorithmic version of Theorem~\ref{thm:recenter}.
We note that the guarantees of the algorithm are relatively robust. This is to
make it simpler to use within other algorithms, since it may be called on
invalid inputs as well as output incorrectly with small probability.

\begin{theorem}\label{thm:recenter-alg}
Let $P$ be a parallelepiped in $\R^n$ containing the origin and $K \subseteq
\R^n$ be a convex body of Gaussian measure at least $1/2$, given by a membership
oracle, and let $\delta \geq 0$ and $\eps \in (0,1)$. Then,
Algorithm~\ref{alg:constrbana} on these inputs either returns ${\rm FAIL}$ or a
point $p \in P \cap K$. Furthermore, if the input is correct, then with
probability at least $1-\eps$, it returns $p$ satisfying
\begin{enumerate}
\item The Gaussian measure of $(K - p) \cap W_P(p)$ on $W_P(p)$, is at least
	that of $K$;
\item $\|b((K - p) \cap W_P(p))\|_2 \le \delta$. 
\end{enumerate}
Moreover, Algorithm~\ref{alg:constrbana} runs in time polynomial in $n$,
$1/\delta$ and $\ln(1/\eps)$.
\end{theorem}

\begin{proof}


Firstly, it easy to check by induction, that at the beginning of each iteration
of the for loop that
\begin{equation}
\label{eq:rec-0}
q \in P \cap K,\quad W = W_P(q), \quad \bar{K} = (K-q) \cap W_P(q), \quad \bar{P} = F_P(q)-q \text{ .}
\end{equation}

To prove correctness of the algorithm, we must show that the algorithm returns a
point $q$ satisfying the conditions of Theorem~\ref{thm:recenter-alg} with
probability at least $1-\epsilon$. 

For this purpose, we shall condition on the event that all the
barycenter estimates computed on line~\ref{step1} are within distance
$\delta/6$ of the true barycenters, which we denote by $\mathcal{E}$.
Since we run the barycenter estimator at most $N$ times, by the union
bound, $\mathcal{E}$ occurs with probability at least $1-\eps$. We
defer the discussion of how to implement the barycenter estimator 
till the end of the analysis.  

With this conditioning, we prove a lower bound on the Gaussian mass as
a function of the number of iterations, which will be crucial for
establishing the correctness of the algorithm. 

\begin{claim}
\label{cl:g-lb}
Let $W,\bar{K},\bar{P}$ denote the state after $t \geq 0$
non-terminating iterations. Let $k_t \geq 0$ denote number of
iterations before time $t$, where the dimension of $W$
decreases. Then, conditioned on $\mathcal{E}$, we have that
\[
\gamma_W(\bar{K}) \geq e^{(t-k_t)\delta^2/24} \gamma_n(K)  \text{ .}
\]
\end{claim}

\begin{proof}
We prove the claim by induction on $t$. At the base case $t=0$
(i.e.~at the beginning of the first iteration), note that $k_t = 0$ by
definition. If $W = \R^n$, the inequality clearly holds since $\bar{K}
= K$. If $W \subset \R^n$, then since $\gamma_n(K) \geq 1/2$ by
Lemma~\ref{lem:origincut}, we have $\gamma_W(\bar{K}) \geq
\gamma_n(K)$. The base case holds thus holds. 

We now assume that the bound holds at time $t$ and prove it for $t+1$,
assuming that iteration $t+1$ is non-terminating. Let $b$,$b'$,$s$
denote the corresponding loop variables, and $W',\bar{K}',\bar{P}'$
denote the new values of $W,\bar{K},\bar{P}$ after line 16.

Since the iteration is non-terminating, we have that $\|b'\|_2 >
\delta/2$. Since by our conditioning $\|b'-b\|_2 \leq \delta/6$, by
Lemma~\ref{shiftbary} and the induction hypothesis, we have that
\begin{equation}
\label{eq:rec-1}
\gamma_W(\bar{K}-b') \geq e^{\delta^2/24} \gamma_W(\bar{K}) \geq
e^{(t+1-k_t)\delta^2/24} \gamma_n(K) \text{ .}
\end{equation}
Note that we drop in dimension going from $W$ to $W'$ if and only if
$s$ lies on the boundary of $\bar{P}$ relative to $W$ (since then the
minimal face of $\bar{P}$ containing $s$ is lower dimensional). 

We now examine two cases. In the first case, we assume $b'$ is in the
relative interior of $\bar{P}$. In this case, we have $s = b'$, and
hence $W = W'$ and $\bar{K}' = \bar{K}-b'$. Given this, $k_{t+1} =
k_t$ (no drop in dimension) and the desired bound is derived directly
from Equation~\eqref{eq:rec-1}.  

In the second case, we assume that $b'$ is not in the interior of
$\bar{P}$ relative to $W$. In this case, $s = \lambda b' \in \partial
\bar{P}$ relative to $W$, for some $\lambda \in [0,1]$. Furthermore,
$W' \subset W$ and $k_{t+1} = k_t + 1$. From here, by Ehrhard's
inequality and Equation~\ref{eq:rec-1}, we get that  
\begin{equation}
\label{eq:rec-2}
\begin{split}
\gamma_W(\bar{K}-s) &\ge
\Phi((1-\lambda)\Phi^{-1}(\gamma_W(\bar{K}))+\lambda
\Phi^{-1}(\gamma_W(\bar{K}-b'))) 
                    \ge \gamma_W(\bar{K}) \\ 
        &\ge e^{(t-k_t)\delta^2/24} \gamma_n(K) =
e^{(t+1-k_{t+1})\delta^2/24} \gamma_n(K) \geq 1/2  \text{ .}
\end{split}
\end{equation}
Lastly, by Lemma~\ref{lem:origincut} since $\gamma_W(\bar{K}-s) \geq
1/2$, we have that 
\begin{equation}
\label{eq:rec-3}
\gamma_{W'}(\bar{K}') = \gamma_{W'}(\bar{K}-s) \geq
\gamma_{W}(\bar{K}-s) \text{ .}
\end{equation}
The desired bound now follows combining
Equations~\eqref{eq:rec-2},\eqref{eq:rec-3}.
\end{proof}

We now prove correctness of the algorithm conditioned on
$\mathcal{E}$. We first show that conditioned on $\mathcal{E}$, the
algorithm returns $q$ from line 8 during some iteration of the for
loop. For the sake of contradiction, assume instead that the algorithm
returns ${\rm FAIL}$. Let $W,\bar{K},\bar{P}$ denote the state after the
end of the loop. Then, by Claim~\ref{cl:g-lb}, we have that
\[
\gamma_{W}(\bar{K}) \geq e^{(N-k_N)\delta^2/24}\gamma_n(K)
\geq e^{(N-n)\delta^2/24} \gamma_n(K) \geq e \gamma_n(K) > 1 \text{, }
\]    
where we used that fact $k_N \leq n$, since dimension cannot drop more
than $n$ times. This is clear contradiction however, since Gaussian
measure is always at most $1$.

Given the above, we can assume that the algorithm returns $q$ during
some iteration of the for loop. Let $W,\bar{K},\bar{P},b'$ denote the
state at this iteration. Since we return at this iteration, we must
have that $\|b'\|_2 \leq \delta/2$. Given $\mathcal{E}$, we have that
the barycenter $b$ of $\bar{K}$ satisfies 
\[
\|b\|_2 \leq \|b'\|_2+\|b-b'\|_2 \leq \delta/2+\delta/6 < \delta
\text{ .}
\] 
By Claim~\ref{cl:g-lb}, we also know that $\gamma_{W}(\bar{K}) \geq
\gamma_n(K)$. Since by Equation~\ref{eq:rec-0}, $q \in P$ and $\bar{K}
= (K-q) \cap W_P(q)$, the correctness of the algorithm follows.

For the runtime, we note that it is dominated by the $N =
O(1/\delta^2+n)$ calls to the barycenter estimator. Thus, as long as
the estimator runs in $\poly(n,\ln(1/(\epsilon \delta)))$ time, the
desired runtime bound holds.

It remains to show that we can estimate the barycenter efficiently. We show how
to do this in appendix in Theorem~\ref{bary} with failure probability at most
$\eps/N$ in time $\poly(n,1/\delta,\ln(N/\eps)) =
\poly(n,1/\delta,\ln(1/\epsilon))$, as needed. 
\end{proof}

\section{Algorithmic Reduction from Asymmetric to Symmetric Banaszczyk}
\label{sec:alg-asym-to-sym}

In this section, we make algorithmic the reduction in
section~\ref{sec:red-asym-to-sym-formal} from the asymmetric to the symmetric case.
This will directly imply that given an algorithm to return a vertex of $P$
contained in a {\em symmetric} convex body $K$ of Gaussian volume at least a
half, we can also efficiently find a vertex of $P$ contained in an asymmetric
convex body of Gaussian measure at least a half. 

\begin{definition}[Symmetric Body Coloring Algorithm]
\label{def:symcolor-alg}
We shall say that $\mathcal{A}$ is a symmetric body coloring algorithm, if given
as input an $n$-dimensional parallelepiped $P \ni 0$ of side lengths at most
$l_{\mathcal{A}}(n) $, $l_\mathcal{A}$ a non-negative non-increasing function
of $n$, and a {\em symmetric} convex body $K \subseteq \R^n$ satisfying
$\gamma_n(K) \ge 1/2$, given by a membership oracle, it returns a vertex of $P$
contained in $K$ with probability at least $1/2$.
\end{definition}

Let $\alpha = 4(1+\pi\sqrt{8 \ln 2})$. We now present an algorithm, which uses
$\mathcal{A}$ as a black box and achieves the same guarantee for asymmetric
convex bodies, with only a constant factor loss in the length of the vectors.

\begin{algorithm}[h]
\caption{Reducing asymmetric convex bodies to symmetric convex bodies}
\label{alg:asymbana}
\begin{algorithmic}[1]
\State \textbf{Input}: Algorithm $\mathcal{A}$ as in~\eqref{def:symcolor-alg},
$K\subseteq\mathbb{R}^n$ convex body, given by membership
oracle, with $\gamma_n(K)\ge 1/2$, $P \ni 0$ an $n$-dimensional parallelepiped
of side lengths at most $l_{\mathcal{A}}(n)/\alpha$.
\State \textbf{Output}: A vertex $v$ of $P$ contained in $K$.
\State Call Recentering Procedure on $K$ and $P$ and $\delta = \frac{1}{32\sqrt{2\pi}}$ and $\eps = 1/4$. 

\noindent Restart from line 3 if the call outputs ${\rm FAIL}$, and otherwise let
$q$ be the output. 
\State Call $\mathcal{A}$ on $\alpha(F_P(q)-q)$ and 
$\alpha (K-q)\cap(q-K)\cap W_P(q)$ inside $W_P(q)$. 

\noindent Let $v$ be the output. 
\State If $v/\alpha+q \in K$ and is a vertex of $P$, return $v/\alpha+q$. 
       Else, restart from line 3.
\end{algorithmic}
\end{algorithm}

\begin{theorem}
\label{thm:reduceit}
Algorithm~\ref{alg:asymbana} is correct and runs in expected polynomial time.
\end{theorem}
\begin{proof}
Clearly, by line 5 correctness is trivial, so we need only argue that it runs in
expected polynomial time. Since the runtime of the recentering procedure
(Algorithm~\ref{alg:constrbana}) is polynomial, and the runs are independent, we
need only argue that line 5 accepts with constant probability. Since the
recentering procedure outputs correctly with probability at least $1-\eps = 3/4$, we may
condition on the correctness of the output $q$ in line 3. 

Under this conditioning, by the guarantees of the recentering algorithm, letting
$d = \dim(F_P(q))$, $W = W_P(q)$ and $C = (K-q) \cap W$, we have that 
\[
\gamma_W(C) \geq 1/2 \quad \text{ and } \quad
\|b(C)\|_2 \leq \frac{1}{32\sqrt{2\pi}} \text{ .}
\]
Thus by Lemma~\ref{lem:med-to-exp}, for $X \in W$ the $d$-dimensional standard
Gaussian on $W$, we have that
\[
\E[\|X\|_{C \cap -C}] \leq 2\E[\|X\|_C] \leq 2(1+\pi\sqrt{8 \ln 2}) \text{ .}
\]
Hence by Markov's inequality, $\Pr[X \in \alpha (C \cap -C)] = \Pr[X \in
4(1+\pi\sqrt{8 \ln 2})(C \cap -C))] \geq 1/2$.

Now by construction $\alpha P$ has side lengths at most $l_{\mathcal{A}}(n)$,
and hence $\alpha(F_P(q)-q)$ also has side lengths at most $l_{\mathcal{A}}(n)
\leq l_{\mathcal{A}}(d)$. Thus, $\mathcal{A}$ on input $\alpha(F_P(q)-q)$ and
$\alpha (C \cap - C)$, outputs a vertex $v$ of $\alpha(F_P(q)-q)$ contained in
$\alpha (C \cap -C) \subseteq \alpha (K-p) \cap W$ with probability at least
$1/2$. Hence, the check in line 5 succeeds with constant probability, as needed.
\end{proof}

The above directly implies Theorem~\ref{thm:weak-bana}, as shown below.

\begin{proof}[Proof of Theorem~\ref{thm:weak-bana}] 
Let $P \subseteq \R^n$ be an $n$-dimensional parallelepiped containing the
origin of side lengths at most $c/\sqrt{\log n}$.  Let $v_1,\dots,v_n \in \R^n$,
$\max_{i \in [n]} \|v_i\| \leq \frac{c}{\sqrt{\log n}}$, and $t \in \sum_{i=1}^n
[-v_i,v_i]$, denote the vectors such that $P = \sum_{i=1}^n [-v_i,v_i]-t$.

On input $v_1,\dots,v_n$ and $t$, the random walk sampler
(Algorithm~\ref{alg:randwalk}) outputs in expected polynomial time a random
$\chi \in \set{-1,1}^n$ such that $\sum_{i=1}^n \chi_i v_i - t$ is
$O(c)$-subgaussian and supported on the vertices of $P$. Thus by
Lemma~\ref{lem:subg-red} part~\ref{lem:subg-red-sym}, we can pick $c > 0$ small
enough such that for any symmetric convex body $K \subseteq \R^n$ with
$\gamma_n(K) \geq 1/2$, we have $\Pr[\sum_{i=1}^n \chi_i v_i - t \in K] \geq 1/2$. 

Thus, letting $\mathcal{A}$ denote the above sampler, we see that $\mathcal{A}$
satisfies the conditions~\ref{def:symcolor-alg} with $l_{\mathcal{A}}(n) =
2c/\sqrt{\log n}$. The theorem now follows by combining
Algorithm~\ref{alg:asymbana} with $\mathcal{A}$. 
\end{proof}

\section{Body Centric Algorithm for Asymmetric Convex Bodies}
\label{sec:asymmetric}
 
In this section, we give the algorithmic implementation of the extended
recentering procedure, which returns full colorings matching the guarantees of
Theorem~\ref{thm:weak-bana}. Interestingly, the coloring output by the
procedure will be essentially deterministic. The only randomness will be in
effect due to the random errors incurred in estimating barycenters. 
 
For a convex body $K \subseteq \R^n$, unit vector $\theta \in
\R^n$,$\|\theta\|_2=1$, and $v \in \R$, we define the shifted slice $K^\theta_v
= (K-v\theta) \cap \set{x \in \R^n: \pr{\theta}{x}=0}$.  The main technical
estimate we will require in this section, is the following lower bound on the
gaussian measure of shifted slices. We defer the proof of this estimate to
section~\ref{sec:slice-estimate-bar}.
 
\begin{theorem}
\label{thm:estneg}
There exists universal constants $v_0,\eta_0,c_0 > 0$, such that for any $n \geq
1$, convex body $K \subseteq \R^n$ satisfying $\|b(K)\|_2 = \eta \leq \eta_0$ and
$\gamma_n(K) = \alpha \geq 3/5$, $v \in [-v_0,v_0]$ and $\theta
\in \R^n$, $\|\theta\|_2=1$, we have that 
\[
\gamma_{n-1}(K_v^\theta) \geq 
 (\alpha-c_0\eta)(1-\frac{e^{-\frac{1}{100v^2}}}{4\sqrt{2\pi}}) \text{.}
\]
\end{theorem}
 
 The above inequality says that if barycenter of $K$ is close to the origin, then the
 Gaussian measure of parallel slices of $K$ does not fall off too quickly as we
 move away from the origin.
 
 Recall that the problem can be recast as finding a vertex of a
 parallelepiped $P$ contained inside the convex body $K$, where the
 parallelepiped $P=\sum_{i=1}^n [-v_i,v_i]-t$ and $t\in\sum_{i=1}^n [-v_i,v_i]$. Thus, $0\in P$.
 
 
 We start out by calling the recentering procedure to get the barycenter, $b$,
close to the origin. This recentering allows us to rescale $K$ by a constant
factor such that the Gaussian volume of K increases i.e. we replace $P$ by
$\beta P$ and $K$ by $\beta K$ where $\beta=1+\pi\sqrt{8\log 2}+4\pi\sqrt{\log
2}$ is chosen such that the volume of $K$ after rescaling is at least $3/4$. 
 Then we find a point $q^*$ on $\partial P$, the boundary of $P$ which is closest to the origin. We recurse by taking a $(n-1)$-dimensional slice of $K$(here we abuse notation by calling the convex body after rescaling as also $K$) with the facet containing $q^*$. A crucial point here is that we choose $q^*$ as the origin of the $(n-1)$-dimensional space we use in the induction step. This is done to maintain the induction hypothesis that the parallelepiped contains the origin. Theorem~\ref{thm:estneg} guarantees that in doing so, we do not lose too much Gaussian volume. 
 
\begin{algorithm}[h]
\caption{Body centric algorithm for general convex bodies}
\label{alg:bodycentric}
\begin{algorithmic}[1]
\State \textbf{Input}: Convex body $K\subseteq\mathbb{R}^n$, given by a
membership oracle, with $\gamma_n(K)\ge 1/2$, an $n$-dimensional parallelepiped
$P \ni 0$ of side lengths at most $2\alpha_n$.
\State \textbf{Output}: A vertex of $P$ contained in $K$.

\State \label{step20} Call Recentering Procedure on $K$ and $P$ with 
parameters $\delta = \eta_n$ and $\epsilon = \frac{1}{2(n+1)}$. 

Restart from line~\ref{step20} if the call outputs ${\rm FAIL}$, and otherwise let $q$ denote the output.
\State \label{step21} Set $\bar{q}=0$, $\bar{W}=W_P(q)$, $\bar{K}=\beta((K-q) \cap \bar{W})$, 
           $\bar{P}=\beta(F_P(q)-q)$.
\Repeat
	\State \label{step22} Call Recentering Procedure on $\bar{P}$ and $\bar{K}$ with 
    $\delta = \eta_n$ and $\epsilon = \frac{1}{2(n+1)}$. 

     Restart from line~\ref{step20} if the call outputs ${\rm FAIL}$, and otherwise let $s$ denote the output.
	\State  \label{step23} Set $\bar{q} = \bar{q} + s$, $\bar{W} = W_{\bar{P}}(s)$, 
                $\bar{K} = (\bar{K}-s) \cap \bar{W}$, $\bar{P}=(F_{\bar{P}}(s)-s)$.
  \If{$\dim(\bar{W}) \neq 0$}
		\State \label{step24} Compute $s \in {\rm argmin} 
						\set{\|p\|_2: p \in \partial \bar{P} \text{ relative to } \bar{W}}$.  
	  \State  \label{step25} Set $\bar{q} = \bar{q} + s$, $\bar{W} = W_{\bar{P}}(s)$, 
                $\bar{K} = (\bar{K}-s) \cap \bar{W}$, $\bar{P}=(F_{\bar{P}}(s)-s)$.
  \EndIf
\Until{$\dim(\bar{W}) = 0$}
\State If $q+\bar{q}/\beta \in K$ and is a vertex of $P$, return $q+\bar{q}/\beta$. 
Else, restart from line~\ref{step20}.
\end{algorithmic}
\end{algorithm}
 
 \begin{lemma}
 \label{lemma:volincr}
 Given a convex body $K$ in $\mathbb{R}^n$ such that $\gamma_n(K)\ge 1/2$ and
$\|b(K)\|_2 \le \frac{1}{32\sqrt{2\pi}}$, then $\gamma_n(\beta K)\ge 3/4$, where
$\beta=1+\pi\sqrt{8\log 2}+4\pi\sqrt{\log 2}$.
 \end{lemma}
 \begin{proof}
 Let $X$ be the standard $n$-dimensional Gaussian. From Lemma~\ref{lem:med-to-exp}, $\E[\|X\|_K] \le 1+\pi\sqrt{8\log 2}$. This gives 
 \begin{eqnarray*}
\Pr[\|X\|_K>\beta] &=& \Pr[\|X\|_K-\E[\|X\|_K] >\beta-\E[\|X\|_K]]\\
& \le &\Pr[\|X\|_K-\E[\|X\|_K] >\beta-1-\pi\sqrt{8\log 2}]
\end{eqnarray*}
 
By Lemma~\ref{lem:ball-cont} and Lemma~\ref{lem:lipsch}, the function $\|.\|_K$
is $4-$Lipschitz. Then, by Theorem~\ref{thm:gaussian-con}
\[
\Pr[\|X\|_K>\beta]  \le e^{-\frac{2}{\pi^2}(\frac{\beta-1-\pi\sqrt{8\log 2}}{4})^2} = 1/4 \text{ .}
\]
Thus, $\gamma_n(\beta K)=\Pr[X\in \beta K]=1-\Pr[\| X\|_K>\beta] \ge 1-1/4=3/4$, as needed.
\end{proof}
 
For Algorithm~\ref{alg:bodycentric}, we use the parameters
\[
\alpha_n = \min \set{v_0,\frac{1}{10\sqrt{\log(2n)}}}, \quad \eta_n = \min \set{\eta_0, \frac{1}{32\sqrt{2\pi}}, \frac{1}{14c_0n}} \text{, }
\]
where $v_0,\eta_0,c_0$ are as in Theorem~\ref{thm:estneg}. We now give the
formal analysis of the algorithm.

We begin by explaining how to compute the minimum norm point on the boundary of a
parallelepiped.

\begin{lemma} Let $P = \sum_{i=1}^k [-v_i,v_i]-t \subset \R^n$ be a
parallelepiped of side lengths at most $2\alpha$ containing the origin, with
$v_1,\dots,v_k$ linearly independent. Let $v_1^*,\dots,v_k^*$ denote the dual
basis, i.e.~the unique set of vectors lying inside $W := {\rm
span}(v_1,\dots,v_k)$ satisfying $\pr{v_i^*}{v_j} = 1$ if $i=j$ and $0$
otherwise, and let $s$ denote an element of minimum $\ell_2$ norm in the set 
$\set{\frac{\pm 1 + \pr{v_i^*}{t}}{\|v_i^*\|_2} \cdot \frac{v_i^*}{\|v_i^*\|_2}:
i \in [k]}$. Then, the following holds:
\begin{enumerate}
\item $s \in {\rm argmin} \set{\|p\|: p \in \partial P \text{ relative to } W}$.
\item $W_P(s) \subseteq \set{x \in W: \pr{s}{x} = 0}$.
\item $\|s\|_2 \leq \alpha$.
\end{enumerate}
Furthermore, $s$ can be computed in polynomial time.
\label{lem:closest-pt}
\end{lemma}
\begin{proof}
Note that for any $x \in W$, we have that $x = \sum_{i=1}^k \pr{v_i^*}{x}v_i$.
From here, given that $P = \sum_{i=1}^k [-v_i,v_i] -t$, it is easy to check that
\begin{equation}
\label{eq:cl-1}
P = \set{x \in W: -1+\pr{v_i^*}{t} \leq \pr{v_i^*}{x} \leq 1+ \pr{v_i^*}{t},
\forall i \in [k]} \text{ .}
\end{equation}

We now show that $s \in {\rm argmin} \set{\|p\|_2: p \in \partial P \text{
relative to } W}$. Since $0 \in P$, we must show that $x \in P$ if $\|x\|_2 \leq
\|s\|_2$ and $x \in W$, and that $s \in \partial P$ relative to $W$. Given that the vectors
$v_i^*/\|v_i^*\|_2$ have unit $\ell_2$ norm, the norm of $s$ is equal to 
\[
\omega := \min \set{\frac{|\pm 1 + \pr{v_i^*}{t}|}{\|v_i^*\|_2}: i \in [k]}
\text{ .}
\]
Now assume $x \in W$ and $\|x\|_2 \leq \omega$. Since by assumption $0 \in P$,
we must have $-1+\pr{v_i^*}{t} \leq 0 \leq 1+\pr{v_i^*}{t}$, $\forall i \in
[k]$. Therefore, for $i \in [k]$, by Cauchy-Schwarz
\begin{align*}
\pr{v_i^*}{x} &\leq \|v_i^*\|_2 \cdot \omega \leq 1+\pr{v_i^*}{t}\text{, } \\
\pr{v_i^*}{x} &\geq -\|v_i^*\|_2 \cdot \omega \geq -1+\pr{v_i^*}{t} \text{ .}
\end{align*}
Hence $x \in P$, as needed. Next, we must show that $s \in \partial P$ relative
to $W$. Firstly, clearly $s \in W$ since each $v_1^*,\dots,v_k^* \in W$, and
thus by the above argument $s \in P$. Now choose $i \in [k]$, $r \in \set{-1,1}$
such that $s = \frac{r + \pr{v_i^*}{t}}{\|v_i^*\|_2} \cdot
\frac{v_i^*}{\|v_i^*\|_2}$. Then, by a direct calculation $\pr{v_i^*}{s} =
r+\pr{v_i^*}{t}$, and hence $s$ satisfies one of the inequalities of $P$ (see
Equation~\ref{eq:cl-1}) at equality. Thus, $s \in \partial P$ relative to $W$
(note that $P$ is full-dimensional in $W$), as needed.

We now show that $W_P(s) \subseteq \set{x \in W: \pr{s}{x}=0}$. By the above
paragraph, every element $x$ of the minimal face $F_P(s)$ of $P$ containing $s$
satisfies $\pr{v_i^*}{x} = r+\pr{v_i^*}{t} = \pr{v_i^*}{s}$. In particular
$\pr{v_i^*}{x-s} = 0$. Since $s$ is collinear with $v_i^*$ ($s$ may be $0$), we
have $\pr{v_i^*}{x-s} = 0 \Rightarrow \pr{s}{x-s} = 0$. The claim now follows
since $W_P(s)$ is the span of $F_P(s)-s$ and $F_P(s)-s \subseteq \set{x \in W:
\pr{s}{x}=0}$ by the previous statement.

We now show that $\|s\|_2 \leq \alpha$. Firstly, by minimality of $s$, note that
$|r+\pr{v_i^*}{t}| \leq 1$, for $r$ and $i$ as above. Thus, by Cauchy-Schwarz,
\[
\|s\| = \left|\frac{r+\pr{v_i^*}{t}}{\|v_i^*\|_2}\right| \leq \frac{1}{\|v_i^*\|}
      = \frac{\pr{v_i}{v_i^*}}{\|v_i^*\|_2} \leq
\frac{\|v_i\|_2\|v_i^*\|_2}{\|v_i^*\|_2} = \|v_i\| \text{ .}
\]
Since $P$ has side lengths at most $2\alpha$, we have $\|v_i\| \leq \alpha$.
Thus, $\|s\|_2 \leq \alpha$, as claimed.

We now prove the furthermore. Let $V$ denote the matrix whose columns are
$v_1,\dots,v_k$. By linear independence of $v_1,\dots,v_k$, the matrix $V^\T V$
is invertible. Since then $(V (V^\T V)^{-1})^\T V = I_k$, we see that
$v_1^*,\dots,v_k^*$ are the columns of $V(V^\T V)^{-1}$ (note that these lie in
$W$ by construction), and hence can be constructed in polynomial time. Since $s$
can clearly be constructed in polynomial time from the dual basis and $t$, the
claim is proven. 
\end{proof}

\begin{theorem}
Algorithm~\ref{alg:bodycentric} is correct and runs in expected polynomial time.
\end{theorem}
\begin{proof}
Clearly, by the check on line 13, correctness is trivial. So we need only show
that the algorithm terminates in expected polynomial time. In particular, it
suffices to show the probability that a run of the algorithm terminates without
a restart is at least $1/2$.  

For this purpose, we will show that the algorithm terminates correctly
conditioned on the event that each call to the recentering procedure terminates
correctly, which we denote $\mathcal{E}$. Later, we will show that this event
occurs with probability at least $1/2$, which will finish the proof.

Let $W_1,K_1,P_1$ denote the values $\bar{W},\bar{K},\bar{P}$ after line $4$.

During the iterations of the repeat loop, it is easy to check by induction that
after the execution of either line 7 or 10, the variables
$\bar{q},\bar{W},\bar{K},\bar{P}$ satisfy:
\begin{equation}
\label{eq:bc-0}
\bar{q} \in P_1, \quad \bar{W} = W_{P_1}(\bar{q}), \quad 
\bar{K} = (K_1-\bar{q}) \cap W_{P_1}(\bar{q}), \quad \bar{P} = F_{P_1}(\bar{q})-q \text{ .}
\end{equation}

We now establish the main invariant of the loop, which will be crucial in
establishing correctness conditioned on $\mathcal{E}$:

\begin{claim} Let $\bar{W},\bar{K},\bar{P}$ denote the state after $k \geq 0$
successful iterations of the repeat loop. Then, the following holds:
\begin{enumerate}
\item $\dim{W} \leq n-k$.
\item Conditioned on $\mathcal{E}$, $\gamma_{\bar{W}}(\bar{K}) \geq
3/4-\frac{k}{7n} > 3/5$.
\end{enumerate}
\label{cl:bd-lb}
\end{claim}
\begin{proof}
We prove the claim by induction on $k$. 

For $k=0$, the state corresponds to $W_1$, $K_1$ and $P_1$. Trivially,
$\dim(W_1) \leq n = n-k$, so the first condition holds. Conditioned on
$\mathcal{E}$, we have that $K_1/\beta$ has Gaussian mass at least $1/2$
restricted to $W_1$ and its barycenter has $\ell_2$ norm at most $\eta_n$. Since
$\eta_n \leq \frac{1}{32\sqrt{2\pi}}$ by Lemma~\ref{lemma:volincr}, we have that
$\gamma_{W_1}(K_1) \geq 3/4$. Thus, the second condition holds as well. 

We now assume the statement holds after $k$ iterations, and show it holds after
iteration $k+1$, assuming that we don't terminate after iteration $k$ and that
we successfully complete iteration $k+1$. Here, we denote the state at the
beginning of iteration $k+1$ by $\bar{W},\bar{K},\bar{P}$, after line 7 by
$\bar{W}_1,\bar{K}_1,\bar{P}_1$ and at the end the iteration by
$\bar{W}_2,\bar{K}_2,\bar{P}_2$. 

We first verify that $\bar{W}_2 \leq n-(k+1)$. By the induction hypothesis $n-k
\geq \dim(\bar{W})$ and by construction $\bar{W}_2 \subseteq \bar{W}_1 \subseteq
\bar{W}$. Thus, we need only show that $\dim(\bar{W}_2) < \dim(\bar{W})$. Given
that we successfully complete the iteration, namely the call to the
recentering algorithm on line $6$ doesn't return ${\rm FAIL}$, we may
distinguish two cases.  Firstly, if $\dim(\bar{W}_2)=0$, then we must have
$\dim(\bar{W}_2) < \dim(\bar{W})$, since otherwise $\dim(\bar{W})=0$ and the
loop would have exited after the previous iteration. Second if $\dim(\bar{W}_2)
> 0$, we must have entered the if statement on line 8 since $\dim(\bar{W}_2)
\leq \dim(\bar{W}_1)$. From here, we see that $\dim(\bar{W}_2)$ corresponds to
the dimension of the minimal face of $\bar{P}_1$ containing $s$. Since $s$ is on
the boundary of $\bar{P}_1$ relative to $\bar{W}_1$, we get that
$\dim(\bar{W}_2) < \dim(\bar{W}_1) \leq \dim(\bar{W})$, as needed. Thus,
condition 1 holds at the end of the iteration as claimed. 

We now show that conditioned $\mathcal{E}$, $\gamma_{\bar{W}_2}(\bar{K}_2) \geq
3/4-(k+1)/(7n)$. By the induction hypothesis, recall that
$\gamma_{\bar{W}}(\bar{K}) \geq 3/4-k/(7n)$, thus it suffices to prove that
$\gamma_{\bar{W}_2}(\bar{K}_2) \geq \gamma_{\bar{W}}(\bar{K})-1/(7n)$.  Note
that since we decrease dimension at every iteration (as argued in the previous
paragraph), the number of iterations of the loop can never exceed $n$.  Thus,
after any valid number of iterations $l$, we always have $3/4-l/(7n) \geq
3/4-1/7 > 3/5$. In particular, we have $\gamma_{\bar{W}}(\bar{K}) \geq 3/5$.

We now track the change in Gaussian mass going from $\bar{K}$ to $\bar{K}_2$.
Since the recentering procedure on line 6 terminates correctly by our
conditioning on $\mathcal{E}$, we get that 
\[
\gamma_{\bar{W}_1}(\bar{K}_1) \geq \gamma_{\bar{W}}(\bar{K}) \quad \text{ and }
\quad \|b(\bar{K}_1)\|_2 \leq \eta_n \text{ .}
\] 
If $\dim(\bar{W}_1)=0$, then clearly $\bar{W}_2=\bar{W}_1$ and
$\bar{K}_2=\bar{K}_1$, and hence $\gamma_{\bar{W}_2}(\bar{K}_2) \geq
\gamma_{\bar{W}}(\bar{K})$ as needed. If $\dim(\bar{W}_1) > 0$, we enter the if
statement at line 8. Since $\bar{P}_1$ is a parallelepiped containing $0$ of
side length $2\alpha_n$, by Lemma~\ref{lem:closest-pt} we have that $\|s\| \leq
\alpha_n$ and $\bar{W}_2 = W_{\bar{P}_1}(s) \subseteq H_s$ where $H_s := \set{x
\in \bar{W}_1: \pr{s}{x} = 0}$.  Now if $s = 0$, then $\bar{K}_2 = \bar{K}_1
\cap \bar{W}_2$, and thus by Lemma~\ref{lem:origincut}, we have
$\gamma_{\bar{W}_2}(\bar{K}_2) \geq \gamma_{\bar{W}_1}(\bar{K}_1) \geq
\gamma_{\bar{W}}(\bar{K})$, as needed.  If $s \neq 0$, given that $\|s\| \leq
\alpha_n \leq v_0$, $\|b(\bar{K}_1)\|_2 \leq \eta_n \leq \eta_0$ and
$\gamma_{\bar{W}_1}(\bar{K}_1) \geq 3/5$, by applying Theorem~\ref{thm:estneg}
on $\bar{K}_1$ with $v = \|s\|_2$ and $\theta = s/v$, we get that
\begin{equation}
\label{eq:bd-1}
\begin{split}
\gamma_{H_s}((\bar{K}_1-s) \cap H_s) &\geq 
(\gamma_{\bar{W}_1}(\bar{K}_1)-c_0\eta_n)
(1-\frac{e^{-\frac{1}{100\alpha_n^2}}}{4\sqrt{2\pi}}) \\
&\geq \gamma_{\bar{W}}(\bar{K})
      -c_0\eta_n -\frac{e^{-\frac{1}{100\alpha_n^2}}}{4\sqrt{2\pi}} \geq
\gamma_{\bar{W}}(\bar{K}) - \frac{1}{7n} \text{ .}
\end{split}
\end{equation}
Since $\bar{K}_2 = (\bar{K}_1-s) \cap \bar{W}_2$, $\bar{W}_2 \subseteq H_s$ and
$\gamma_{H_s}((\bar{K}_1-s) \cap H_s) \geq 3/4-(k+1)/(7n) > 1/2$, by
Lemma~\ref{lem:origincut}
\begin{equation}
\label{eq:bd-2}
\gamma_{\bar{W}_2}(\bar{K}_2) \geq \gamma_{H_s}((\bar{K}_1-s) \cap H_s) \text{.} 
\end{equation}
The desired estimate follows combining~\eqref{eq:bd-1} and~\eqref{eq:bd-2}.
\end{proof}

By Claim~\ref{cl:bd-lb}, we see that the number of iterations of the repeat loop
is always bounded by $n$. Furthermore, conditioned on $\mathcal{E}$, the loop
successfully terminates with $\bar{W}$,$\bar{K}$,$\bar{P}$ satisfying
$\gamma_{\bar{W}}(\bar{K}) > 0$ and $\dim(\bar{W})=0$. Since $\dim(\bar{W}) =
0$, this implies that $\bar{W}=\bar{K}=\set{0}$. Furthermore, by
Equation~\ref{eq:bc-0}, this implies that $\bar{q} \in K_1 \cap P_1$ and
$\dim(W_{P_1}(\bar{q})) = 0$, and hence $\bar{q}$ is a vertex of $P_1$.  Since
$K_1 = \beta(K-q)$ and $P_1 = \beta (P-q)$, we get that $q+\bar{q}/\beta$ is a
vertex of $P$ contained in $K$, as needed. Thus, conditioned on $\mathcal{E}$,
the algorithm returns correctly.

To lower bound $\mathcal{E}$, by the above analysis, note that we never call the
recentering procedure more than $n+1$ times, i.e.~once on line 3 and at most $n$
times on line 6. By the union bound, the probability that one of these calls
fails is at most $(n+1) \cdot 1/(2(n+1)) = 1/2$. Thus, $\mathcal{E}$ occurs with
probability at least $1/2$, as needed.
\end{proof}

\section{An estimate on the Gaussian measure of slices}
\label{sec:slice-estimate-bar}

In this section, we prove Theorem~\ref{thm:estneg}. We will need the following estimate on Gaussian tails~\cite[Formula~7.1.13]{AS72}.

\begin{lemma}[Gaussian Tailbounds]
\label{lem:gaus-tail}
Let $X \sim N(0,1)$. Then for any $t \geq 0$,
\[
\sqrt{\frac{2}{{\pi}}} \frac{e^{-t^2/2}}{t+\sqrt{t^2+4}} \leq \Pr[X \geq t] \leq
\sqrt{\frac{2}{{\pi}}}  \frac{e^{-t^2/2}}{t+\sqrt{t^2+8/\pi}} \text{ .}
\]
\end{lemma}

Before proving Theorem~\ref{thm:estneg}, we first prove a similar result for a
special class of convex bodies in $\mathbb{R}^2$.  

We define a convex body $K$ in $\mathbb{R}^2$ to be downwards closed if
$(x,y)\in K$ implies $(x,y')\in K$ for all $y'\le y$. For notational
convenience, we shall denote the first and second coordinate of a vector in
$\R^2$ respectively as the x and y coordinates. We shall say the slice of $K$ at
$x=t$ or $y=t$ to denote either the vertical slice of $K$ having x-coordinate
$t$ or horizontal slice having y-coordinate $t$. We define the height of $K$ at
$x=t$ to be maximum y-coordinate of any point $(t,y) \in K$. By
convention, we let the height of $K$ at $x=t$ be $-\infty$ if $K$ does not a
contain a point with x-coordinate $t$.

\begin{lemma}
\label{lem:estim}
Let $K \subseteq \R^2$ be a downwards closed convex body with
$\gamma_2(K)=\alpha\ge 1/2$ and barycenter $b = b(K)$ satisfying $b_1 \geq 0$, and
let $g = \Phi^{-1}(\alpha) \geq 0$. Then, there exists a universal constant $v_0
> 0$ such that for all $0\le v \le v_0$, the height of $K$ at $x=v$ is least
$f(v,g) := g-\min \set{e^{\frac{g^2}{2}-\frac{1}{8ev^2}},(4e+2)\frac{v^2}{g}}$.
\end{lemma}

\begin{proof} \hspace{1em}

\paragraph*{{\bf Step 1: Reduction to a wedge}}

We first show that the worst-case bodies for the lemma are ``wedge-shaped'' (see
the illustration in Figure~\ref{fig:wedge}). Namely, the worst case down closed
convex bodies are of form 
\begin{equation}
\label{eq:wedge-form}
\set{(x,y) \in \R^2: x \geq -c, s x + t y \leq d} \quad \text{ where } \quad
d,s,t \geq 0, \quad s^2+t^2 = 1, c \in \R_{\geq 0} \cup \set{\infty} \text{ .}
\end{equation}
More precisely, we will show that given any body $K$ satisfying the conditions
of the theorem, there exists a wedge $W$ satisfying the conditions of the
theorem whose height at $x=v$ is at most that of $K$.

Let $K \subseteq \R^2$ satisfy the conditions of the theorem.  We
first show that $K$ contains a point on the line at $x=v$. If not, we
claim that $K$ has Gaussian mass at most $\gamma_1(-v,v) \leq
\gamma_1(-v_0,v_0) < 1/2$ by choosing $v_0$ small enough, a clear
contradiction.  To see this, note that by pushing the mass of $K$ to
the right as much as possible towards the line $x=v$, we can replace
$K$ by a band $[a,v] \times \R$ with the same Gaussian mass, and
barycenter to the right of $b(K)$. Clearly, such a band has barycenter
to the right of the y-axis iff $a \geq -v$, and hence $K$ has Gaussian
mass at most $\gamma_2([-v,v] \times \R) = \gamma_1(-v,v)$, as needed.

Now assume that $K$ has height at least $g$ at $x=v$, where we recall that $g =
\Phi^{-1}(\alpha)$. Note now that the band $W = \R \times (-\infty,g]$ (corresponding
to $s=0,t=1,d=g,c=\infty$) has height at $x=v$ at most that of $K$ and satisfies
the conditions of the theorem. Thus, we may assume that the height of $K$ at
$x=v$ is $f$, where $-\infty < f < g$. Note that $(v,f)$ is now a point on the
boundary of $K$.  

Let $g'$ denote the height of $K$ at $x=0$. Since $\gamma_2(K) \geq
1/2$, by Lemma~\ref{lem:origincut} we have that $\gamma_1(\infty,g')
\geq \gamma_2(K)$, and hence $g' \geq g \geq 0$. Thus $g' \geq g > f$,
and hence $v > 0$ (since otherwise we would have $f = g'$). By
convexity of $K$, we may choose a line $\ell$ tangent to $K$ passing
through $(v,f)$. We may now choose $t \geq 0, s, d\in \R$, such that
$s^2+t^2=1$ and $\ell = \set{(x,y) \in \R^2: s x + t y = d}$.  Since
$K$ is downwards-closed, $t \geq 0$ and $\ell$ is tangent to $K$, we
must have that $K \subseteq H_{\ell} := \set{(x,y): sx+ty \leq
  d}$. Since $0$ is below $(0,g') \in K$, we have that $0 \in
H_{\ell}$, and hence $d \geq 0$.  Given the $(0,g') \subseteq
H_{\ell}$, we have that $tg' \le d$, and, because $\ell$ is tangent at
$(v,f)$, also $sv + tf = d$; using that $v > 0$ and $g' > f$, we conclude
that $s > 0$. 

We will now show that the wedge $W = H_\ell \cap \set{(x,y) \in \R^2: x \geq
-c}$ satisfies our requirements for an appropriate choice of $c$ (note the
conditions for $s,t,d$ are already satisfied by the above paragraph). Let
$B^-_v = \set{(x,y) \in \R^2: x \leq v}$ and $B^+_v = \set{(x,y) \in
\R^2: x \geq v}$. 
Choose $c \geq -v$ such that 
$\gamma_2(W \cap B^-_v) = \gamma_2(K \cap B^-_v)$. 
Note that such a $c$ must
exist since $K \subseteq H_\ell$. Now by construction, note that $W$ has the
same height as $K$ at $x=v$, so it remains to check that $c \geq 0$,
$\gamma_2(W) \geq 1/2$ and $b(W)_1 \geq 0$. To bound the Gaussian mass, again by
construction, we have that  
\begin{align*}
\gamma_2(W) &= \gamma_2(W \cap B^-_v) + 
               \gamma_2(W \cap B^+_v)
             = \gamma_2(K \cap B^-_v) +
               \gamma_2(W \cap B^+_v) \\
            &\geq \gamma_2(K \cap B^-_v) + \gamma_2(K \cap B^+_v) 
             = \gamma_2(K) \geq 1/2 \text{ .}
\end{align*}
Given that $\gamma_2(W) \geq 1/2$, we must have that $0 \in W$ and hence $c \geq
0$, as needed. It remains to check that $b(W)_1 \geq 0$. For this purpose, note
first that we can transform $K \cap B^-_v$ into $W \cap B^-_v$ by only
pushing mass to the right, and hence
\begin{equation}
\label{eq:estim-01}
\int_{K \cap B^-_v} x e^{-(x+y)^2/2}{\rm dy\,dx} \leq \int_{W \cap
B^-_v} x e^{-(x+y)^2/2}{\rm dy\,dx} \text{ .}
\end{equation}
Since $v \geq 0$ and $K \cap B^+_v \subseteq W \cap B^+_v$, we also immediately
get that 
\begin{equation}
\label{eq:estim-02}
\int_{K \cap B^+_v} x e^{-(x+y)^2/2} {\rm dy\,dx} \leq \int_{W \cap
B^+_v} x e^{-(x+y)^2/2}{\rm dy\,dx} \text{ .}
\end{equation}
We derive $b(W)_1 \geq 0$ by combining~\eqref{eq:estim-01},~\eqref{eq:estim-02}
and our assumption that $b(K)_1 \geq 0$.


\begin{figure}[!t]
\centering
\includegraphics[width=4in]{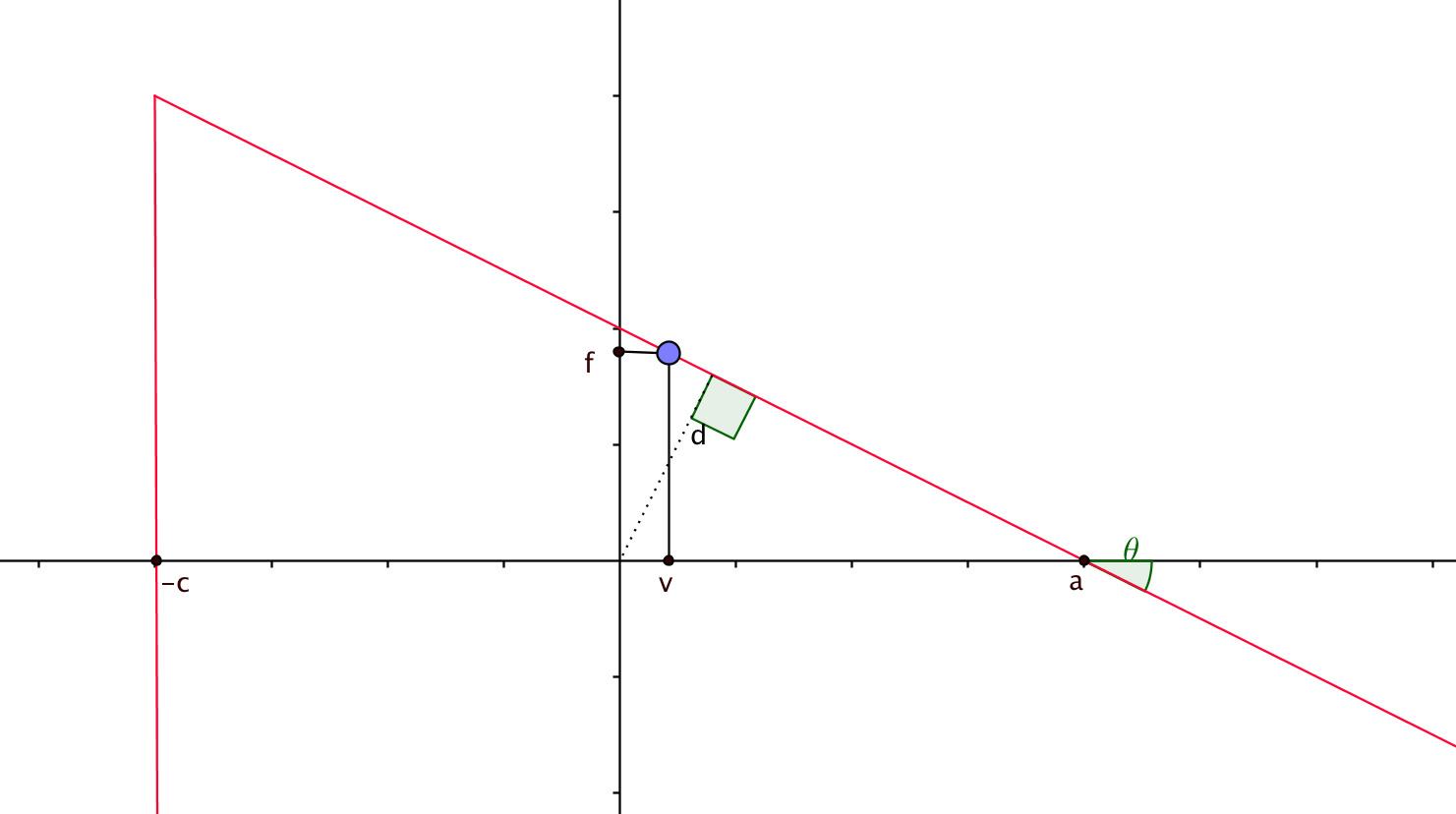}
\caption{$B$ is the triangular region beneath the red line}
\label{fig:wedge}
\end{figure}

Given the above reduction, we may now assume that $K$ is a wedge of the form \\
$\set{(x,y) \in \R^2: x \geq -c, sx + ty \leq d}$, $d,s,t \ge 0$, $s^2
+ t^2 = 1$, as in equation~\eqref{eq:wedge-form}. We first take care
of some trivial cases. Firstly, if $s = 0,t=1$, we have $K =
[-c,\infty) \times (-\infty,d]$. Then the height at $x=v$ is clearly
$d$, and since $\gamma_1(-\infty,d) \geq \gamma_2(K)$, we get $d \geq
g$ as is needed. Now assume that $s=1,t=0$, then $K = [-c,d] \times
\R$, and hence the height at $x=v$ is infinite (note that $K$ always
intersects the line at $x=v$ by the first part), and thus the desired
bound trivially holds. We may thus assume that both $s,t > 0$.

In this setting, the line $\ell = \set{(x,y) \in \R^2: sx + ty=d}$ intersects
the x-axis at $a = d/s$ and forms an angle $\theta \in (0,\pi/2)$ with the
x-axis as in Figure~\ref{fig:wedge}. Given the normalization $s^2+t^2=1$,
note that $d$ is the perpendicular distance from $0$ to the edge $K \cap \ell$ of
$K$. In what follows, we maintain the parametrization of the wedge $K$ in terms
of $a$,$\theta$,$c$, using the relations $d = a \sin \theta$, $s = \sin \theta$
and $t = \cos \theta$ to recover $d,s,t$ when needed. 

Recall that $\gamma_2(K)\ge\alpha=\gamma_1(-\infty, g)$ and $b(K)_1\ge 0$. Let
$f^* = f(v,g)$ and let $f$ be the height of $K$ at $x=v$. We want to prove that
$f \ge f^*$. If $f \ge g$, we are already done since $f^* \le g$. Note that by
Lemma~\ref{lem:origincut}, $f \geq g$ if $v = 0$ and hence we may assume $v >
0$. Our goal is now to show that $g-f \leq \min
\set{e^{\frac{g^2}{2}-\frac{1}{8ev^2}},(4e+2)\frac{v^2}{g}}$. 

\paragraph*{{\bf Step 2: Using the barycenter condition}}

We now derive a bound on how large $d$ must be, given $c$ and $\theta$, such
that the x-coordinate of barycenter is non-negative.

Let $H$ denote the halfspace $\set{(x,y) \in \R^2: s x + t y \leq d}$ and let $E=H\setminus K$.
A simple computation shows $\int_{H}x\frac{1}{2\pi}e^{-1/2(x^2+y^2)}{\rm dy\,dx} =
-\frac{e^{-d^2/2}}{\sqrt{2\pi}}\sin\theta$. Now from the fact that $b(K)_1 \ge
0$, and $E=H\setminus K$, we get 
\begin{eqnarray}
& -\frac{e^{-d^2/2}}{\sqrt{2\pi}}\sin\theta\ge
\int_{E}x\frac{1}{2\pi}e^{-1/2(x^2+y^2)}{\rm dy\,dx} \ge \int_{x\le
-c}\frac{1}{\sqrt{2\pi}}xe^{-1/2x^2}{\rm dx} = -\frac{e^{-c^2/2}}{\sqrt{2\pi}}.
\label{eqn28a}
\end{eqnarray}

It follows that $\sin \theta  \le  e^{-\frac{1}{2}(c^2-d^2)}$.
When $c^2\ge -2\log\sin\theta$, we get the following bound on $d$:
\[ 
d \ge  \sqrt{c^2+2\log\sin\theta} \text{ .}
\]

When  $c^2< -2\log\sin\theta$, this gives no useful bound on $d$ since in
that case the barycenter is non-negative even for $d=0$. But $d\ge 0$ always as
the Gaussian measure of $K$ is at least half. Thus, 
\begin{equation}
d \ge  \sqrt{\max\{0,c^2+2\log\sin\theta\}} \text{ .}
\label{eqn28}
\end{equation}

\paragraph*{{\bf Step 3: Getting a bound on $f$}}

By construction of $K$, the point $(v,f)$ lies on the boundary of $H$, and hence
\begin{equation}
v s + f t = v\sin \theta+f\cos\theta=a\sin\theta=d \text{ .}
\label{eqn13}
\end{equation}

Now, 
\[
\frac{1}{\sqrt{2\pi}}\frac{e^{-\frac{1}{2}c^2}}{c+\sqrt{c^2+4}} \le
\frac{1}{2}\gamma_1(c,\infty)\le \gamma_2(E) \qquad \textrm{(using
lemma~\ref{lem:gaus-tail})} \text{ .} 
\]
Also, 
\[
 \gamma_2(E)=\gamma_2(H)-\gamma_2(K)=\gamma_1(-\infty,d)-\gamma_1(-\infty,g)=\gamma_1(g,d)\le
\frac{1}{\sqrt{2\pi}}(d-g)e^{-\frac{1}{2}g^2}  \qquad \textrm{(since $d\ge g\ge 0$)}
\text{ .} 
\]

Combining the above two, we get
\begin{eqnarray}
\frac{e^{\frac{1}{2}(g^2-c^2)}}{c+\sqrt{c^2+4}} \le d-g \text{ .}
\label{eqn17}
\end{eqnarray}

From (\ref{eqn28}) and (\ref{eqn17}),
\[
d\ge \max\{ \sqrt{\max\{0,c^2+2\log\sin\theta\}}
,g+\frac{e^{\frac{1}{2}(g^2-c^2)}}{c+\sqrt{c^2+4}} \} \text{ .} 
\]

Putting the above in (\ref{eqn13}),
\[
f\ge \max \set{ \frac{\sqrt{\max\{0,c^2+2\log\sin\theta\}}-v\sin\theta}{\cos\theta}
,\frac{g+\frac{e^{\frac{1}{2}(g^2-c^2)}}{c+\sqrt{c^2+4}}-v\sin\theta}{\cos\theta}
} \text{ .} \]

Observe that 
\[
\gamma_1(-\infty,g)=\gamma_2(K)\le \gamma_1(-c,\infty)=\gamma_1(-\infty,c) \text{, }
\]
giving $c\ge g$. Also $\theta\in[0,\pi/2]$. Thus, the above lower bound on $f$ holds if we minimize over all $c\ge g$ and $\theta\in[0,\pi/2]$.
\[
f\ge \min_{c\ge g, \theta\in[0,\pi/2]} \max \set{
\frac{\sqrt{\max\{0,c^2+2\log\sin\theta\}}-v\sin\theta}{\cos\theta}
,\frac{g+\frac{e^{\frac{1}{2}(g^2-c^2)}}{c+\sqrt{c^2+4}}-v\sin\theta}{\cos\theta}
} \text{ .} \]

We will first minimize with respect to $c$. For this, we make the following observations:
\begin{itemize}
\item for a fixed $\theta$, the first term inside the maximum is a non-decreasing function of $c$ while the second is a decreasing function of $c$. 
\item for $c=\sqrt{g^2-2\log\sin\theta} \ge g$, the first term is smaller than the second term
\item  for $c=\sqrt{g^2-2\log\sin\theta+1} \ge g$, the first term is greater than the second term.
\end{itemize}
Thus, the two terms must become equal somewhere in the range 
\[c\in[\sqrt{g^2-2\log\sin\theta},\sqrt{g^2-2\log\sin\theta+1}] \text{ .} \]
In particular, substituting $c=\sqrt{g^2-2\log\sin\theta+1}$ in the second term provides a lower bound for $f$:
\begin{eqnarray}
f &\ge &  \min_{\theta\in[0,\pi/2]}\frac{g+\frac{\sin\theta}{\sqrt{e}(\sqrt{g^2-2\log\sin\theta+1}+\sqrt{g^2-2\log\sin\theta+5})}-v\sin\theta}
{\cos\theta} \nonumber\\
& \ge &
\min_{\theta\in[0,\pi/2]}\frac{g+\frac{\sin\theta}{2\sqrt{e}\sqrt{g^2-2\log\sin\theta+5}}-v\sin\theta}{\cos\theta}\nonumber
\text{ .}
\label{eqn20}
\end{eqnarray}

This expression goes to $g$ as $\theta\rightarrow 0$ and to $\infty$ as $\theta\rightarrow \pi/2$. If it is increasing in this whole interval, we are already done. Else, it achieves its minimum somewhere in $(0,\pi/2)$. Let this be at $\theta^*$. Setting the derivative to zero, we get
\begin{eqnarray}
f & \ge & g\cos\theta^*-\frac{\sin2\theta^*}{4\sqrt{e}\sqrt{g^2-2\log\sin\theta^*+5}^3} \nonumber \\
& = &
g-2g\sin^2(\theta^*/2)-\frac{\sin2\theta^*}{4\sqrt{e}\sqrt{g^2-2\log\sin\theta^*+5}^3}
\text{ ,}
\label{eqn21}
\end{eqnarray}

where $\theta^*$ satisfies
\[
v= g\sin\theta^*+\frac{1}{2\sqrt{e}\sqrt{g^2-2\log\sin\theta^*+5}} \text{ .}
\]

From the two terms above, we can get two upper bounds on $\sin\theta^*$:
\begin{eqnarray}
\sin\theta^*& \le &v/g  \nonumber \text{ ,}\\
\sin\theta^* &\le& \frac{e^{g^2/2+5/2}}{e^{\frac{1}{8ev^2}}} \nonumber \text{ .}
\label{eqn23}
\end{eqnarray}

Using these, we can simplify (\ref{eqn21}) as
\begin{eqnarray}
f &\ge & g-2g\sin^2(\theta^*/2)-2e\sin2\theta^*(v-g\sin\theta^*)^3 \nonumber \\
& \ge & g-2g\sin^2(\theta^*/2)-2ev\sin2\theta^*  \nonumber \\
& \ge & g-2g\sin^2(\theta^*)-4ev\sin\theta^*  \nonumber \text{ .}
\label{eqn24}
\end{eqnarray}

We derive two bounds on the above expression, one which will be useful when $g$
is small and other when $g$ is large. For the small $g$ bound, using that
$v<v_0$ for $v_0$ small enough, 
\[
2g\sin^2(\theta^*)+4ev\sin\theta^* \le
(2g\frac{v}{g}+4ev)\frac{e^{g^2/2+5/2}}{e^{\frac{1}{8ev^2}}}\le\frac{e^{g^2/2}}{e^{\frac{1}{8ev^2}}}
\text{ .} \]

For the large $g$ bound, 
\[
2g\sin^2(\theta^*)+4ev\sin\theta^* \le 2g\frac{v^2}{g^2}+\frac{4ev^2}{g} =
(4e+2)\frac{v^2}{g} \text{ .} 
\]

Thus,
\[
f\ge g-\min\{e^{\frac{g^2}{2}-\frac{1}{8ev^2}},(4e+2)\frac{v^2}{g}\}=f^* \text{ , as needed.}
\]
\end{proof}

We now prove Theorem~\ref{thm:estneg} in the special case where the barycenter lies
to the right of the hyperplane $\theta^\perp$. We show later how to reduce
Theorem~\ref{thm:estneg} to this case.

\begin{lemma}
\label{lem:est}
There exists universal constants $v_0,c_0 > 0$, such that for any $n \geq 1$, $v
\in [0,v_0]$ and $\theta \in \R^n$, $\|\theta\|_2=1$, convex body $K \subseteq \R^n$ satisfying
$\gamma_n(K) = \alpha \geq 1/2$ and $\pr{b(K)}{\theta} \geq 0$, we have that 
\[
\gamma_{n-1}(K_v^\theta) \geq \alpha(1-\frac{e^{-\frac{1}{100v^2}}}{4\sqrt{2\pi}}) \text{.}
\]
\end{lemma}

\begin{proof}
We split the proof into two steps.  In step one, we reduce to a 2-dimensional
problem and show that it suffices to prove our theorem for a downwards closed
convex body $K^\prime \subseteq \R^2$.  This reduction will guaranteee that
$K^\prime$ has barycenter on the y-axis and that the Gaussian measure of slices
of $K^\prime$ parallel to the y-axis will correspond in the natural way to that
of slices of $K$ parallel to the hyperplane $\theta^\perp = \set{x \in \R^n:
\pr{x}{\theta}=0}$. We then invoke Lemma \ref{lem:estim} to get a lower bound on
the height of $K^\prime$ at $x=v$. Lastly, in step 2, we show that implies the
required lower bound on the slice measure.

Let $g$ be s.t. $\gamma_1(-\infty,g)=\alpha$, i.e. $g=\Phi^{-1}(\alpha)$. Note
that $g\ge 0$ since $\alpha\ge1/2$.

\paragraph*{{\bf Step 1: reduction to a 2-dimensional case}}

We will reduce our problem to one for a $2$-dimensional downwards closed convex
body $K^\prime$. To specify $K^\prime$, we need only specify the height of the
boundary at each x-coordinate. At x-coordinate $t$, we define the height of
$K^\prime$ to be the $y_t$ satisfying $\gamma_1(-\infty, y_t) =
\gamma_{n-1}(K^{\theta}_t)$. From Ehrhard's inequality, we see that $K^\prime$
is in fact convex. Furthermore, it is easy to check that
$\gamma_2(K^\prime)=\gamma_n(K)$ and $b(K^\prime)_1 = \pr{b(K)}{\theta} \geq 0$. 

Thus, $K^\prime$ is a downwards closed convex body in $\R^2$ with
$\gamma_2(K^\prime)=\alpha, b(K^\prime)_1 \geq 0$. From here, we may invoke
Lemma~\ref{lem:estim} to conclude that the height of $K^\prime$ at $x=v$ is at
least $f^* := f(v,g)$. We now have that
\[
\gamma_{n-1}(K^\theta_t) = \gamma_1(-\infty,y_t) \geq \gamma_1(-\infty,f^*) \text{ .}
\]
From the above, it suffices to give a lower bound on $\gamma_1(-\infty,f^*)$ in order to derive the theorem.

\paragraph*{\bf Step 2: Bounding $\gamma_1(-\infty, f^*)$}

Our goal is to show that $\gamma_1(-\infty, f^*) \ge
\alpha(1-\frac{e^{-\frac{1}{100v^2}}}{4\sqrt{2\pi}})$. Clearly, it suffices to
show $\gamma_1(f^*,g) \le \alpha\frac{e^{-\frac{1}{100v^2}}}{4\sqrt{2\pi}}$. Let
$\epsilon_g=g-f^*=\min
\set{e^{\frac{g^2}{2}-\frac{1}{8ev^2}},(4e+2)\frac{v^2}{g}}$. We split the
analysis in two cases depending on whether $g$ is small or big.

\paragraph*{\bf Step 2a: $g\le\frac{1}{5v}$}
\[
\gamma_1(f^*,g) \le \frac{\epsilon_g}{\sqrt{2\pi}}
                \le \frac{e^{\frac{g^2}{2}-\frac{1}{8ev^2}}  }{\sqrt{2\pi}} 
                \le \frac{e^{\frac{1}{50v^2}-\frac{1}{8ev^2}}  }{\sqrt{2\pi}} 
                \le \frac{e^{-\frac{1}{100v^2}}}{8\sqrt{2\pi}} 
                \le\alpha \frac{e^{-\frac{1}{100v^2}}}{4\sqrt{2\pi}} \text{ .} 
\]
The penultimate inequality holds for an appropriate choice of $v_0$, and the last inequality uses $\alpha\ge 1/2$. 

\paragraph*{\bf Step 2b: $g>\frac{1}{5v}$}

Here we will use the other bound for $\epsilon_g$. 
\begin{eqnarray}
\gamma_1(f^*,g) &\le& \frac{\epsilon_g}{\sqrt{2\pi}}e^{-(f^*)^2/2} 
=  \frac{\epsilon_g}{\sqrt{2\pi}}e^{-g^2/2}e^{g\epsilon_g-\epsilon_g^2/2} \nonumber \\ 
&\le & \frac{\epsilon_g}{\sqrt{2\pi}}e^{-g^2/2+g\epsilon_g}   
 \le  \frac{(4e+2)v^2}{g\sqrt{2\pi}}e^{-g^2/2+(4e+2)v^2}   \nonumber \\
& \le & \frac{5(4e+2)v^3}{\sqrt{2\pi}}e^{-\frac{1}{50v^2}+(4e+2)v^2}  \le
\frac{e^{-\frac{1}{100v^2}}}{8\sqrt{2\pi}} \le\alpha
\frac{e^{-\frac{1}{100v^2}}}{4\sqrt{2\pi}} \nonumber  \text{ .}
\label{eqn71}
\end{eqnarray}
The penultimate inequality holds for an appropriate choice of $v_0$, and the last
inequality uses $\alpha\ge 1/2$.
\end{proof}

We now come to the proof of Theorem~\ref{thm:estneg}.\\

\noindent{\bf Theorem \ref{thm:estneg} (restated):} {\em
There exist universal constants $v_0,\eta_0,c_0 > 0$, such that for any $n \geq
1$, convex body $K \subseteq \R^n$ satisfying $\|b(K)\|_2 = \eta \leq \eta_0$
and $\gamma_n(K) = \alpha \geq 3/5$, $v \in [-v_0,v_0]$ and $\theta \in \R^n$,
$\|\theta\|_2=1$, we have that 
\[
\gamma_{n-1}(K_v^\theta) \geq (\alpha-c_0\eta)(1-\frac{e^{-\frac{1}{100v^2}}}{4\sqrt{2\pi}}) \text{.}
\]
}
\begin{proof}
By rotational invariance, we may assume that $\theta = e_1$, the first standard
basis vector. By possibly replacing $K$ by $-K$, we may also assume that $v \geq 0$.

If $b(K)_1 \geq 0$, the desired lower bound follows directly from
Lemma~\ref{lem:est}. Given this, we may assume that $-\eta \leq b(K)_1
< 0$. To deal with this second case, the main idea is to remove some
portion of $K$ lying to the left of the hyperplane $e_1^\perp = \set{x
  \in \R^n: x_1=0}$ so that the barycenter of the remaining body
lies on $e_1^\perp$. After this, we apply Lemma~\ref{lem:est} again on
the truncated body.

Define
\[
b := b(K)_1 \gamma_n(K) 
           = \int_K x_1 \frac{e^{-\frac{1}{2}\|x\|^2_2}}{\sqrt{2\pi}^n} {\rm dx} < 0 \text{ .}
\]
Let $H^-_t = \set{x \in \R^n: x_1 \leq t}$ and $H^+_t = \set{x \in \R^n: x_1
\geq t}$ for $t \in \R$. Let $z < 0$ be defined as the smallest negative number satisfying 
\begin{equation}
\label{eq:neg-1}
\int_{K \cap H^+_z} x_1 \frac{e^{-\frac{1}{2}\|x\|_2^2}}{\sqrt{2\pi}^n} {\rm dx} = 0 \text{ .}
\end{equation}
By continuity, such a $z$ must exists, since as $z \rightarrow -\infty$ the left
hand side tends to $b < 0$ and at $z=0$ it is positive. Given the above, note
that $b(K \cap H^+_z)_1 = 0$. We will now show that $\gamma_n(K \cap H^+_z) \geq
\alpha-c_0\eta$ if $\eta \leq \eta_0$, for $c_0,\eta_0$ appropriately chosen
constants.

By our choice of $z$, we have the equality
\[
\int_{K \cap H^-_z} x_1 \frac{e^{-\frac{1}{2}\|x\|_2^2}}{\sqrt{2\pi}^n} {\rm dx} = b \text{ .}
\]
From here, we see that
\[
\gamma_n(K \cap H^-_z) = \int_{K \cap H^-_z} \frac{e^{-\frac{1}{2}\|x\|^2_2}}{\sqrt{2\pi}^n} {\rm dx} \le
\int_{K \cap H^-_z} \frac{x_1}{z}
\frac{e^{-\frac{1}{2}\|x\|^2_2}}{\sqrt{2\pi}^n} {\rm dx} = \frac{b}{z} =
\left|\frac{b}{z}\right| \text{ .}
\]
Given this, we get that
\[
\gamma_n(K \cap H^+_z)=\gamma_n(K)-\gamma_n(K \cap H^-_z) = \alpha-\gamma_n(K
\cap H^-_z) \ge
\alpha-\left|\frac{b}{z}\right| \geq \alpha-\left|\frac{\eta}{z}\right| \text{ .}
\]

We now show that there exists a constant $c_0$ s.t. $1/|z| \leq c_0$. Let $\beta =
\gamma_n(K \cap H^+_0)$, and note that
\[
\beta = \gamma_n(K) - \gamma_n(K \cap H^-_0) \geq \alpha-1/2 \geq 3/5-1/2 = 1/10 \text{ .}
\]
Let $\tau > 0$ be positive number satisfying $\gamma_1(0,\tau) = \beta$,
i.e.~$\tau = \Phi^{-1}(1/2+\beta)$. By pushing the mass of $K \cap
H^+_0$ to the left towards $e_1^\perp$ as much as possible, we see that
\begin{equation}
\label{eq:neg-2}
\int_{K \cap H^+_0} x_1 \frac{e^{-\frac{1}{2}\|x\|_2^2}}{\sqrt{2\pi}^n} {\rm dx}
\geq \int_{\set{x \in \R^n: 0 \leq x_1 \leq \tau}} x_1
\frac{e^{-\frac{1}{2}\|x\|_2^2}}{\sqrt{2\pi}^n} {\rm dx} =
\frac{1}{\sqrt{2\pi}}(1-e^{-\tau^2/2}) \text{ .}
\end{equation}
Next, by inclusion
\begin{equation}
\label{eq:neg-3}
\int_{K \cap \set{x \in \R^n: z \leq x \leq 0}} x_1
\frac{e^{-\frac{1}{2}\|x\|_2^2}}{\sqrt{2\pi}^n} {\rm dx} \geq 
\int_{\set{x \in \R^n: z \leq x \leq 0}} x_1
\frac{e^{-\frac{1}{2}\|x\|_2^2}}{\sqrt{2\pi}^n} {\rm dx} =
\frac{1}{\sqrt{2\pi}}(e^{-z^2/2}-1) \text{ .}
\end{equation}
Given that $z$ satisfies \eqref{eq:neg-1}, combining
equations \eqref{eq:neg-2},~\eqref{eq:neg-3}, we must have that
\[
0 \geq e^{-z^2/2}-e^{-\tau^2/2} \Rightarrow |z| \geq \tau \geq \Phi^{-1}(6/10) > 0 \text{ .}
\]
Thus, we may set $c_0 = 1/\Phi^{-1}(6/10)$. Set $\eta_0 = \frac{1}{10c_0}$.
Since $\eta \leq \eta_0$, we have that
\[
\gamma_n(K \cap H^+_z) \geq \alpha - c_0 \eta \geq 3/5 - 1/10 = 1/2 \text{ .}
\] 
Lastly, using Lemma~\ref{lem:est} on $K \cap H^+_z$, we now get that
\[
\gamma_{n-1}(K^{e_1}_v) = \gamma_{n-1}((K \cap H^+_z)^{e_1}_v) \geq 
(\alpha-c_0\eta)(1-\frac{e^{-\frac{1}{100v^2}}}{4\sqrt{2\pi}}) \text{ ,}
\]
as needed.
\end{proof}

\section{Constructive Vector Koml{\'o}s}
\label{sec:vector-komlos}

In this section we give a new proof of the main result of
\cite{Nikolov13} that the natural SDP for the Koml\'os problem has
value at most $1$. While the proof in \cite{Nikolov13}
used duality, our proof is direct and immediately yields an algorithm
to compute an SDP solution which only uses basic linear algebraic
operations, and does not need a general SDP solver. We state the main
theorem next. 

\begin{theorem}\label{thm:SDP-feas}
  Let $v_1, \ldots, v_n \in \R^m$ be vectors of Euclidean length at
  most $1$, and let $\alpha_1, \ldots, \alpha_n \in [0,1]$. There
  exists an $n \times n$ PSD matrix $X$ such that
  \begin{align*}
    X_{ii} &= \alpha_i \ \ \ \forall 1 \le i \le n\\
    VXV^\T &\preceq I_m,
  \end{align*}
  where $V = (v_1, \ldots, v_m)$ is the $n \times m$ matrix whose
  columns are the vectors $v_i$.
\end{theorem}


To prove Theorem~\ref{thm:SDP-feas} we make use of a basic identity
about inverses of block matrices. This is a standard use of the Schur
complement and we will not prove it here.
\begin{lemma}\label{lm:schur}
  Let
  \[
  A = \left(
    \begin{array}{cc}
      A_{11} &A_{12}\\
      A_{21} &A_{22}\\
    \end{array}
    \right)
  \]
  be a $(k + \ell) \times (k + \ell)$ block matrix, where $A_{11}$ is a $k\times
  k$ matrix, $A_{12}$ is a $k\times \ell$ matrix, $A_{21}$ is a $\ell \times k$
  matrix, and $A_{22}$ is a $\ell \times \ell$ matrix. Assume $A$ and $A_{22}$
 are invertible, and write $B = A^{-1}$ in block form as
 \[
  B = \left(
    \begin{array}{cc}
      B_{11} &B_{12}\\
      B_{21} &B_{22}\\
    \end{array}
    \right),
 \]
 where $B_{ij}$ has the same dimensions as $A_{ij}$. Then $B_{11} =
 (A_{11} - A_{12}A_{22}^{-1}A_{21})^{-1}$ (i.e.~the inverse of the
 Schur complement of $A_{11}$ in $A$).
\end{lemma}

From Lemma~\ref{lm:schur} we derive the main technical claim used in
the proof of Theorem~\ref{thm:SDP-feas}.

\begin{lemma}\label{lm:inverse}
  Let $A = V^\T V$ be an $n\times n$ positive definite matrix, and let
  $v_1, \ldots, v_n$ be the columns of $V$. Let $B =
  A^{-1}$. Then, for each $1 \leq i \leq n$ 
  \[
  B_{ii} = \frac{1}{\|(I - \Pi_{-i})v_i\|_2^2} \ge \frac{1}{\|v_i\|_2^2},
  \]
  where $\Pi_{-i}$ is the orthogonal projection matrix onto
  $\mathrm{span}\{v_j: j \neq i\}$.
\end{lemma}
\begin{proof}
  It is sufficient to prove the lemma for $i = 1$. Let $U$ be the
  matrix with columns $v_2, \ldots, v_n$. Since $A$ is positive
  definite, the principal minor $U^\T U$ is positive definite as
  well, and, therefore, invertible. By Lemma~\ref{lm:schur},
  \[
  B_{11} = \frac{1}{\|v_1\|_2^2 - v_1^\T U (U^\T U)^{-1} U^\T v_1}.
  \]
  Let $\Pi = U (U^\T U)^{-1} U^\T$. Since $\Pi$ is symmetric and
  idempotent (i.e.~$\Pi^2 = \Pi$), it is an orthogonal projection
  matrix. Moreover $\Pi U = U$ and $\Pi$ has the same rank as $U$, so
  $\Pi$ is the orthogonal projection matrix onto the column span of
  $U$, i.e.~$U (U^\T U)^{-1} U^\T = \Pi_{-1}$ and the lemma follows.
\end{proof}

\begin{proof}[Proof of Theorem~\ref{thm:SDP-feas}]
  We prove the theorem by induction on $n$.

  In the base case $m=1$, we have a single vector $v \in \R^m$,
  $\|v\|_2 \le 1$, and an $\alpha \in [0,1]$. We set $x = \alpha$, and
  we clearly have $vxv^T \preceq \alpha I \preceq I$.

  We now proceed with the inductive step. 
  Consider first the case that $V^\T V$ is singular. Then there exists a
  vector $x \neq 0$ such that $Vx = 0$. Scale $x$ so that
 $x_i^2 \le \alpha_i$ for all $i$, and there exists $k$ such that $x_k^2 =
  \alpha_k^2$. Apply the inductive hypothesis to the vectors $(v_i: i \neq k)$ and
  the reals $(\alpha'_i = \alpha_i - x_i^2: i \neq k)$ to get a matrix
  $Y \in \R^{([m] \setminus \{k\}) \times ([n] \setminus \{k\})
  }$. Extend $Y$ to a matrix $\tilde{Y} \in \R^{n \times n}$ by padding
  with $0$'s,
  i.e.~$\tilde{Y}_{ij} = Y_{ij}$ if $i,j \neq k$ and $\tilde{Y}_{ij} =
	0$, otherwise. Define $X = xx^\T + \tilde{Y}$: it is easy to
verify that both conditions of the theorem are satisfied.

  Finally, assume that $V^\T V$ is invertible, and let $B =
  (V^\T V)^{-1}$. Define
  \begin{align*}
    \beta &= \min_i \alpha_i/B_{ii},\\
    k &= \arg \min_i \alpha_i/B_{ii},\\
    \gamma &= \max_i \alpha_i - \beta B_{ii}.
  \end{align*}
  Apply the inductive hypothesis to the vectors $(v_i: i \neq k)$ and
  the reals $(\alpha'_i = (\alpha_i - \beta b_{ii})/\gamma: i \neq k)$
  to get a matrix $Y \in \R^{([n] \setminus \{k\}) \times ([n]
    \setminus \{k\}) }$, which we then pad with $0$'s to an $n\times
  n$ matrix $\tilde{Y}$, as we did in the first case above. Define $X$
  as $X = \beta B + \gamma \tilde{Y}$. It is easy to verify that
  $X_{ii} = \alpha_i$ for all $i$. We have
  \[
  VXV^\T = \beta VBV^\T + \gamma V^\T \tilde{Y}V = \beta V(V^\T V)^{-1}V^\T +
\gamma U^\T YU,
  \]
  where $U$ is the submatrix of $V$ consisting of all columns of $V$ except
  $v_k$. $U^\T YU \preceq I$ by the induction hypothesis. Since
  $V(V^\T V)^{-1}V^\T$ is symmetric and idempotent, it is an orthogonal
  projection matrix, and therefore $V(V^\T V)^{-1}V^\T \preceq
  I$. Because $B_{ii} \geq \|v_i\|_2^{-2} \geq 1$ by
  Lemma~\ref{lm:inverse}, we have $\gamma \leq \max_i \alpha_i -
  \beta$. Therefore,
  \[
  VXV^\T = \beta V(V^\T V)^{-1}V^\T + \gamma U^\T YU \preceq (\beta +
  \gamma)I_n \preceq (\max_i \alpha_i)I_n \preceq I_n.
  \]
  This completes the proof.
\end{proof}

Observe that the proof of Theorem~\ref{thm:SDP-feas} can be easily
turned into an efficient recursive algorithm. 


\bibliographystyle{plainurl}
\bibliography{weak-bana}

\section{Appendix}

\subsection{Estimating the Barycenter}
In this section we show how to efficiently estimate the barycenter of $K$ up to a small accuracy in $\ell_2$-norm. For a convex body $K \subseteq \R^n$, we
let $\gamma_K$ denote the Gaussian measure restricted to $K$.
For a
random variable $X$ in $\mathbb{R}^n$, we denote the covariance of $X$
by $cov[X] = \E[(X - \E[X])(X - \E[X])^\T]$.

The following lemma shows that the covariance of a Gaussian random vector
shrinks when restricted to a convex body. We include a short proof for
completeness. 

\begin{lemma}
\label{covarshrink}
Given a convex body $K$ in $\mathbb{R}^n$, let $\gamma_K$ be the
Gaussian distribution restricted to $K$, and let $X$ be a random
variable distributed according to $\gamma_K$. Then, $cov[X] \preceq I_n$.
\end{lemma}
\begin{proof}
  Consider $f(t)=\ln\,\gamma_n(K+t)$. $f$ is concave in $t$. This
  follows from log-concavity of $\gamma_n$, an easy consequence of the
  Prekopa-Leindler inequality. Hence, the Hessian of $f$, $H(f)$, is
  negative semi-definite. It can be calculated that
\[H(f)=H(\ln\,\gamma_n(K+t))=cov[X + t] -I_n,\] where $X \sim \gamma_K$.
Setting $t=0$ completes the proof.
\end{proof}

We will also need to use Paouris' inequality~\cite{Paouris06}, which we restate
slightly:
\begin{theorem}
\label{paris}
If $X\subseteq \mathbb{R}^n$ is a log-concave random vector with mean $0$ and positive-definite covariance matrix $C$, then for every $t\ge 1$,
\[ \Pr[\sqrt{X^\T C^{-1}X}\ge \beta t\sqrt{n}] \le e^{-t\sqrt{n}} \]
where $\beta>0$ is an absolute constant.
\end{theorem}

\begin{theorem}
\label{bary}
Let $K$ be a convex body in $\mathbb{R}^n$, given by a membership oracle, with
$\gamma_n(K)\ge 1/2$. For any $\delta>0$ and $\epsilon\in(0,1)$, there is an algorithm which computes the barycenter
of $K$ within accuracy $\delta$ in $\ell_2$-norm with
probability at least $1-\epsilon$ in time polynomial in $n,1/\delta$ and $\log(1/\epsilon)$.
\end{theorem}
\begin{proof}
Let $b$ be the barycenter of $K$ and $X_i$ for $1\le i\le N$ be i.i.d generated from $\gamma_K$, where $N=\lceil(\beta/\delta)^2\log^2(e/\epsilon)n\rceil$. Here $\beta$ is the constant from Theorem~\ref{paris}. Defining the following quantities 
\begin{eqnarray*}
 b'&=&\frac{1}{N}\sum_{i=1}^N X_i \\ 
Y_i&=&X_i-b\\
 Y&=&\frac{1}{N}\sum_{i=1}^N Y_i\\
 C&=&\E_{X\sim\gamma_K}[(X-b)(X-b)^\T]
\end{eqnarray*}
we can see that $b'$ is an estimate of the barycenter, generated by averaging random samples from the distribution $\gamma_K$ and $Y$ is the difference vector between the true barycenter and $b'$. Thus it suffices to bound the probability that $Y$ is large and then show how to efficiently generate random samples from the distribution $\gamma_K$. It holds that
\[ \E[Y_i]=\E[X_i-b]=b-b=0 \]
Also, using Lemma~\ref{covarshrink},
\[  \E[Y_iY_i^\T]= cov[X_i]=C \preceq I_n \]
Thus, 
\[\E[Y]=0 \textrm{ and } cov[Y]=C/N \preceq I_n/N\]
Since $\gamma_K$ is a log-concave distribution, $X_i$ and hence $Y_i$
are log-concave random vectors. It is easily checked (using the
Prekopa-Leindler inequality) that the average of log-concave random variables is also log-concave and hence $Y$ is a log-concave random vector. Now, 
\begin{eqnarray*}
 \Pr[\| Y\|_2\ge\delta] &=& \Pr\left[\sqrt{Y^\T\left(\frac{I_n}{N}\right)^{-1}Y}\ge \delta\sqrt{N}\right] \\
 & \le & \Pr\left[\sqrt{Y^\T\left(\frac{C}{N}\right)^{-1}Y}\ge \delta\sqrt{N}\right] \qquad \textrm{(using $C/N \preceq I_n/N$)}
\end{eqnarray*}
Putting $N=\lceil(\beta/\delta)^2\log^2(e/\epsilon)n\rceil$ and using Theorem~\ref{paris} with $t=\log(e/\epsilon)$, we get
\begin{equation}
 \Pr[\| Y\|_2 \ge \delta] \le e^{-\sqrt{n}\log(e/\epsilon)}\le \epsilon/e \le \epsilon/2. \label{probabbary}
\end{equation}

We can generate the random points $X_i$ using rejection sampling. For
each $i$, we generate a sequence of i.i.d.~standard Gaussian random variables
$X_i^{(1)}, \ldots, X_i^{(k)} \in \R^n$, $k = \lceil
\log_2(2N/\epsilon)\rceil$. We set $X_i$ to the first $X_i^{(j)}$ in
the sequence that belongs to $K$; if no such $X_i^{(j)}$ exists, we set $X_i$
arbitrarily. Clearly, conditional on the existence of a $j \le k$ such
that $X_i^{(j)} \in K$, $X_i \sim \gamma_K$. Furthermore, because $K$
has Gaussian measure at least $1/2$, for every $j$ we have
$\Pr[X_i^{(j)} \not \in K] \le 1/2$, so \[\Pr[\forall j: X_i^{(j)} \not
\in K] \le 2^{-k} \le \epsilon/2N.\] By a union bound, with
probability at least $1 - \epsilon/2$, all $X_i$ are distributed
according to $\gamma_K$; let us call this event $E$. Conditional on $E$, inequality
\eqref{probabbary} holds, and
\begin{align*}
\Pr[\|b - b'\|_2 \ge \delta] &= \Pr[\|Y\|_2 \ge \delta] \\
&= \Pr[\|Y\|_2 \ge \delta \mid E]\cdot \Pr[E] + \Pr[\|Y\|_2 \ge \delta \mid
E^c] (1-\Pr[E]) \\
&\le \epsilon/2 + \epsilon/2 = \epsilon.
\end{align*}
The algorithm needs to generate $O(N\log(N/\epsilon))$ $d$-dimensional
Gaussian random variables, check membership in $K$ for each of them,
and compute the average of $N$  points. Since each of
these operations takes polynomial time, and $N$ is polynomial in $n$,
$\delta$, and $\log(1/\epsilon)$, the running time of the 
algorithm is polynomial.

\end{proof}

\subsection{Proof of Lemma~\ref{lem:mom-to-subg}}

\begin{proof}
We first prove subgaussianity. Let $\theta \in S^{n-1}$, $t \geq 0$. By
assumption on $X$,
\begin{align*}
\Pr[|\pr{X}{\theta}| \geq t] 
  &= \min_{\lambda > 0} \Pr[\cosh(\lambda \pr{X}{\theta}) \geq \cosh(\lambda t)] 
  \leq \min_{\lambda > 0} \beta \cdot \frac{e^{\sigma^2 \lambda^2/2}}{\cosh(\lambda t)} \\
  &\leq \min_{\lambda > 0} 2\beta \cdot e^{\sigma^2 \lambda^2/2-\lambda t} \leq
2 \beta \cdot e^{-\frac{1}{2}(t/\sigma)^2} \text{, }
\end{align*}
where the last inequality follows by setting $\lambda = t/\sigma^2$.

Let $\alpha = \sqrt{\log_2 \beta + 1}$. To prove that $X$ is $\alpha
\sigma$-subgaussian, since probabilities are always at most one, it suffices to
prove that
\[
\min \set{1, 2\beta \cdot e^{-\frac{1}{2}(t/\sigma)^2}} \leq
2 e^{-\frac{1}{2}(t/(\alpha \sigma))^2}, \forall t \geq 0 \text{.}
\]
Replacing $t \leftarrow \sqrt{2} \sigma t$, the above simplifies to showing
\begin{equation}
\label{eq:simp-subg-trans1}
\min \set{1, 2\beta \cdot e^{-t^2}} \leq 2 e^{-(t/\alpha)^2}, \forall t \geq 0 \text{.}
\end{equation}
From here, we see that
\begin{equation}
\label{eq:simp-subg-trans2}
\beta \cdot e^{-t^2} \leq e^{-(t/\alpha)^2} \Leftrightarrow
\beta \leq e^{(1-1/\alpha^2)t^2} \Leftrightarrow t \geq \sqrt{\ln \beta
\cdot \frac{\alpha^2}{\alpha^2-1}} = \sqrt{\ln(2\beta) }  \text{.}
\end{equation}
Let $r = \sqrt{\ln(2\beta)}$, noting that $1 = 2\beta \cdot e^{-r^2} = 2
e^{-(r/\alpha)^2}$, we have that for $t \leq r$, the LHS
of~\ref{eq:simp-subg-trans1} is $1$ and the RHS is at least $1$, for $t > r$,
the LHS is equal to $2\beta \cdot e^{-t^2}$ and the RHS is larger
by~\ref{eq:simp-subg-trans2}. Thus, $X$ is $\alpha \sigma$-subgaussian as needed. 

We now prove the furthermore. For $X$ an $n$-dimensional standard Gaussian, note
that $\pr{X}{w}$ is distributed like $\sigma Y$, where $Y \sim N(0,1)$ and
$\sigma = \|\vec{w}\|_2$. Hence,
\begin{align*}
\E[e^{\pr{X}{w}}] &= \E[e^{\sigma Y}] 
 = \frac{1}{\sqrt{2\pi}} \int_{-\infty}^\infty e^{\sigma x} e^{-x^2/2} dx \\
&= e^{\sigma^2/2} \left(\frac{1}{\sqrt{2\pi}} \int_{-\infty}^\infty
e^{-(x-\sigma)^2/2} dx\right) = e^{\sigma^2/2} \text{, }
\end{align*}
as needed.
\end{proof}

\end{document}